\documentclass[12pt,dvipsnames]{article}
\usepackage[top=80pt,bottom=85pt,left=85pt,right=85pt]{geometry}

\pdfoutput=1
\usepackage[utf8]{inputenc}
\usepackage[T1]{fontenc} 
\usepackage[english]{babel} 
\usepackage{braket}
\usepackage[dvipsnames]{xcolor} 
\usepackage{physics}
\usepackage{tabularx} 
\usepackage{ragged2e}
\usepackage{booktabs} 
\usepackage{rotating}
\usepackage{makecell,booktabs}
\usepackage{multirow} 
\usepackage{comment} 
\usepackage{amsmath}
\usepackage{cite}
\usepackage[all]{xy}

\usepackage{tikz}
\usetikzlibrary{math} 
\usetikzlibrary{3d} 
\usepackage{tikz-3dplot}
\usepackage{tikz-cd} 
\usepackage{microtype}
\usepackage{amsthm}
\usepackage{mathrsfs} 
\usepackage[strict]{changepage}
\newcommand\scalemath[2]{\scalebox{#1}{\mbox{\ensuremath{\displaystyle #2}}}}

\tikzcdset{
  cells={font=\everymath\expandafter{\the\everymath\displaystyle}},
}

\usepackage{amssymb} 
\usepackage{subcaption} 
\DeclareCaptionFormat{custom}

\usepackage[pdftex,colorlinks=true]{hyperref} 
\hypersetup{urlcolor=MidnightBlue, citecolor=red, linkcolor=MidnightBlue}

\usepackage{float} 
\usepackage{todonotes}

\usepackage[capitalise]{cleveref}

\newcolumntype{E}{>{\hfil$}p{0.65cm}<{$\hfil}}

\newcolumntype{L}{>{\hfil$}p{16cm}<{$\hfil}}

\newcolumntype{D}{>{\hfil$}p{7.4cm}<{$\hfil}}
\newcolumntype{C}{>{\hfil$}p{3cm}<{$\hfil}}
\newcolumntype{P}{>{\hfil$}p{7.7cm}<{$\hfil}}
\newcolumntype{F}{>{\hfil$}p{5.7cm}<{$\hfil}}
\newcolumntype{S}{>{\hfil$}p{1.8cm}<{$\hfil}}
\newcolumntype{R}{>{\hfil$}p{5.2cm}<{$\hfil}}
\newcolumntype{U}{>{\hfil$}p{4.2cm}<{$\hfil}}
\newcolumntype{Q}{>{\hfil$}p{6.4cm}<{$\hfil}}
\newcolumntype{T}{>{\hfil$}p{1.9cm}<{$\hfil}}
\newcolumntype{V}{>{\hfil$}p{5.8cm}<{$\hfil}}
\newcolumntype{H}{>{\hfil$}p{1.8cm}<{$\hfil}}
\newcolumntype{A}{>{\hfil$}p{6cm}<{$\hfil}}
\newcolumntype{B}{>{\hfil$}p{2cm}<{$\hfil}}

\makeatletter
\newcommand\xleftrightarrow[2][]{%
  \ext@arrow 9999{\longleftrightarrowfill@}{#1}{#2}}
\newcommand\longleftrightarrowfill@{%
  \arrowfill@\leftarrow\relbar\rightarrow}
\makeatother

\newcommand{\Calo}{\mathcal{O}}

\newcommand{\C}{\mathbb{C}} 
\newcommand{\R}{\mathbb{R}}
\newcommand{\Z}{\mathbb{Z}}

\newcommand{\Pj}{\mathbb{P}}

\DeclareMathOperator{\Spec}{Spec}

\newtheorem{theorem}{Theorem}[section]
\newtheorem{proposition}[theorem]{Proposition}

\newtheorem{definition}[theorem]{Definition}

\newtheorem{example}[theorem]{Example} 

\newcommand{\re}[1]{\textcolor{red}{#1}}

\numberwithin{equation}{section}

\definecolor{cambridgeblue}{rgb}{0.64, 0.76, 0.68}
\definecolor{caribbeangreen}{rgb}{0.0, 0.8, 0.6}
\definecolor{celadon}{rgb}{0.67, 0.88, 0.69}
\definecolor{champagne}{rgb}{0.97, 0.91, 0.81}
\definecolor{cream}{rgb}{1.0, 0.99, 0.82}
\definecolor{cyan(process)}{rgb}{0.0, 0.72, 0.92}
\definecolor{brilliantlavender}{rgb}{0.96, 0.73, 1.0}
\definecolor{candypink}{rgb}{0.89, 0.44, 0.48}

\begin{document}

\begin{titlepage}

\phantom{wowiezowie}

\vspace{-1cm}

\begin{center}

{\Huge {\bf 5d Conformal Matter}}

\vspace{1cm}

{\Large  Mario De Marco$^{\dagger}$, Michele Del Zotto$^{\dagger,\ddagger}$},\\ 

\medskip

{\Large  Michele Graffeo$^{\sharp}$, and Andrea Sangiovanni$^{*}$}\\

\vspace{1cm}

{\it
{\small

$^\dagger$ Mathematics Institute, Uppsala University, \\ Box 480, SE-75106 Uppsala, Sweden\\
\vspace{.25cm}
$^\ddagger$ Department of Physics and Astronomy, Uppsala University,\\ Box 516, SE-75120 Uppsala, Sweden\\
\vspace{.25cm}
$^\sharp$ Department of Mathematics, Politecnico di Milano,\\ Via Bonardi 9, Milano 20133, Italy\\
\vspace{.25cm}
$^*$ Department of Physics, King's College London,\\ Strand, WC2R 2LS London, United Kingdom\\
}}

\vskip .5cm
{\footnotesize \tt mario.demarco@math.uu.se \hspace{1cm} michele.delzotto@math.uu.se } \\
{\footnotesize \tt    michele.graffeo@polimi.it \hspace{1cm} andrea.1.sangiovanni@kcl.ac.uk}

\vskip 1cm
     	{\bf Abstract }
\vskip .1in
\end{center}

\noindent Six-dimensional superconformal field theories (SCFTs) have an atomic classification in terms of elementary building blocks, conformal systems that generalize matter and can be fused together to form all known 6d SCFTs in terms of generalized 6d quivers. It is therefore natural to ask whether 5d SCFTs can be organized in a similar manner, as the outcome of fusions of certain elementary building blocks, which we call 5d conformal matter theories. In this project we begin exploring this idea and we give a systematic construction of 5d generalized ``bifundamental'' SCFTs, building from geometric engineering techniques in M-theory. In particular, we find several examples of $(\mathfrak {e}_6,\mathfrak {e}_6)$, $(\mathfrak {e}_7,\mathfrak {e}_7)$ and $(\mathfrak {e}_8,\mathfrak {e}_8)$ 5d bifundamental SCFTs beyond the ones arising from (elementary) KK reductions of the 6d conformal matter theories. We show that these can be fused together giving rise to 5d SCFTs captured by 5d generalized linear quivers with exceptional gauge groups as nodes, and links given by 5d conformal matter. As a first application of these models we uncover a large class of novel 5d dualites, that generalize the well-known fiber/base dualities outside the toric realm. 
 
\eject

\end{titlepage}

\tableofcontents

\section{Introduction}
The existence of superconformal field theories (SCFTs) in dimension higher than four is among the striking consequences of string theory \cite{Witten:1995ex,Strominger:1995ac,Witten:1995em,Ganor:1996mu,Seiberg:1996qx,Seiberg1996,Morrison_1997,Douglas:1996xp}. Thanks to higher dimensional SCFTs several very non-trivial aspects of lower dimensional QFTs can be predicted via compactifications, ranging from dualities \cite{Ganor:1996xd,Ganor:1996pc,Gaiotto:2009we,Gaiotto:2009hg,Benini:2009mz,Ohmori:2015pua,DelZotto:2015rca,Ohmori:2015pia,Bhardwaj:2019fzv,Sacchi:2021wvg} and non-invertible symmetries \cite{Gukov:2020btk,Bashmakov:2022uek,Bashmakov:2022jtl,Carta:2023bqn,Bashmakov:2023kwo,Chen:2023qnv} to correspondences between QFTs in different dimensions \cite{Alday:2009aq,Wyllard:2009hg,Nekrasov:2009rc,Dimofte:2011ju,Gadde:2013sca,Cordova:2013cea,Cordova:2016cmu,Dedushenko:2017tdw,Dedushenko:2019mnd} and beyond \cite{Cecotti:2011iy}. 
Higher dimensional irreducible SCFTs are rigid and extremely constrained \cite{Nahm:1977tg,Minwalla:1997ka,Cordova:2016emh}: this suggests looking for classification schemes as a pathway towards a better understanding of these systems, by highlighting generic features and identifying exotics.


There are various possible classification schemes which we broadly divide in three categories: geometric classification schemes, algebraic classification schemes, and atomic classifications. Geometric classification schemes rely upon features such as the constrained structure of moduli spaces of superconformal systems, or the constraints that superconformal symmetry imposes on the geometrical realization of these models within string theory, or a combination of the two. Algebraic classification schemes build on the constraints arising from representations of superconformal algebras, as well as on the existence of generalized versions of chiral rings. The latter further constrain certain BPS subsectors of the spectrum of operators of SCFTs and have an interesting interplay with bootstrap techniques. Geometric and algebraic classifications can sometimes be combined, which gives interesting consistency checks.\\
\indent Atomic classifications are slightly different in spirit. An atomic classification builds on identifying some basic more elementary building blocks for SCFTs (the atoms, or ``conformal matter'' systems), and constraining the structures of the other SCFTs viewing them as molecules, built out of atoms via field theoretical operations analogue to the conformal gauging of diagonal global symmetries in four dimensions. These operations are typically very constrained by the lack of gauge anomalies and by analogues of the condition of vanishing beta functions in four dimensions.

In the context of 5d SCFTs, there are well-developed geometric classification schemes building on M-theory geometries \cite{Jefferson:2017ahm,Jefferson:2018irk,Bhardwaj:2018vuu,Apruzzi:2019opn,Bhardwaj:2019jtr,Bhardwaj:2020gyu,Xie:2022lcm}, while both an algebraic classification and an atomic classification are lacking at the time of this writing. In this paper we aim at laying the foundations for an atomic classification scheme for 5d SCFTs. Our main result in this paper is the construction of 5d ``bifundamental'' conformal matter theories for gauge groups of type $A_k$, $D_k$, $E_6,E_7$ and $E_8$, as well as the prediction of new 5d dualities associated to 5d linear quivers built out of conformal matter theories. The existence of 5d conformal matter theories of bifundamental kind is the starting point for an atomic classification program, but it is not the end of it. In particular, in 5d we expect to find other 5d conformal matter systems that are of ``trifundamental'' kind. These 5d trifundamentals are the subject of an upcoming work \cite{DMDZGSinprep}. In this introduction we proceed by explaining this logic in detail by contrasting geometric and atomic classifications in the context of 6d theories \cite{Heckman:2013pva,Heckman:2015bfa,Bhardwaj:2019hhd}, as well as for 4d SCFTs. We then proceed outlining some of the features of the 5d atomic classification.

\subsection{Geometric and atomic classification schemes} In order to illustrate the main result in this project, let us proceed by quickly reviewing the existing classifications of six-dimensional SCFTs. There are only two 6d superconformal algebras that admit a conserved stress-energy tensor supermultiplet \cite{Cordova:2016xhm}, $\mathfrak{osp}(8^*|2)$ and $\mathfrak{osp}(8^*|4)$.\footnote{ \ More exotic systems with $\mathfrak{osp}(8^*|8)$ superconformal symmetry do exist within string theory \cite{Hull:2000zn,Hull:2000rr}, but these systems lack a conserved stress-energy tensor multiplet and therefore we would prefer not to identify them as quantum field theories. An indication that this is indeed the case, is that the circle reduction of these theories is believed to give rise to $\mathcal N=4$ supergravity in five dimensions \cite{Hull:2000zn}.} Thanks to the enhanced supersymmetry, in this context one can envision geometric classification schemes via geometric engineering dictionaries within string/M/F theory. For $\mathfrak{osp}(8^*|4)$ this leads to the 6d $(2,0)$ SCFTs, which occur in an ADE series corresponding to the geometric engineering limit of IIB superstrings on orbifolds $\mathbb C^2/\Gamma$ where $\Gamma \subseteq SU(2)$ is a finite subgroup \cite{Witten:1995ex,Strominger:1995ac}.\footnote{ \ The latter can be justified with field theoretical methods, either starting from the structure of BPS strings \cite{Henningson:2004dh}, or adding some further requirements on the axioms for 6d SCFTs \cite{Cordova:2015vwa}.} In this case, there is also a possible algebraic classification, arising from VOAs organizing specific 6d (2,0) chiral rings and its interplay with the bootstrap \cite{Beem:2014kka,Beem:2015aoa}.

For theories with $\mathfrak{osp}(8^*|2)$ superconformal symmetry, there is no known algebraic classification and the geometric classification scheme becomes much richer and intricate. These 6d (1,0) SCFTs are now realized within F-theory via singular elliptic Calabi-Yau (CY) threefolds with orbifold bases of the form $\mathbb C^2/\Gamma$, where $\Gamma \subseteq U(2)$  \cite{Heckman:2013pva}.\footnote{ \ See also eg. \cite{DelZotto:2014fia,DelZotto:2015rca,Morrison:2016nrt,DelZotto:2018tcj} for more details about the F-theory orbifold approach to 6d SCFTs.} On top of the data of $\Gamma$, in order to have a consistent F-theory model over such orbifold base, one needs a non-Kodaira singularity in the $T^2$ fiber supported at the origin of $\mathbb C^2/\Gamma$, so that the total space of the fibration is a Calabi-Yau. This can be supplemented with further ``freezing'' data \cite{Bhardwaj:2018jgp}, and possible non-compact Kodaira singularities supported along lines meeting at the origin of the base, that can be further decorated by T-branes \cite{DelZotto:2014hpa}.\footnote{ \ For non-frozen geometries it has been conjectured that the T-brane data give rise to systems that also have a geometric realization without T-branes \cite{Heckman:2015bfa}. Evidence towards this conjecture has been given exploiting systems that have a dual realization in massive Type IIA superstrings \cite{Heckman:2016ssk}.}

For 6d $(1,0)$ SCFTs there is however an alternative classification scheme: the \textit{atomic classification scheme} \cite{Heckman:2015bfa} where generic 6d SCFTs are constructed as molecules built out of more elementary building blocks or \textit{atoms}. These building blocks are conformal systems that play the role of generalized forms of matter, and can be coupled together to form generalized 6d quivers that give a description of more complicated SCFTs \cite{DelZotto:2014hpa}. The 6d conformal matter theories needed for the atomic classification have the structure of some sort of generalized fundamentals and bifundamental ``matter'' fields. The generic 6d SCFTs are obtained from these elementary building blocks via fission and fusion operations \cite{Heckman:2018pqx,DelZotto:2018tcj} leading to a collection of rather constrained shapes for generalized quiver structures \cite{Heckman:2015bfa}. The resulting models can be further enriched, for example with the additional data of frozen singularities \cite{Tachikawa:2015wka,Bhardwaj:2018jgp}. This results in a rather comprehensive description of the existing 6d $(1,0)$ SCFTs \cite{Heckman:2013pva,Heckman:2015bfa,Bhardwaj:2015xxa,Bhardwaj:2019hhd}. As we discuss below, we believe this feature of the 6d classification will carry over to 5d. Namely, on top of the various 5d conformal matter systems we will also have to include discrete data to complete the classification. We comment about these further possibilities in the conclusions of this paper.

In order to develop our 5d atomic classification, it is interesting to ask whether SCFTs in other dimensions also admit an ``atomic classification'' of sort. From the 6d example, we expect that atomic classifications are a feature of theories with 8 (or less) conserved Poincar\'e supercharges, while for higher amounts of supersymmetries we expect the resulting SCFTs to behave like some sort of noble gases in the Mendeleev table of theories. Let us consider the four-dimensional case as another example. For 4d $\mathcal N \geq 2$ SCFTs several classification schemes have been designed \cite{Argyres:2015ffa,Xie:2015rpa,Argyres:2015gha,Chen:2016bzh,Argyres:2016xmc,Caorsi:2018zsq,Argyres:2020nrr,Argyres:2020wmq,Cecotti:2021ouq,Closset:2021lhd,Argyres:2022lah,Argyres:2022puv,Argyres:2022fwy,Cecotti:2022uep} which are all classifications of `geometric' type. In this context, moreover, algebraic classification schemes are also possible, building on the interplay between VOAs and bootstrap \cite{Beem:2013qxa,Beem:2014zpa,Bonetti:2018fqz} or modularity \cite{Kaidi:2022sng}. One could then ask whether 4d $\mathcal N=2$ SCFTs theories have an atomic classification as well. Theories with $\mathcal N=3$ and $\mathcal N = 4$ would be the noble gasses. Building on the original analysis for the theories of class $\mathcal S$ \cite{Gaiotto:2009we,Gaiotto:2009hg}, we see that indeed also 4d $\mathcal N=2$ SCFTs already have an obvious ``atomic'' structure. The building blocks of these systems are indeed generalized fundamentals, bifundamentals and a new ingredient, the trinions (or generalized tri-fundamentals), arising from 3-punctured spheres  \cite{Gaiotto:2009we} (see also \cite{Chacaltana:2010ks,Chacaltana:2011ze,Chacaltana:2012zy,Chacaltana:2012ch,Chacaltana:2013oka,Chacaltana:2015bna,Chacaltana:2017boe,Chacaltana:2018vhp} for a comprehensive study and \cite{Tachikawa:2015bga} for a review). In 4d there are also more exotic classes of theories (rare-earths and isotopes), the Argyres-Douglas (AD) models \cite{Argyres:1995jj,Eguchi:1996ds}, as well as other similar systems  \cite{Cecotti:2010fi,Cecotti:2011gu,Xie:2012hs,Closset:2020afy}.\footnote{ \ \ For all these AD-like theories the question whether these are to be thought as atoms themselves or as derivatives is debatable: AD models (almost by definition) arise as special fixed points of RG flows from more conventional molecules, however some AD theories can definitely be thought as building blocks for more general models. As an example, consider the $D_p^b(G)$ theories \cite{Cecotti:2012jx,Cecotti:2013lda,giaco1,Wang:2015mra}, which, despite being Argyres-Douglas points, indeed end up playing the role of generalized fundamentals in many contexts \cite{DelZotto:2015rca,Kang:2021lic,Kang:2021ccs}.} This quick superficial glance at 4d $\mathcal N=2$ models indicates that an atomic classification of these theories is, in a sense, already existing, especially for systems in class $\mathcal{S}$. For 4d $\mathcal N=1$ theories, to our best knowledge, neither a geometric, nor an algebraic, nor an atomic classification schemes are available. In that context, it is interesting to notice that from the geometric engineering perspective, evidence for the existence of yet another more general atom (the 4d $\mathcal N=1$ tetraons) was recently found from geometric engineering limits of M-theory on $G_2$ orbifolds \cite{Acharya:2023bth}.

\subsection{A 5d atomic classification} 5d SCFTs are governed by the representation theory of the unique five-dimensional superconformal algebra $\mathfrak f(4)$ and have known classification schemes of two kinds. On the one hand there is a classification in terms of their six-dimensional origin \cite{DelZotto:2017pti,Jefferson_2018,Bhardwaj:2018yhy,Bhardwaj:2018vuu,Bhardwaj:2019xeg,Bhardwaj:2020gyu}. On the other hand, there is a classification in terms of CY geometries \cite{Xie_2017,Apruzzi:2019opn,Closset:2020afy,Apruzzi_2020,Bourget:2023wlb, Collinucci:2021ofd,Collinucci:2021wty,DeMarco:2021try,DeMarco:2022dgh,Collinucci:2022rii}. Large classes of examples are obtained in this way. An especially large and interesting class is obtained from M-theory geometric engineering on orbifolds of $\mathbb C^3/\Gamma$ where $\Gamma \subseteq SU(3)$ is a finite subgroup \cite{Xie_2017,Acharya:2021jsp,Tian:2021cif,DelZotto:2022fnw}. 

For 5d SCFTs, however, an atomic classification scheme is still lacking. Evidence for such a scheme has been known since a long time. 5d SCFTs sporting $A \times A$ and $A \times D$ flavor algebras can be easily constructed from M-theory \cite{KatzVafa} --- one substantially ends up with free hypers (up to discrete gauging) \cite{Closset:2020scj,Closset:2021lwy,Collinucci:2022rii}. Similarly, it is known that generalized 5d trifundamental of $A$--type can be constructed via the M-theory dual of trivalent junctions of (p,q)-fivebranes \cite{Benini:2009gi}. The analogue of the fusion operation in 6d is provided geometrically, via glueing operations \cite{Hayashi:2019fsa}.

In this work we begin exploring the building blocks for an atomic classification,  looking at the existence of the simplest possible atoms, the ones of bi-fundamental type. Of course, from the above discussion we expect to find also further 5d trinion theories, generalizing the 5d $T_N$ theories to ones arising from collisions of 3 lines of singularities (possibly folded and or decorated by additional T-brane like data). These geometries are much more intricate to analyze in detail, and for this reason we mention them only briefly in the conclusion section, whereas in this project we mainly focus on engineering generalized bi-fundamentals theories of $(\mathfrak{g},\mathfrak{g})$ types.\footnote{ \ The detailed study of bi-fundamentals theories of $(\mathfrak{g},\mathfrak{g}')$ as well as of 5d trinions will appear elsewhere \cite{DMDZGSinprep}.}\\
\indent The question that motivates this work is the following: \textit{can bifundamental conformal matter theories, analogous to the ones in six dimensions, and yet not directly descending from them via elementary Kaluza-Klein reduction, be constructed in five dimensions?} Our main result is that the answer is in the affirmative, and that it involves a surprisingly simple M-theory construction. The existence of conformal matter readily allows some predictions on five-dimensional dualities that give generalizations of the well-known fiber/base dualities \cite{Hollowood:2003cv} to 5d quiver theories of $ADE$ type. In the following we give a summary of the main features of 5d conformal matter theories of ``bifundamental'' type.

\subsection{Generalized bi-fundamental conformal matter} In this work, we geometrically engineer 5d conformal matter ``bifundamental'' theories of types $(\mathfrak g,\mathfrak g)_{x_i}$ where $\mathfrak g \in ADE$, by looking at specific M-theory singularities, and the label $x_i$ refers to the fact that we find that in 5d there are various ``species'' of generalized bifundamentals, exhibiting slightly different properties which we identify. The techniques of \cite{Morrison_1997,Intriligator_1997} can be exploited to study the properties of these models along the Coulomb branch, displaying other possible gauge theory phases \cite{Closset_2019}. All these models have at least one gauge theory phase which we identify and exploit in order to confirm their flavor symmetries, exploiting the Tachikawa flavor enhancement criterion for 5d quiver gauge theories \cite{Tachikawa} --- see also \cite{Yonekura,Zafrir:2015uaa}. Moreover, we show that $(\mathfrak g,\mathfrak g)_{x_i}$ 5d conformal matter theories (with $x_i$ one among $x,y,z$) can be employed to form generalized linear 5d quiver SCFTs of the form
\begin{equation}\label{eq:generalizedlinearquiver}
\xymatrix{*++[F-]{\mathfrak{g}} \ar@{-}[rr]^{(\mathfrak g,\mathfrak g)_{x_{i_1}}} && *++[o][F-]{\mathfrak{g}}\ar@{-}[rr]^{(\mathfrak g,\mathfrak g)_{x_{i_2}}} &&  *++[o][F-]{\mathfrak{g}}\ar@{-}[rr]^{(\mathfrak g,\mathfrak g)_{x_{i_3}}} && \cdots\ar@{-}[rr]^{(\mathfrak g,\mathfrak g)_{x_{{i_{n-1}}}}} &&*++[o][F-]{\mathfrak{g}} \ar@{-}[rr]^{(\mathfrak g,\mathfrak g)_{x_{{i_N}}}} && *++[F-]{\mathfrak{g}} }
\end{equation}
where the edges are 5d generalized bifundamentals, and the circular nodes represent a fusion operation corresponding to a 5d generalization of the familiar operation of gauging a diagonal flavor symmetry of type $\mathfrak g$ \cite{Hayashi:2019fsa}.
\begin{table}
\label{table:DuVal list}
    \centering
    \begin{tabular}{c|c}
        $Y_{A_k}$ &  $P_{A_k}(x,y,z) = x^2+y^2-z^{k+1}$ \\[2px]
        $Y_{D_k}$ &  $P_{D_k}(x,y,z) = x^2 +zy^2+z^{k-1}$ \\[2px]
        $Y_{E_6}$ &  $P_{E_6}(x,y,z) =x^2+y^3+z^4$ \\[2px]
        $Y_{E_7}$ &  $P_{E_7}(x,y,z) = x^2+y^3+yz^3$ \\[2px]
        $Y_{E_8}$ &  $P_{E_8}(x,y,z) = x^2+y^3+z^5$,
    \end{tabular}
    \caption{Representatives for each isomorphism class of Du Val singularities.}
    \label{tab:singularities}
\end{table}

Our construction starts from a Du Val singularity (see Table \ref{tab:singularities})
\begin{equation}
P_{\mathfrak g}(x_1,x_2,x_3) = 0, \qquad \mathfrak g \in ADE
\end{equation}
which is a hypersurface singularity $Y_{\mathfrak g} \subset \mathbb C^3$. We then apply a base change $x_i= uv$ (with $i=1,2,3$) and obtain a CY three-fold singularity $X^{(i)}_{\mathfrak g} \subset \mathbb C^4$
\begin{equation}
X^{(i)}_{\mathfrak g} \quad \colon \qquad \begin{cases}P_{\mathfrak g}(x_1,x_2,x_3) = 0 \\ uv = x_i \end{cases}
\end{equation}
The latter has by construction two transversally intersecting singularities of type $\mathfrak g$ located along $u = 0$ and $v=0$ (see Figure \ref{fig:intersection}), and possibly an enhanced singularity at the intersection point $x_j=u=v=0$ for $j\neq i$. By crepantly blowing up the singularities we show that at the intersection point there lies trapped five-dimensional matter, charged under $\mathfrak g \times \mathfrak g$ flavor symmetry, as can be expected from the geometric perspective. The theory so obtained has different features depending on the choice of coordinate $x_i$ for the base change, and correspondingly we obtain 5d conformal matter theories of different types. Gauging together these theories gives rise to generalized linear 5d quivers of $A$-type, which have a geometrical counterpart in the singularities
\begin{equation}
\begin{cases}P_{\mathfrak g}(x_1,x_2,x_3) = 0 \\ uv =x_1^{k_1}x_2^{k_2}x_3^{k_3}  \end{cases}, \quad k_i \geq 0
\end{equation}
which follows from a detailed analysis we present below. Furthermore, in many examples we can show that the theory obtained in such fashion is \textit{not} a simple KK dimensional reduction of a 6d conformal matter with the same flavor symmetry, thus guaranteeing our construction is giving rise to genuine \textit{5d conformal matter}. Indeed, in most cases we consider, the resulting singularities are not embedded into F-theory models as local limits sending the F-theory torus to infinite volume.\footnote{ \ \ We note that, however, it is still possible that the theories we study descend in a non-obvious way by some parent theories in 6d, as per the conjecture of \cite{Jefferson_2018}.} As a  further result, we demonstrate that these quivers undergo 5d dualities that generalize fiber/base dualities to $D$ and $E$-type diagrams.
From our results, it follows that linear $A$ quivers of the type depicted in Equation \ref{eq:generalizedlinearquiver} enjoy a flavor symmetry algebra enhancement of the form
\begin{equation}
    \mathcal{G}_{UV} = \mathfrak{g}\times \mathfrak{g} \times \mathfrak{su}(n_1)\times\cdots \times \mathfrak{su}(n_{\nu}) \times \mathfrak{u}(1)^{\nu-1},
\end{equation}
with $n_i$ and $\nu$ determined by the specific choice of quiver. We refer our readers to Section \ref{sec:quiverdualities} for the detailed list of generalized 5d dualities we obtain in this paper.

Moreover, we rule out $P$-valent junctions of the form
\begin{equation}
\begin{gathered}
\xymatrix{
*++[F-]{\mathfrak{g}} \ar@{-}[ddrr]_{(\mathfrak g,\mathfrak g)_{x_{i_1}}}&*++[F-]{\mathfrak{g}} \ar@{-}[ddr]^{(\mathfrak g,\mathfrak g)_{x_{i_2}}}& \cdots &&*++[F-]{\mathfrak{g}} \ar@{-}[ddll]^{(\mathfrak g,\mathfrak g)_{x_{i_P}}}\\
&&&&\\
&& *++[o][F-]{\mathfrak{g}} &&}
\end{gathered}
\end{equation}
for all $P\geq 3$ by requiring the resulting singularities are at finite distance in the CY moduli space. Therefore, 5d bifundamental conformal matter theories of types $(\mathfrak g,\mathfrak g)_{x_i}$ can only form linear quivers with gauge nodes $\mathfrak g$.


\begin{figure}
\centering
 \scalebox{0.9}{
    \begin{tikzpicture}
        \draw[thick] (-2.5,-3)--(2.5,3);
        \draw[thick] (-2.5,3)--(2.5,-3);
        \draw[thick,dashed,->] (-1,3.5) to [out=270, in = 40] (-1.6,2.2);
        \draw[thick,dashed,->] (1,3.5) to [out=270, in = 140] (1.6,2.2);
        \draw[thick,dashed,<-] (0.5,0)--(1.5,0);
        \draw[fill=red] (0,0) circle (0.13);
        \node at (-2.4,3.8) {$\mathfrak{g}$ singularity};
         \node at (2.4,3.8) {$\mathfrak{g}$ singularity};
        \node at (3.7,0) {Enhanced singularity};
        \node[rotate=49] at (-1.5,-2.2) {$u$-axis};
        \node[rotate=-49] at (1.5,-2.2) {$v$-axis};
        \end{tikzpicture}}
    \caption{Pictorial representation of a threefold with two transversely intersecting lines of $\mathfrak g$ singularities.}
     \label{fig:intersection}
\end{figure}
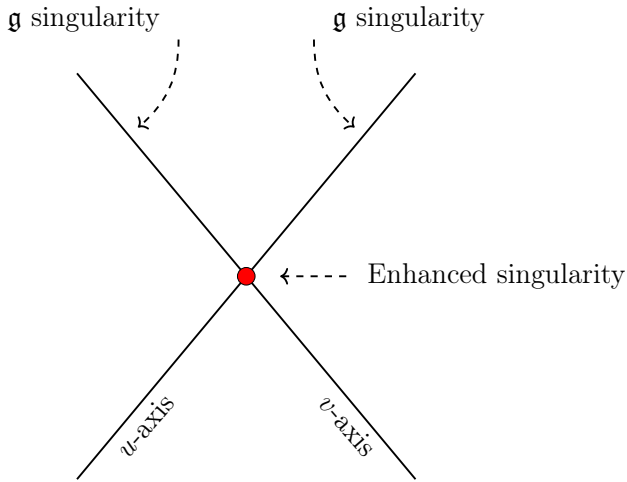  


\subsection{Outline of this paper}
The structure of this paper is as follows. In Section \ref{sec:geom-eng} we quickly review the dictionary between geometric data and physics of the 5d SCFT and its low-energy quiver gauge theory phase. In Section \ref{sec:coreidea} we introduce the basics of our construction, in particular in Section \ref{sec:singularities} we briefly review the theory of $ADE$ singularities and in Section \ref{sec:threefold singularities via base change} we discuss the threefold singularities $X^{(i)}_{\mathfrak g}$ in details. In Section \ref{sec:exampleE6} we illustrate the resolution procedure of the threefold singularities following an example, characterizing the corresponding 5d SCFT and its low-energy quiver. We generalize this example in Section \ref{sec:generalrecipe} providing the construction of all the $(\mathfrak g,\mathfrak g)_{x_i}$ conformal matter theories. In particular, in Section \ref{sec:quivers} we write down the low-energy quivers for all the 5d SCFTs engineered by the threefolds introduced in the previous sections. In all cases, we compute the candidate UV enhanced flavor symmetry both from geometry and from gauge theory. We also make contact with existing constructions in the literature and 6d uplifts. In section \ref{sec:exceptional quivers}, we roll out a natural generalization of our construction, producing quivers with nodes displaying exceptional groups, showing that indeed these SCFTs have the right flavor symmetries to earn the name of \textit{5d bifundamental conformal matter theories}, serving as building blocks for 5d exceptional linear quivers. As an outcome we obtain a novel plethora of 5d dualities, which are summarized in Section \ref{sec:quiverdualities}. We present our conclusion and outline some future directions we are currently developing in Section \ref{sec:conclusions}. Several more technical appendices complement the results presented in this work. In particular, we rigorously present the explicit resolution maps of Du Val singularities in appendix \ref{sec:duvalres}, the toric tools employed throughout the main text in appendix \ref{sec:appendixtoriclocalmodels}, and go through a consistency check of the quiver gauge theories via standard techniques involving their prepotential in appendix \ref{sec:prepotential}.

\section{Geometric origin of generalized 5d bifundamentals}

In this section we start by briefly reviewing the M-theory geometric engineering of 5d $\mathcal{N}=1$ SCFTs, in order to fix few notations and conventions we are using below in \ref{sec:geom-eng}. In \ref{sec:coreidea} we illustrate the relevant geometries that give rise to 5d conformal matter theories of generalized bifundamental type.

\subsection{Geometric engineering 5d SCFTs: lightning review}\label{sec:geom-eng}

It is well-known that M-theory reduced on a Calabi-Yau threefold $X$ with a canonical singularity gives rise to a 5d  $\mathcal{N}=1$ SCFT on flat spacetime, denoted by $\mathcal{T}_{5d}$ \cite{Seiberg1996,Morrison_1997,Douglas:1996xp} -- see also \cite{Intriligator_1997} and \cite{Jefferson:2017ahm,Closset:2018bjz,Jefferson:2018irk}, where the role of tensionless 5d BPS strings was emphasized in the constraints arising from 5d Coulomb branches. It was actually conjectured that \textit{all} 5d SCFTs have such an engineering \cite{Xie_2017}.  The Calabi-Yau threefolds we will consider are precisely of such kind. The geometric features of the compactification space have a direct correspondence in terms of the properties of the 5d theory. Roughly speaking, the moduli spaces of vacua of the 5d SCFT can be described as follows:
\begin{equation*}
    \begin{split}
    \text{Resolutions of } X \quad &\longleftrightarrow \quad \text{Coulomb branch of }\mathcal{T}_{5d}\\
    \text{Deformations of } X \quad &\longleftrightarrow \quad \text{Higgs branch of }\mathcal{T}_{5d}\\
    \end{split}
\end{equation*}
In this work we focus on the resolution of $X$, and hence on exploring the Coulomb branch of the 5d SCFTs under examination. In this fashion, we will be able to (indirectly) extract information on an intrinsic Higgs branch feature, namely the flavor symmetry. \\ 
\indent M-theory geometric engineering further provides a precise way to translate information on a resolved phase $\widetilde{X}$ of $X$ to data of the Coulomb branch of the SCFT. Schematically, the dictionary goes as follows:
\begin{itemize}
    \item Compact divisors in $H_4(\widetilde{X},\mathbb Z)$ are dual to harmonic normalizable 2-forms, which from the reduction of the M-theory $C$-field, give rise to the $U(1)$ vector fields in the 5d $\mathcal N=1$ vector multiplets. The scalar components of these vector multiplets arise from K\"ahler moduli corresponding to the volumes of those compact curves that are dual to compact divisors. These volumes give the Coulomb branch vevs of the SCFT. Wrapping M5 branes on compact divisors gives rise to BPS monopole strings;
    \item Non-compact divisors encode the flavor symmetry. Curves arising from the intersections of compact divisors with non-compact ones correspond to relevant deformations of the SCFT (i.e.\ masses or gauge couplings);
    \item Wrapping M2 branes on curves gives rise to BPS particles with spins obtained quantizing the resulting moduli spaces \cite{Witten:1996qb}. In particular, $\mathbb P^1$'s with normal bundle $\mathcal O_{\mathbb P^1}(-1) \oplus \mathcal O_{\mathbb P^1}(-1)$ are rigid, and give rise to 5d hypermultiplets in this way.
\end{itemize}
Different resolutions corresponding to birational CYs related to one another via flop transitions are interpreted as different chambers in the Coulomb phase of the 5d SCFT, which are separated by codimension-one walls where a given BPS particle becomes massless. The resulting jump in the prepotential by integrating such a particle back in, corresponds to the change in the triple intersection numbers among the birationally equivalent threefolds. In the singular limit, all the volumes of the compact divisors and compact curves have to vanish as we reach the SCFT point and all scales must disappear. In general, such limit can be reached if all compact divisors and curves are \textit{shrinkable} in the sense of \cite{Jefferson_2018}: if one adopts a top-down approach, i.e.\ conjures up a resolved geometry composed of a collection of compact divisors and curves, this is a non-trivial check. On the other hand, in the following we will adopt a bottom-up stance, starting from a singular geometry and performing a sequence of blowups: for such cases, \textit{shrinkability} is guaranteed by construction.\\
\indent Non-compact curves of singularities in M-theory can be interpreted as higher dimensional degrees of freedom which in the 5d case happen to decouple from the conformal fixed point, leading to a (generalized) flavor symmetry for these theories \cite{Hayashi:2019fsa,Acharya:2023bth}. Hence in the search of 5d bifundamentals, we are naturally lead to look at a non-isolated singularity located at the collision of two ADE singularities of type $\mathfrak g$.\\
\indent Given an SCFT fixed point $\mathcal{T}_{5d}$ in the UV, we can ask which are the possible weakly-coupled quiver gauge theories that arise as its mass deformations. We can explicitly write down one of these quiver gauge theory limits by \textit{choosing} a specific resolution of $X$, say $\widetilde{X}$. Different choices are equally valid and correspond to different mass deformations of the same SCFT and/or different chambers of the 5d Coulomb branch. This gives rise to the phenomenon of 5d dualities \cite{Aharony:1997bh} (see also eg. \cite{Gaiotto:2015una,Closset:2018bjz,Bhardwaj:2019ngx}).
The examples we discuss in this paper all admit Coulomb branch chambers with a 5d gauge theory interpretation, more precisely we obtain 5d quiver gauge theories with gauge algebras $\prod_i \mathfrak{su}(n_i)$. 
Focusing on a specific choice $\widetilde{X}$ compatible with such a gauge theory phase, it is immediate to translate the 5d SCFT data into its low-energy quiver counterpart:
\begin{itemize}
    \item Compact divisors and their intersection patterns encode the data of the gauge nodes and bifundamental matter for the quiver \cite{KatzVafa};
    \item Non-compact divisors dictate the presence of flavor nodes;\footnote{ \ The complete UV flavor symmetry of the quiver gauge theory can be determined employing the techniques of \cite{Tachikawa} which we will briefly summarize in Section \ref{sec:flavorUV}. A different approach to flavor symmetry was also discussed in \cite{Apruzzi:2019opn,Bhardwaj:2020ruf,Bhardwaj:2020avz}: when overlapping, our results agree. }
    \item The compact divisors turn out to be ruled surfaces, possibly blown-up at a number of points. The blow-up pattern dictates finer details of the quiver node (such as e.g.\ the Chern-Simons levels).\footnote{ \ Chern-Simons levels can also be computed by comparing the prepotential computed from the quiver with the expectation given by the triple intersections in the geometry. In large quivers this may prove a daunting task, and we will indeed recur to a different argument for gleaning the CS levels.}
\end{itemize}
Keeping this dictionary in mind, in the next section we explicitly construct the singular threefolds we covet, in order to retrieve the five-dimensional conformal matter.

\subsection{Intersecting families of \texorpdfstring{$ADE$}{} singularities}\label{sec:coreidea}

In this section we begin in \ref{sec:singularities} by quickly settling our conventions for the well-known theory of Du Val singularities and their resolutions. In section \ref{sec:threefold singularities via base change} we illustrate the singular threefold geometries corresponding to 5d generalized bifundamentals. Readers not interested in the mathematical details of these singularities can safely skip to the next section.

\subsubsection{\texorpdfstring{$ADE$}{} singularities and their resolution}\label{sec:singularities}

Recall that there are only five families of isomorphism classes of absolutely isolated rational double points of dimension 2, called $ADE$ or Du Val or Kleinian singularities \cite{DuVal}.\footnote{ \ We will employ these denominations interchangeably.} More precisely, there are two infinite families (that we denote as $Y_{A_k}$ and $Y_{D_k}$) and three sporadic examples ($Y_{E_6}$, $Y_{E_7}$ and $Y_{E_8}$). In this work, we fix as representatives $Y_{\mathfrak{g}}$ for each isomorphism class the zero loci of the polynomials in \cref{tab:singularities} (with  $k \geq 1$ in the $A_k$ case and $k \geq 4$ in the $D_k$ case):

\indent Recall that each of these singularities has trivial canonical bundle. Moreover, they can be characterized as the only surface singularities admitting a resolution of singularities with trivial canonical bundle, i.e., in this setting, a crepant resolution.

The exceptional locus of the crepant resolution of a Du Val singularity $Y_{\mathfrak{g}}$ is a collection of $\mathbb P^1$'s with normal bundles of degree -2 intersecting each other according to the Dynkin diagram of $\mathfrak{g}$, as displayed in Table \ref{tab:DynkinADE}. 

\begin{table}
\label{Table:duvalsings}
\centering
    
    \begin{tabular}{c|c}
        $Y_{A_k}$ &  $\overbrace{\begin{matrix}
            \begin{tikzpicture}
    
    \draw (0.1,1.5)--(0.9,1.5);
    \draw (1.1,1.5)--(1.9,1.5);

    \draw (0,1.5) circle (0.1); 
    
	\draw (1,1.5) circle (0.1); 
    
	\draw (2,1.5) circle (0.1); 
	\draw (3,1.5) circle (0.1);
 \node a t (2.5,1.5) {$\cdots $};
\end{tikzpicture}  
        \end{matrix}}^{\mbox{\tiny $k$ nodes}}$ \\[2px]
        $Y_{D_k}$ &   $\underbrace{\begin{matrix}
            \begin{tikzpicture}
    
    \draw (0.1,1.5)--(0.9,1.5);
    \draw (1.1,1.5)--(1.9,1.5);
    \draw (3.1,1.5)--(3.9,1.5);
    \draw (1,1.6)--(1,2.4);

    \draw (0,1.5) circle (0.1); 
    
	\draw (1,1.5) circle (0.1); 
    
	\draw (2,1.5) circle (0.1); 
	\draw (3,1.5) circle (0.1); 
	\draw (4,1.5) circle (0.1);
	\draw (1,2.5) circle (0.1);
 \node a t (2.5,1.5) {$\cdots $};
\end{tikzpicture}  
        \end{matrix}}_{\mbox{\tiny $k$ nodes}}$ \\[2px]
        $Y_{E_6}$ &  $\begin{matrix}
            \begin{tikzpicture}
    
    \draw (0.1,1.5)--(0.9,1.5);
    \draw (1.1,1.5)--(1.9,1.5);
    \draw (2.1,1.5)--(2.9,1.5);
    \draw (3.1,1.5)--(3.9,1.5);
    \draw (2,1.6)--(2,2.4);

    \draw (0,1.5) circle (0.1); 
    
	\draw (1,1.5) circle (0.1); 
    
	\draw (2,1.5) circle (0.1); 
	\draw (3,1.5) circle (0.1); 
	\draw (4,1.5) circle (0.1);
	\draw (2,2.5) circle (0.1); 
\end{tikzpicture}  
        \end{matrix}$\\[2px]
        $Y_{E_7}$ &   $\begin{matrix}
            \begin{tikzpicture}
    
    \draw (0.1,1.5)--(0.9,1.5);
    \draw (1.1,1.5)--(1.9,1.5);
    \draw (2.1,1.5)--(2.9,1.5);
    \draw (3.1,1.5)--(3.9,1.5);
    \draw (4.1,1.5)--(4.9,1.5);
    \draw (2,1.6)--(2,2.4);

    \draw (0,1.5) circle (0.1); 
    
	\draw (1,1.5) circle (0.1); 
    
	\draw (2,1.5) circle (0.1); 
	\draw (3,1.5) circle (0.1); 
	\draw (4,1.5) circle (0.1);
	\draw (5,1.5) circle (0.1);
	\draw (2,2.5) circle (0.1); 
\end{tikzpicture}  
        \end{matrix}$ \\[2px]
        $Y_{E_8}$ &  $\begin{matrix}
            \begin{tikzpicture}
    
    \draw (0.1,1.5)--(0.9,1.5);
    \draw (1.1,1.5)--(1.9,1.5);
    \draw (2.1,1.5)--(2.9,1.5);
    \draw (3.1,1.5)--(3.9,1.5);
    \draw (4.1,1.5)--(4.9,1.5);
    \draw (5.1,1.5)--(5.9,1.5);
    \draw (2,1.6)--(2,2.4);

    \draw (0,1.5) circle (0.1); 
    
	\draw (1,1.5) circle (0.1); 
    
	\draw (2,1.5) circle (0.1); 
	\draw (3,1.5) circle (0.1); 
	\draw (4,1.5) circle (0.1);
	\draw (5,1.5) circle (0.1);
	\draw (6,1.5) circle (0.1);
	\draw (2,2.5) circle (0.1); 
\end{tikzpicture}  
        \end{matrix}$ .
    \end{tabular}  
    \caption{The dual graph of the exceptional locus of the crepant resolution of the Du Val singularities.}
    \label{tab:DynkinADE}
\end{table}
In appendix \ref{sec:duvalres} we show explicit maps for the resolutions of the $ADE$ singularities of type $A_k$, $D_4, D_5, E_6, E_7, E_8$, that will be best suited for our needs in the course of the work.\\
\indent In the next section, we show how we can construct threefolds displaying the desired properties to yield five-dimensional conformal matter, employing the $ADE$ singularities as building blocks.

\subsubsection{Singular geometries for 5d bifundamentals}\label{sec:threefold singularities via base change}
We are interested in describing affine Calabi-Yau threefold hypersurface singularities arising by the transversal collision of two lines of $Y_{\mathfrak{g}}$ singularities. This, from an M-theory perspective, corresponds to the fact that we expect to find a flavor symmetry of type $\mathfrak{g}\times \mathfrak{g}$ coupled to the five-dimensional degrees of freedom trapped at the point where the two lines collide. To obtain this construction we proceed by ``gluing'' the singular surface $Y_{\mathfrak{g}}$ along two colliding lines that we model as the zero locus, in $(u,v) \in \mathbb C^2$, of the monomial $uv$. Mathematically, this procedure is implemented by a ``base-change''. Intuitively this corresponds to substituting, in the equation defining $Y_{\mathfrak{g}}$, one among $\bullet= x,y, z$ with $uv$.\footnote{ \ 
More formally, we can draw the following commutative diagram:

  \begin{equation} \label{eq:basechangeI}
  \begin{tikzcd}[ampersand replacement=\&]
       \arrow{r}\mathbb C^2 \underset{\mathbb C^1}{\times} Y_{\mathfrak{g}}\arrow{d} \arrow{r} \&  X\arrow{d}{\pi_\bullet} \& (x,y,z)\arrow[mapsto]{d} \\
         \mathbb C^2 \arrow{r} \& \mathbb C \ni{t} \& t = \bullet\\
         (u,v)\arrow[mapsto]{r} \& t = uv \&
    \end{tikzcd}
    \end{equation}
    where $\bullet$ is either $x$, $y$, or $z$. For example we can re-obtain the coordinate ring $R$ of $X_{E_6}^x = \mathbb C^2 \underset{\mathbb C^1}{\times} Y_{\mathfrak{g}} $ defined in \eqref{E6yexample} applying the definition of base-change:
\begin{equation*}
R = \frac{\mathbb C[u,v,x,y,z,t]}{\left(t - uv, t - x,\re{x}^2 + y^3 + z^4\right)} \cong \frac{\mathbb C[x,z,u,v]}{\left((\re{uv})^2 + y^3 + z^4\right)}
\end{equation*}
}

We can readily visualize such prescription in an example: we take  $ Y_{\mathfrak{g}} = E_{6}$ and use $\bullet = x$.
Hence we get:
\begin{equation}\label{E6yexample}
    Y_{E_6}: \quad x^2+y^3+z^4 =0 \quad \xrightarrow{x=uv}\quad X_{E_6}^x: \quad (uv)^2+y^3+z^4=0.
\end{equation}
The singular locus of the threefold $X_{E_6}^x$ is:
\begin{equation}
    \text{Sing}(X_{E_6}^x) = \{(uv,y,z) =0 \} \subset \mathbb{C}^4.
\end{equation}
The singular locus stratifies into two parts where we have, respectively $(u,v) \neq (0,0)$ and $(u,v)=(0,0)$. The former  consists of two one-dimensional families of  $E_6$ singularities. The latter component is an enhanced singularity supported at the collision point. This configuration can be pictorially represented as in Figure \ref{fig:intersection}.\\

The general prescription is the following: call $\bullet$ one among $x,y,z$ and replace it with the monomial $uv$ in the equation defining $Y_{\mathfrak{g}}$. We obtain a singular threefold, that we denote by $X_{\mathfrak{g}}^{\bullet}$: 
\begin{equation}
\label{eq:threefoldsingdef}
 X_{\mathfrak{g}}^{\bullet}\equiv \Set{ P_{\mathfrak{g}}(x,y,z)\rvert_{\bullet = uv} = 0} \subset \mathbb C^4,
\end{equation}
where the coordinates of $\mathbb C^4$ are $u,$ $v$ and the two coordinates among $(x,y,z)$ that we \textit{did not} substitute with the monomial $uv$.
The singular locus of $X_{\mathfrak{g}}^{\bullet}$ is, by direct computation,
\begin{equation}
\label{eq:singularlocus}
\text{Sing}(X_{\mathfrak{g}}^{\bullet}) = \Set{ (x,y,z)\rvert_{\bullet \to uv} = 0 } \subset \mathbb C^4,
\end{equation}
which is stratified as in the previous example, with two lines of $Y_{\mathfrak{g}}$ singularities intersecting at a point where an enhanced singularity resides. Let us remark the relevance of the threefolds defined in \eqref{eq:threefoldsingdef}:\\

\indent \textit{The singular threefolds $X_{\mathfrak{g}}^{\bullet}$ are the basic building blocks of this work. The expectation is that M-theory geometric engineering on $X_{\mathfrak{g}}^{\bullet}$, with $\mathfrak{g}=A,D,E$, produces an interacting 5d $\mathcal{N}=1$ SCFT with matter charged under the flavor group $\mathfrak{g}\times\mathfrak{g}$. We will prove this fact in the subsequent sections, highlighting the due caveats.}\\

\indent To substantiate the above claim we must perform a crepant resolution of the singularities $X_{\mathfrak{g}}^{\bullet}$ and directly probe the Coulomb branch. Indirectly, this will give us information \cite{Yonekura} on the flavor symmetry of the SCFT. In general, the threefolds $X_{\mathfrak{g}}^{\bullet}$ admit many possible crepant resolutions, each one corresponding to a certain chamber of the $X_{\mathfrak{g}}^{\bullet}$ K\"ahler cone. In order to study the geometry and the physics of $X_{\mathfrak{g}}^{\bullet}$ it is convenient to exhibit the recipe to compute a specific (full) resolution. This resolution, in section \ref{sec:exampleE6quiver}, will allow us to write down a quiver that describes, physically, the Coulomb phase of the 5d SCFT obtained by M-theory geometric engineering on $X_{\mathfrak{g}}^{\bullet}$.\\
Let us add a few remarks, before delving into the core of the work:
\begin{itemize}
    \item Consider the $A_k$ singularities in a slightly different notation with respect to Table \ref{tab:singularities}, namely written as $\tilde{x}\tilde{y}=z^{k+1}$. Applying the base changes introduced in \eqref{eq:threefoldsingdef} can yield two substantially different outcomes, that we will not explore, as they have already been profusely studied in the literature:
    \begin{itemize}
        \item if we apply the base change $z = uv$, we fall into the cases analyzed in the seminal paper \cite{KatzVafa}, that yield a collection of hypermultiplets localized on the origin, charged only under $\mathfrak{su}(k+1)\times \mathfrak{su}(k+1)$ flavor symmetry. This happens because the resulting threefold singularity is of compound Du Val (cDV) type, and can thus give rise only to a small resolution \cite{reidminimal1983};
        \item if we apply the base change $\tilde{x} = uv$ (or, equivalently, $\tilde{y}= uv$) we obtain the famous $T_{k+1}$ theories (introduced in the 5d setting in \cite{Benini:2009gi}). Since the resolution of these singularities displays compact exceptional divisors they give rise to a non-trivial gauge dynamics. These singularities are \textit{not} of cDV type, as it is evident from their presentation.
    \end{itemize}
     \item Now go back to the $A_k$ singularities written as:
    \begin{equation}
        x^2 +y^2 = z^{k+1},
    \end{equation}
    If we substitute $x = uv$ or $y =uv$, for $k>1$ we obtain a non cDV singularity, and hence we expect some interesting interacting 5d SCFT.\footnote{ \ The $k=1$ case is evidently of cDV type and thus cannot yield an interacting theory. Furthermore, choosing $x = uv$ or $y=uv$ is completely equivalent for symmetry reasons: we will stick to considering the $x=uv$ case in the following.} We will come back to these theories in the following sections.
    \item Any base change of the kind displayed in equation \eqref{eq:threefoldsingdef} transforms $D$ and $E$ singularities into a threefold which is \textit{not} of cDV type. 
    Thus, it is natural to expect that the geometry gives rise to compact divisors on top of the origin (we will rigorously prove this fact in section \ref{sec:exampleE6} and \ref{sec:generalrecipe}), resulting in a non-trivial 5d theory. These are the theories that we are going to investigate right away.
\end{itemize}

Let us start in the next section by computing the resolution in a concrete example.

\section{A concrete example: $(E_6,E_6)_x$ conformal matter}\label{sec:exampleE6}

In this section we present a detailed analysis for the $(E_6,E_6)_x$ conformal matter. In Section \ref{sec:concreslutio} we present the detailed resolution of the singularity $X_{E_6}^x$. In Section \ref{sec:exampleE6quiver} we interpret our result through the lense of geometric engineering and we identify a quiver gauge theory phase. In section \ref{sec:flavorUV} we confirm our prediction on the UV flavor symmetry from geometric engineering exploiting field theoretical methods \cite{Tachikawa}.

\subsection{Resolution of the singularity $X_{E_6}^x$}\label{sec:concreslutio}

This section is also slightly technical and our readers that are not interested in the details of the resulting algebraic geometry can skip to section \ref{sec:exampleE6quiver} after having a quick glance at the intuitive description of the resulting resolved geometry in Figure \ref{fig:E6intuitive}.

\medskip

Consider the threefold $X_{E_6}^x$, that we already encountered in \eqref{E6yexample}. $X_{E_6}^x$ sports two lines of $E_6$ singularities intersecting transversally at the origin, parametrized by the coordinates $u$ and $v$ appearing in the base change $x=uv$. We repeat the presentation of $X_{E_6}^x$ as a hypersurface equation, for ease and clarity:
\begin{equation}
\label{eq:explicitex1}
    X_{E_6}^x: \quad (uv)^2+y^3+z^4=0.
\end{equation}
We now give a recipe to completely resolve the singularity and extract the related physical consequences. We remark here that, if we forget about the Chern-Simons levels and the $\theta$-angles, following the resolution procedure up to step 2 already permits to give a 5d quiver describing a (fully)  resolved phase $\tilde{\tilde{X}}_{E_6}^x$ of $X_{E_6}^x$. The resolution procedure involves the following steps:
\begin{enumerate}
\item The singularities outside of the origin of \eqref{eq:explicitex1} are  $E_6$ Du Val singularities, lying on top of the lines $y=z=u=0$ and $y=z=v=0$. Replacing $x \to uv$ in the resolution maps \eqref{eq:newpatchE6resmap}, \eqref{eq:transfuncttailE6} we can fully resolve the singularities outside the origin $y =  z = u = v = 0$. As we will see in \cref{sec:resbasechange}, this corresponds to lifting to the \textit{resolved} Du Val surface $\tilde{Y}_{E_6}$ the base-change \eqref{E6yexample} that we used to obtain $X_{E_6}^x$ from the \textit{singular} Du Val surface $Y_{E_6}$. We call  $\tilde{X}_{E_6}^x$ the partially resolved threefold obtained in this way and 
$\varepsilon$ the partial resolution map 
\begin{equation}
    \varepsilon: \quad \tilde{X}_{E_6}^x \longrightarrow X_{E_6}^x.
\end{equation}
 The singular locus of $\tilde{X}_{E_6}^x$ is all contained in the subvariety $\varepsilon^{-1}(0)$ contracted to the origin $ y = z=u=v= 0$. $\varepsilon^{-1}(0)$ consists of a collection of $\mathbb{P}^1_i$'s, with $i = 1, \ldots, 6$. Each of the $\mathbb P^1_i$'s is a line of singularities of type $A_{r_i}$ with $r_i \geq 0$, with enhanced singularities at the intersection points $q_{ij}= \mathbb P^1_i \cap \mathbb P^1_j$. Intuitively, the situation after the first blowup is depicted in Figure \ref{fig:E6firstblowup}.
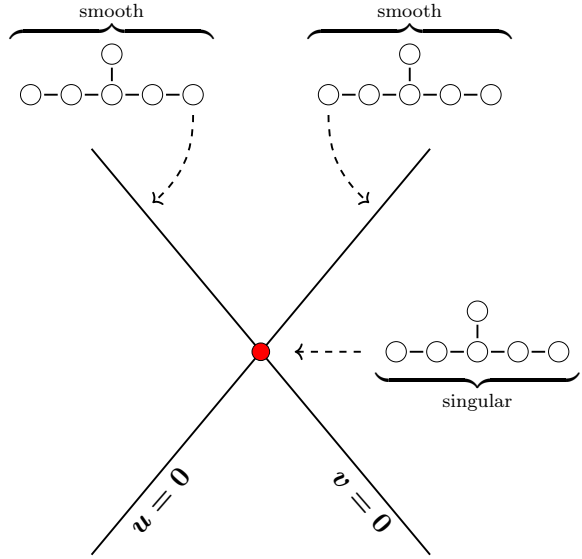
\begin{figure}[H]
\centering
 \scalebox{0.9}{
    \begin{tikzpicture}
        \draw[thick] (-2.5,-3)--(2.5,3);
        \draw[thick] (-2.5,3)--(2.5,-3);

        \draw[thick,dashed,->] (-1,3.5) to [out=270, in = 40] (-1.6,2.2);
        \draw[thick,dashed,->] (1,3.5) to [out=270, in = 140] (1.6,2.2);
        \draw[thick,dashed,<-] (0.5,0)--(1.5,0);
        \draw[fill=red] (0,0) circle (0.13);
        \node at (-2.2,4.9) {$\overbrace{\hspace{3cm}}^{\text{smooth}}$};
        \draw (-3.4,3.8) circle (0.15);
        \draw[thick] (-3.2,3.8)--(-3.0,3.8);
        \draw (-2.8,3.8) circle (0.15);
        \draw[thick] (-2.6,3.8)--(-2.4,3.8);
        \draw (-2.2,3.8) circle (0.15);
        \draw[thick] (-2,3.8)--(-1.8,3.8);
        \draw (-1.6,3.8) circle (0.15);
        \draw[thick] (-1.4,3.8)--(-1.2,3.8);
        \draw (-1,3.8) circle (0.15);
        \draw[thick] (-2.2,4.0)--(-2.2,4.2);
        \draw (-2.2,4.4) circle (0.15);
        \node at (2.2,4.9) {$\overbrace{\hspace{3cm}}^{\text{smooth}}$};
        \draw (3.4,3.8) circle (0.15);
        \draw[thick] (3.2,3.8)--(3.0,3.8);
        \draw (2.8,3.8) circle (0.15);
        \draw[thick] (2.6,3.8)--(2.4,3.8);
        \draw (2.2,3.8) circle (0.15);
        \draw[thick] (2,3.8)--(1.8,3.8);
        \draw (1.6,3.8) circle (0.15);
        \draw[thick] (1.4,3.8)--(1.2,3.8);
        \draw (1,3.8) circle (0.15);
        \draw[thick] (2.2,4.0)--(2.2,4.2);
        \draw (2.2,4.4) circle (0.15);
          \node at (3.2,-0.6) {$\underbrace{\hspace{3cm}}_{\text{singular}}$};
        \draw (4.4,0) circle (0.15);
        \draw[thick] (4.2,0)--(4,0);
        \draw (3.8,0) circle (0.15);
        \draw[thick] (3.6,0)--(3.4,0);
        \draw (3.2,0) circle (0.15);
        \draw[thick] (3,0)--(2.8,0);
        \draw (2.6,0) circle (0.15);
        \draw[thick] (2.4,0)--(2.2,0);
        \draw (2,0) circle (0.15);
        \draw[thick] (3.2,0.2)--(3.2,0.4);
        \draw (3.2,0.6) circle (0.15);
        \node[rotate=49] at (-1.5,-2.2) {$\boldsymbol{u=0}$};
        \node[rotate=-49] at (1.5,-2.2) {$\boldsymbol{v=0}$};
        \end{tikzpicture}}
    \caption{Pictorial representation of $\tilde{X}_{E_6}^{x}$. On top of the origin there remains a collection of singular $\mathbb{P}^1$'s, arranged like a $E_6$ Dynkin diagram.}
     \label{fig:E6firstblowup}
\end{figure}  

\item  To compute the labels $r_i$ we proceed as follows. After replacing $x \to uv$ in the $E_6$ resolution map \eqref{eq:newpatchE6resmap}, \eqref{eq:transfuncttailE6} we obtain a hypersurface equation for each chart $U_0,\ldots,U_6$ of \eqref{eq:newpatchE6resmap}, \eqref{eq:transfuncttailE6}. We display these local data in \cref{tab:localmodelsE6z}. Such hypersurface equations are threefolds built as intersections of singularities of type $A$:
\begin{equation}
\label{eq:Ajcollision}
    uv=a^nb^k Q(a,b)^j,
\end{equation}
with $Q(a,b)$ an irreducible polynomial of $a,b$. Equation \eqref{eq:Ajcollision} displays a line of $A_{n-1}$ singularities\footnote{ \ As we will see there might be some $\mathbb P^1 \subset \varepsilon^{-1}(0)$  over which $\tilde{X}_{E_6}^{x}$ is, locally, a trivial fibration of the single-center Taub-NUT space $\left\{(x,y,z) \in \mathbb C^3| xy = z \right\}$ (isomorphic to $\mathbb C^2$). We denote the single-center Taub-NUT space as $A_0$. There might also be non-compact curves, outside $\varepsilon^{-1}(0)$, over which the $\tilde{X}_{E_6}^{x}$ looks like an $A_0$ trivial fibration.} on $u=v=a =0$, a line of $A_{k-1}$ on $u= v = b =0$  and a line of $A_{j-1}$ singularities on\footnote{ \ This is true because in all the considered examples $Q(a,b)=0$ is a non-singular curve inside $\mathbb C^2 \ni (a,b)$.} $Q(a,b) = 0$.
Hence, labeling the nodes of $E_6$ as in Figure \ref{fig:E6dynkin} we can read off the $r_i$ by looking at the exponents of the irreducible factors of the r.h.s.\ of the equations appearing in the second column of \cref{tab:localmodelsE6z}.
  \begin{figure}[H]
    \centering
    \scalebox{0.8}{
    \begin{tikzpicture}
        \draw[thick] (0,0) circle (0.65);
        \node at (0,0) {$\alpha_3$};
        \draw[thick] (0.7,0)--(1.3,0);
        \draw[thick] (2,0) circle (0.65);
        \node at (2,0) {$\alpha_4$};
        \draw[thick] (-0.7,0)--(-1.3,0);
        \draw[thick] (-2,0) circle (0.65);
        \node at (-2,0) {$\alpha_2$};
        \draw[thick] (0,0.7)--(0,1.3);
        \draw[thick] (0,2) circle (0.65);
        \node at (0,2) {$\alpha_6$};
        \draw[thick] (2.7,0)--(3.3,0);
        \draw[thick] (4,0) circle (0.65);
        \node at (4,0) {$\alpha_5$};
        \draw[thick] (-2.7,0)--(-3.3,0);
         \draw[thick] (-4,0) circle (0.65);
        \node at (-4,0) {$\alpha_1$};
        \end{tikzpicture}}
    \caption{Labeling convention of the $E_6$ Dynkin diagram.}
    \label{fig:E6dynkin}
    \end{figure}
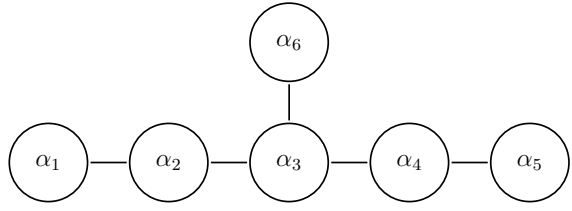

\renewcommand{\arraystretch}{1.2}
\begin{table}[H]
\begin{center}

  \begin{tabular}{c|c}
    \hline
  $\boldsymbol{V_i}$ & \textbf{Hypersurface equation}   \\
    \hline
    \hline
         $U_0$ & $uv = \underbrace{b_0^3}_{\alpha_6}
   \underbrace{\big(\left(a_0+1\right)^2
   b_0+1\big)^4}_{\alpha_2}
   \underbrace{\big(\left(a_0+1\right)^2
   b_0+2\big)}_{\mathcal{D}}$\\
     \hline
$U_1$ & $u v =\underbrace{b_1^2}_{\alpha_1}\underbrace{\left(a_1^{3}b_1^{2}-1\right)^3}_{\alpha_6}\underbrace{\left(a_1^{3}b_1^{2}+1\right)}_{\mathcal{D}}$
    \\
    \hline
     $U_2$ & $u v = \underbrace{a_2^2}_{\alpha_1}\underbrace{b_2^4}_{\alpha_2} \underbrace{\left(a_2^{2}b_2-1\right)^3}_{\alpha_6}\underbrace{\left(a_2^{2}b_2+1\right)}_{\mathcal{D}}$ \\
     \hline
     $U_3$ & $u v = \underbrace{a_3^4}_{\alpha_2} \underbrace{b_3^6}_{\alpha_3} \underbrace{(a_3-1)^3}_{\alpha_6}\underbrace{\left(a_3+1\right)}_{\mathcal{D}}$ \\
     \hline
     $U_4$ & $u v = \underbrace{a_4^6}_{\alpha_3}\underbrace{b_4^4}_{\alpha_4} \underbrace{(b_4-1)^3}_{\alpha_6}\underbrace{\left(b_4+1\right)}_{\mathcal{D}}$ \\
     \hline
     $U_5$ & $u v = \underbrace{a_5^4}_{\alpha_4}\underbrace{b_5^2}_{\alpha_5} \underbrace{\left(a_5b_5^2-1\right)^3}_{\alpha_6}\underbrace{\left(a_5b_5^2-1\right)}_{\mathcal{D}}$ \\ 
     \hline
      $U_6$ & $u v = \underbrace{a_6^2}_{\alpha_5}\underbrace{\left(a_6^2b_6^3-1\right)^3}_{\alpha_6}\underbrace{\left(a_6^2b_6^3+1\right)}_{\mathcal{D}}$ \\ 
     \hline
  \end{tabular}
    \caption{Hypersurface equations in resolution charts for $\tilde{X}_{E_6}^x$. Subscripts refer to the Dynkin nodes pertaining to each factor. $\mathcal{D}$ indicates a non-compact curve.}
     \label{tab:localmodelsE6z}
       \end{center}
\end{table}
\renewcommand{\arraystretch}{1}

In the $U_0$ chart the subset $\mathcal{D}: u = v = (a_0+1)^2b_0+2 = 0$ (that also appears in the other charts) is a \textit{non-compact} line\footnote{ \ We can see from \eqref{eq:transitionfunctionsAk} that  the $\mathbb{C} \ni b_1$ is not compactified into a $\mathbb P^1$ and hence the associated divisor supports a flavor symmetry.} of single-center Taub-NUT $A_0$. 
Summing up, from \cref{tab:localmodelsE6z} we can conclude that: 
\begin{itemize}
    \item over each point of the nodes $\alpha_1$ and $\alpha_5$ of $\varepsilon^{-1}(0)$ there is a $A_1$ singularity;
    \item over each point of the node $\alpha_6$ there is a $A_2$ singularity;
    \item over each point of the nodes $\alpha_2$ and $\alpha_4$ there is a $A_3$ singularity;
    \item over each point of the node $\alpha_3$ there is a $A_5$ singularity.
\end{itemize}
    We are now left to resolve the $A_{r_i}$ singularities over each $\mathbb P^1_i \subset \varepsilon^{-1}(0)$. We can do this via the resolution map $\pi$ (we will be more specific momentarily):
    \begin{equation}
\label{eq: complete resolution}
\pi: \tilde{\tilde{X}}_{E_6}^{x} \to \tilde{X}_{E_6}^{x}. \end{equation}
The irreducible components of $\pi^{-1}(\mathbb P^1_{i})$ will then be compact ruled surfaces intersecting according to the $A_{r_i}$ Dynkin diagram.\footnote{ \ More precisely,  every node of the $A_{r_i}$ diagram represents a compact surface ruled over $\mathbb P^1_i$ and every edge represents some $\mathbb P^1$s being the intersection of two of these compact surfaces.} This configuration is peculiar of our specific crepant resolution $\pi \circ \varepsilon$ but,  using the geometric engineering dictionary, we can already exploit it to determine some aspects of the quiver gauge theory describing the considered resolution of $X_{E_6}^{x}$:
  \begin{enumerate}
      \item each irreducible component $\mathbb P^1_i$ of $\varepsilon^{-1}(0)$ corresponds to a \textit{gauge} node of the quiver, the points $q_{ij}$ correspond to bifundamental hypers between the gauge nodes;
      \item the gauge group associated to $\mathbb P^1_i$ is $\mathfrak{su}(r_i+1)$;
      \item flavor nodes are associated to the component $\mathcal{D}$ appearing in \cref{tab:localmodelsE6z}. Since $\mathcal{D}$ is a non-compact line of $A_0$ inside $\tilde{X}_{E_6}^x$, they correspond to one flavor.
  \end{enumerate}
We note that, if we forget about the Chern-Simons levels and the $\theta$ angles, we already have all the information to write down the quiver in Figure \ref{fig:E6yfigure} that captures the relevant dynamics of M-theory on $X_{E_6}^x$. The following steps of the singularity resolution, in terms of the quiver, will just add this finer information.   
\item 
We now want to understand the enhanced singularities at the intersection points $q_{ij} = \mathbb P^1_{i} \cap \mathbb P^1_{j}$. For example, on $q_{12}$ we have a collision of $A_1$ and $A_3$ lines of singularities that we model as
\begin{equation}\label{toricmodel}
    uv = a^2b^4.
\end{equation}
 Equation \eqref{toricmodel} is the zero locus of a binomial, hence it is toric; in \cref{sec:appendixtoriclocalmodels} we will recall the proof that shows that \eqref{toricmodel} is Calabi-Yau. We depict the corresponding toric diagram in Figure \ref{fig:polytopeexample}.
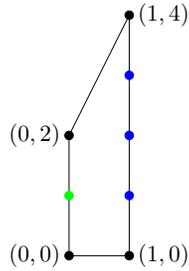
\begin{figure}[H]
    \centering
    \scalebox{0.8}{
    \begin{tikzpicture}
        \draw  (0,0)--(0,-2)--(1,-2)--(1,2)--cycle; 
        \filldraw (0,0) circle (2pt);
        \filldraw (0,-2) circle (2pt);
        \filldraw[blue] (1,-1) circle (2pt);
        \filldraw[green] (0,-1) circle (2pt);
        \filldraw[blue] (1,0) circle (2pt);
        \filldraw[blue] (1,1) circle (2pt);
        \filldraw (1,2) circle (2pt);
        \filldraw (1,-2) circle (2pt);
        \node[left] at (0,0) {\small $(0,2)$};
        \node[left] at (0,-2) {\small $(0,0)$};
        \node[right] at (1,-2) {\small $(1,0)$};
        \node[right] at (1,2) {\small $(1,4)$};
    \end{tikzpicture}}
    \caption{Example of the toric planar diagram for \eqref{toricmodel}.}
    \label{fig:polytopeexample}
\end{figure}
Consequently, we can use the toric model defined in Figure \ref{fig:polytopeexample} to describe a neighborhood of $q_{12}$ inside $\tilde{X}_{E_6}^{x}$.
We can proceed analogously to associate a local toric model to each  $q_{ij}$. 

\item The global threefold $\tilde{X}_{E_6}^x$ is obtained by gluing the local toric models as in Figure \ref{fig:E6nodots}. In Figure \ref{fig:E6nodots} we are employing the usual orthonormal system in $\mathbb{R}^3$ spanned by the vectors $\hat{i},\hat{j},\hat{k}$. Notice that the fact that the gluing cannot happen in $\mathbb{R}^2$ indicates that the threefold is globally non-toric. 
  Below each vertical line we indicate the corresponding $A_{r_i}$ singularity (according to Table \ref{tab:localmodelsE6z}).
  
  We remark that the dashed lines that we draw in Figure \ref{fig:E6nodots} are just a graphical tool to help the reader to reconstruct the gluing of the local toric models, but \textit{do not} represent compact curves and divisors. In the same way, only the vertices of the solid edges in Figure \ref{fig:E6nodots} are associated with divisors of $\tilde{X}_{E_6}^x$. To describe $\tilde{X}_{\mathfrak{g}}^{\bullet}$, one can pictorially construct a 3d analogue of a toric diagram, delimited by the solid edges and obtained by gluing the two diagrams appearing in Figure \ref{fig:E6nodots} along the ``$A_5$'' vertical line. We remark that this intuitive 3d picture \textit{must not} be formally interpreted as a proper toric diagram.  
\item The resolution map $\pi$ that we have introduced in \eqref{eq: complete resolution} corresponds to a triangulation of the toric diagram in Figure \ref{fig:E6nodots}.
The colored nodes represent compact divisors, and the gluing between the local model on the right and the one on the left happens by identifying (with non-toric maps) the red nodes, as well as the two black nodes directly above and below them.

\begin{figure}[H]
\begin{subfigure}{0.5\textwidth}
\centering
    \scalebox{0.8}{
    \begin{tikzpicture}
        \draw (2,4)--(3,2)--(4,0)--(4,-2)--(0,-2)--(0,0)--(1,2)--(2,4);
        \draw[dashed] (2,4)--(2,3)--(2,2)--(2,1)--(2,0)--(2,-1)--(2,-2);
        \draw[dashed] (3,2)--(3,1)--(3,0)--(3,-1)--(3,-2);
        \draw[dashed]  (4,0)--(4,-1)--(4,-2);
        \draw[dashed] (1,2)--(1,1)--(1,0)--(1,-1)--(1,-2);
        \draw[dashed]  (0,0)--(0,-1)--(0,-2);
        %
        %
        %
        \filldraw (0,0) circle (2pt);
        \filldraw (0,-1) circle (2pt);
        \filldraw (0,-2) circle (2pt);
        \filldraw (1,2) circle (2pt);
        \filldraw (1,1) circle (2pt);
        \filldraw (1,0) circle (2pt);
        \filldraw (1,-1) circle (2pt);
        \filldraw (1,-2) circle (2pt);
        \filldraw (2,4) circle (2pt);
        \filldraw (2,3) circle (2pt);
        \filldraw (2,2) circle (2pt);
        \filldraw (2,1) circle (2pt);
        \filldraw (2,0) circle (2pt);
        \filldraw (2,-1) circle (2pt);
        \filldraw (2,-2) circle (2pt);
        \filldraw (3,2) circle (2pt);
        \filldraw (3,1) circle (2pt);
        \filldraw (3,0) circle (2pt);
        \filldraw (3,-1) circle (2pt);
        \filldraw (3,-2) circle (2pt);
        \filldraw (4,0) circle (2pt);
        \filldraw (4,-1) circle (2pt);
        \filldraw (4,-2) circle (2pt);
        %
        %
        \node[below] at (0,-2) {\small $A_1$};
        \node[below] at (1,-2) {\small $A_3$};
        \node[below] at (2,-2) {\small $A_5$};
        \node[below] at (3,-2) {\small $A_3$};
        \node[below] at (4,-2) {\small $A_1$};
    \end{tikzpicture}}
    \caption*{$\hat{i}\hat{j}$ plane}
    \end{subfigure}
    \begin{subfigure}{0.5\textwidth}
    \centering
           \scalebox{0.8}{
    \begin{tikzpicture}
        \draw[dashed] (0,4)--(0,-2);  
        \draw (0,4)--(1,1)--(1,-2)--(-1,-2)--(-1,-1)--(0,4);
        %
        %
        \filldraw (-1,-1) circle (2pt);
        \filldraw (-1,-2) circle (2pt);
        \filldraw (0,4) circle (2pt);
        \filldraw (0,3) circle (2pt);
        \filldraw (0,2) circle (2pt);
        \filldraw (0,1) circle (2pt);
        \filldraw (0,0) circle (2pt);
        \filldraw (0,-1) circle (2pt);
        \filldraw (0,-2) circle (2pt);
        \filldraw (1,1) circle (2pt);
        \filldraw (1,0) circle (2pt);
        \filldraw (1,-1) circle (2pt);
        \filldraw (1,-2) circle (2pt);
        %
        %
        %
        %
        \node[below] at (-1,-2) {\small $A_0$};
        \node[below] at (0,-2) {\small $A_5$};
        \node[below] at (1,-2) {\small $A_2$};
    \end{tikzpicture}}
    \caption*{$\hat{j}\hat{k}$ plane}
    \end{subfigure}
     \caption{Example of the gluing of local toric models in the $\tilde{X}_{E_6}^{x}$ case. }
    \label{fig:E6nodots}
\end{figure}
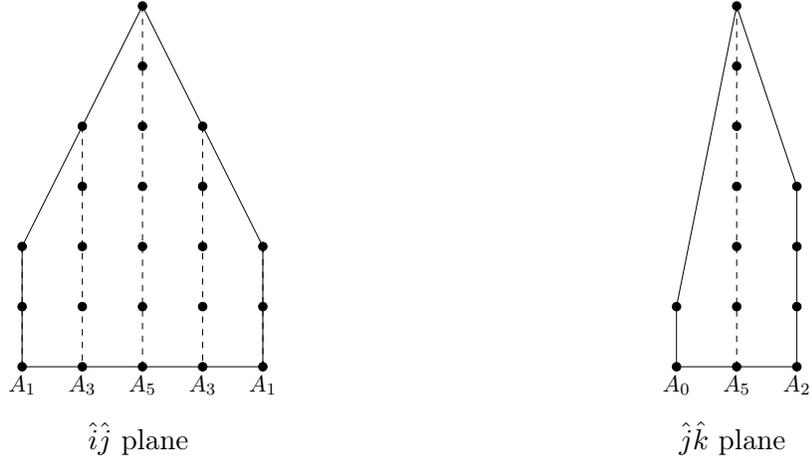

\begin{figure}[H]
\begin{subfigure}{0.5\textwidth}
\centering
    \scalebox{0.8}{
    \begin{tikzpicture}
        \draw (2,4)--(3,2)--(4,0)--(4,-2)--(0,-2)--(0,0)--(1,2)--(2,4);
        \draw (2,4)--(2,3)--(2,2)--(2,1)--(2,0)--(2,-1)--(2,-2);
        \draw (3,2)--(3,1)--(3,0)--(3,-1)--(3,-2);
        \draw (4,0)--(4,-1)--(4,-2);
        \draw (1,2)--(1,1)--(1,0)--(1,-1)--(1,-2);
        \draw (0,0)--(0,-1)--(0,-2);
        \draw (0,-1)--(1,-1);
        \draw (0,-1)--(1,0);
        \draw (0,-1)--(1,1);
        \draw (4,-1)--(3,-1);
        \draw (4,-1)--(3,0);
        \draw (4,-1)--(3,1);
        \draw (3,1)--(4,0);
        \draw (1,1)--(0,0);
        \draw (3,-1)--(4,-2);
        \draw (1,-1)--(0,-2);
        \draw (2,-1)--(1,-2);
        \draw (2,-1)--(3,-2);
        \draw (2,3)--(1,2);
        \draw (2,3)--(3,2);
        \draw (1,-1)--(2,-1);
        \draw (1,-1)--(2,0);
        \draw (3,-1)--(2,-1);
        \draw (3,-1)--(2,0);
        \draw (1,1)--(2,2);
        \draw (1,1)--(2,3);
        \draw (3,1)--(2,2);
        \draw (3,1)--(2,3);
        \draw (1,0)--(2,0);
        \draw (1,0)--(2,1);
        \draw (1,0)--(2,2);
        \draw (3,0)--(2,0);
        \draw (3,0)--(2,1);
        \draw (3,0)--(2,2);
        %
        %
        %
        \filldraw (0,0) circle (2pt);
        \filldraw[green] (0,-1) circle (2pt);
        \filldraw (0,-2) circle (2pt);
        \filldraw (1,2) circle (2pt);
        \filldraw[blue] (1,1) circle (2pt);
        \filldraw[blue] (1,0) circle (2pt);
        \filldraw[blue] (1,-1) circle (2pt);
        \filldraw (1,-2) circle (2pt);
        \filldraw (2,4) circle (2pt);
        \filldraw[red] (2,3) circle (2pt);
        \filldraw[red] (2,2) circle (2pt);
        \filldraw[red] (2,1) circle (2pt);
        \filldraw[red] (2,0) circle (2pt);
        \filldraw[red] (2,-1) circle (2pt);
        \filldraw (2,-2) circle (2pt);
        \filldraw (3,2) circle (2pt);
        \filldraw[blue] (3,1) circle (2pt);
        \filldraw[blue] (3,0) circle (2pt);
        \filldraw[blue] (3,-1) circle (2pt);
        \filldraw (3,-2) circle (2pt);
        \filldraw (4,0) circle (2pt);
        \filldraw[green] (4,-1) circle (2pt);
        \filldraw (4,-2) circle (2pt);
        %
        %
        \node[below] at (0,-2) {\small $A_1$};
        \node[below] at (1,-2) {\small $A_3$};
        \node[below] at (2,-2) {\small $A_5$};
        \node[below] at (3,-2) {\small $A_3$};
        \node[below] at (4,-2) {\small $A_1$};
    \end{tikzpicture}}
    \caption*{$\hat{i}\hat{j}$ plane}
    \end{subfigure}
    \begin{subfigure}{0.5\textwidth}
    \centering
           \scalebox{0.8}{
   \begin{tikzpicture}
        \draw (0,4)--(0,-2);  
        \draw (0,4)--(1,1)--(1,-2)--(-1,-2)--(-1,-1)--(0,4);
        \draw (-1,-1)--(0,-2);
        \draw (-1,-1)--(0,-1);
        \draw (-1,-1)--(0,0);
        \draw (-1,-1)--(0,1);
        \draw (-1,-1)--(0,2);
        \draw (-1,-1)--(0,3);
        \draw (0,-1)--(1,-2);
        \draw (1,-1)--(0,-1);
        \draw (1,-1)--(0,0);
        \draw (1,-1)--(0,1);
        \draw (1,0)--(0,1);
        \draw (1,0)--(0,2);
        \draw (1,0)--(0,3);
        \draw (1,1)--(0,3);
        %
        %
        \filldraw (-1,-1) circle (2pt);
        \filldraw (-1,-2) circle (2pt);
        \filldraw (0,4) circle (2pt);
        \filldraw[red] (0,3) circle (2pt);
        \filldraw[red] (0,2) circle (2pt);
        \filldraw[red] (0,1) circle (2pt);
        \filldraw[red] (0,0) circle (2pt);
        \filldraw[red] (0,-1) circle (2pt);
        \filldraw (0,-2) circle (2pt);
        \filldraw (1,1) circle (2pt);
        \filldraw[violet] (1,0) circle (2pt);
        \filldraw[violet] (1,-1) circle (2pt);
        \filldraw (1,-2) circle (2pt);
        %
        %
        %
        %
        \node[below] at (-1,-2) {\small $A_0$};
        \node[below] at (0,-2) {\small $A_5$};
        \node[below] at (1,-2) {\small $A_2$};
    \end{tikzpicture}}
    \caption*{$\hat{j}\hat{k}$ plane}
    \end{subfigure}
     \caption{Example of the gluing of local toric models in the $\tilde{\tilde{X}}_{E_6}^{x}$ case. }
    \label{fig:E6example}
\end{figure}

We  note that each vertical line of the toric diagram is invariant w.r.t. a reflection around its center.\footnote{ \ By this we mean that, if we call $n_1, n_2$  the number of edges terminating in two vertices $p_1, p_2$ exchanged by the reflection, then $n_1 = n_2$.} As we will see in \cref{sec:parity}, this will  correspond to the parity invariance of the quiver gauge theory associated with the resolution in Figure \ref{fig:E6example}, and will set to zero all the Chern-Simons levels.
\end{enumerate}

At this point, \textit{we have completely resolved the starting singular threefold $X_{E_6}^x$.}\\

\indent We can tentatively draw a picture of the resulting preimage of the origin $\pi^{-1}(\varepsilon^{-1}(0))$ through the blowup, representing compact divisors with the colors with which they appear in Figure \ref{fig:E6example}. One can further compute the normal bundle of the intersection curves following the usual conventions (summarized in appendix \ref{sec:appendixtoriclocalmodels}), checking that they are such that the resolved threefold respects the Calabi-Yau condition, and that the global gluing can only happen in a three-dimensional space (as already shown in Figure \ref{fig:E6example}), signaling that the global threefold is \textit{not} toric. In the Figure we write down the normal bundles of the curves: each number expresses the self-intersection of the curve inside the corresponding compact divisor. Notice that the curved black lines in \ref{fig:E6intuitive} correspond to the edges joining diagonally a red node with a black node in the rightmost diagram of Figure \ref{fig:E6example}. Furthermore, Figure \ref{fig:E6intuitive} is manifestly invariant with respect to a $\mathbb{Z}/2$ reflection along a straight horizontal line that cuts the picture horizontally in half. As we will see momentarily, this has repercussions in terms of the Chern-Simons levels of the gauge nodes of the quiver corresponding to the picture (we identify this isometry with the action of parity).

\begin{figure}[H]
    \centering
\begin{tikzpicture}[scale=0.52]
   \draw[very thick,black,fill=candypink] (-1,2)--(-4,2)--(-4,6)--(4,6)--(4,2)--(-1,2);
   \draw[very thick,gray,opacity=0.3] (4,6)--(4,2);
    \draw[very thick,black,fill=candypink] (-1,2)--(-4,2)--(-5,0)--(-4,-2) --(4,-2)--(5,0)--(4,2)--(-1,2) (0,-2) to [out=90,in=-30] (-1.8,1);
    \draw[very thick,black,fill=candypink] (-1,-2)--(-4,-2)--(-5,-4)--(-4,-6) --(4,-6)--(5,-4)--(4,-2)--(-1,-2) (0,-6) to [out=90,in=-30] (-1.8,-3);
    \draw[very thick,black,fill=cyan(process)] (-4,-2)--(-5,0)--(-8,0)--(-9,-2)--(-8,-4)--(-5,-4)--(-4,-2);
    \draw[very thick,black,fill=cyan(process)] (4,-2)--(5,0)--(8,0)--(9,-2)--(8,-4)--(5,-4)--(4,-2);
    \draw[very thick,black,fill=cyan(process)] (-5,8)--(-4,6)--(-4,2)--(-5,0)--(-8,0)--(-8,8)--(-5,8);
    %
    %
    \draw[very thick,black,fill=candypink] (-1,6)--(-4,6)--(-5,8)--(-4,10)--(4,10)--(5,8)--(4,6)--cycle
    (0,10) to [out=270,in=30] (-1.8,7);
    \draw[very thick,black,fill=candypink] (-1,10)--(-4,10)--(-5,12)--(-4,14)--(4,14)--(5,12)--(4,10)--cycle
    (0,14) to [out=270,in=30] (-1.8,11);
    \draw[very thick,black,fill=violet] (2,3.2)--(-1,4)--(-2,6)--(-2,10)--(-1,12)--(2,10.8)--(2,3.2);
    \draw[very thick,black,fill=violet] (2,-5.2)--(-1,-4)--(-2,-2)--(-2,2)--(-1,4)--(2,3.2)--(2,-5.2);
    \draw[very thick,black,fill=cyan(process)] (-4,10)--(-5,12)--(-8,12)--(-9,10)--(-8,8)--(-5,8)--(-4,10);
    \draw[very thick,black,fill=cyan(process)] (4,10)--(5,12)--(8,12)--(9,10)--(8,8)--(5,8)--(4,10);
    \draw[very thick,black,fill=cyan(process)] (5,8)--(4,6)--(4,2)--(5,0)--(8,0)--(8,8)--(5,8);
    \draw[very thick,black,fill=celadon] (-8,4)--(-8,8)--(-9,10)--(-12,10)--(-12,-2)--(-9,-2)--(-8,0)--(-8,4);
    \draw[very thick,black,fill=celadon] (8,4)--(8,8)--(9,10)--(12,10)--(12,-2)--(9,-2)--(8,0)--(8,4);
    \node[above] at (-10.5,-2) {\small -1};
    \node[below] at (-10.5,10) {\small -1};
    \node[right] at (-12,4) {\small 0};
    \node[left] at (-8,4) {\small -2};
    \node[right] at (-8,4) {\small 0};
    \node[right] at (-4,4) {\small 0};
    \node[left] at (-4,4) {\small -2};
    \node[above] at (-6.5,0) {\small -1};
    \node[below] at (-6.5,0) {\small -1};
    \node[above] at (-6.5,-4) {\small -1};
    \node[below] at (-6.5,8) {\small -1};
    \node[above] at (-6.5,8) {\small -1};
    \node[below] at (-6.5,12) {\small -1};
    \node at (-8,9) {\small -1};
     \node at (-8,11) {\small -1};
    \node at (-9,9) {\small -1};
    \node at (-4,11) {\small -1};
    \node at (-5,11) {\small -1};
    \node at (-4,13) {\small -1};
    \node at (-4,9) {\small -1};
    \node at (-5,9) {\small -1};
    \node at (-4,7) {\small -1};
    \node at (-5,7) {\small -1};
    \node at (-8,-1) {\small -1};
    \node at (-8,-3) {\small -1};
    \node at (-9,-1) {\small -1};
    \node at (-4,-1) {\small -1};
    \node at (-5,-1) {\small -1};
    \node at (-4,1) {\small -1};
    \node at (-5,1) {\small -1};
    \node at (-4,-3) {\small -1};
    \node at (-5,-3) {\small -1};
    \node at (-4,-5) {\small -1};
\end{tikzpicture}
\caption{Pictorial representation of the compact divisors in the preimage of the origin after a specific resolution of the $X_{E_6}^x$ threefold. Curves in this figure correspond to $\mathbb P^1$s, and the numbers represent their self-intersection inside the respective compact divisor. The plaquettes correspond to compact divisors in the threefold geometry. We stress here that all of these are ruled surfaces, and hence we expect a gauge theory interpretation.}
    \label{fig:E6intuitive}
\end{figure}
In the next section, we flesh out the physical interpretation of the compact divisors configurations such as the one depicted in Figure \ref{fig:E6intuitive}.

\subsection{5d gauge theory phase}\label{sec:exampleE6quiver}

The information that we have fleshed out in the previous section allows us to build the low-energy quiver gauge theory counterpart of the 5d SCFT arising from M-theory on $X_{E_6}^x$. All the surfaces in Figure \ref{fig:E6intuitive} are ruled surfaces that have fibration by curves of self-intersection zero. Due to the Calabi-Yau condition, these have a normal bundle $\mathcal O(-2)\oplus \mathcal O(0)$. Hence we can interpret them as resolutions of singularities in the vertical direction. The green divisor has a single -2 curve, and hence shrinking it down one obtains an $A_1$ singularity. The blue surfaces are ruled by a curve that is not irreducible and splits into three distinct -2 curves, corresponding to the three irreducible compact divisors in the geometry. Shrinking that down one obtains an $A_3$ singularity. Similarly, the red divisors give an $A_5$ singularity and the purple ones an $A_2$. Notice that at the intersections these singularities enhance, hence using the Katz-Vafa method we can read off the corresponding bifundamentals from Adjoint Higgsing \cite{KatzVafa}. For instance at the collision of the blue and the red divisors in Figure \ref{fig:E6intuitive} we see an enhanced singularity consisting of 9 curves, and shrinking those down we obtain an $A_9$ singularity corresponding to the Adjoint Higgsing $SU(10) \to SU(4) \times SU(6) \times U(1)$ from which we read off the matter content $$\mathbf{99} = \mathbf{15} \oplus \mathbf{35} \oplus (\overline{\mathbf{4}},\mathbf{6}) \oplus (\mathbf{4},\overline{\mathbf{6}}) \oplus \mathbf{(1,1)}$$
that indeed gives rise to a full bifundamental. Shrinking all the curves in the vertical direction down to zero size, all the divisors go to zero volume, but we do not reach the conformal point. Instead we are left with a collection of rational curves, call them, $\mathbb P^1)_{b,i}$, where $b$ stands for the base of the ruling, that intersect along an $E_6$ diagram. These curves have a dual interpretation in geometry as the relevant deformation of the 5d SCFT to the 5d gauge theory phase. In particular, $\text{vol}(\mathbb P^1_{b,i}) \sim 1/g^2_i$ where $g^2_i$ are the gauge couplings for the various gauge groups in the quiver. Sending these gauge couplings to infinity, amounts to sending the volumes of the $\mathbb P^1_{b,i}$ to zero, and flowing to the conformal point. 

Schematically, it suffices to look at the non-toric gluing depicted in Figure \ref{fig:E6example}: compact divisors correspond to gauge nodes, and non-compact ones to flavor nodes. We depict the quiver for $X_{E_6}^x$ in Figure \ref{fig:E6yfigure}, where we used the same color-code of Figure \ref{fig:E6intuitive} to associate the compact divisors to the corresponding quiver nodes. \\
       \begin{figure}[H]
    \centering
    \scalebox{0.8}{
    \begin{tikzpicture}
        \filldraw[color=red!60, fill=red!5,thick] (0,0) circle (0.65);
        \node at (0,0) {\small$\mathfrak{su}(6)_0$};
        \draw[thick] (0.7,0)--(1.3,0);
        \draw[color=blue!60, fill=blue!5,thick] (2,0) circle (0.65);
        \node at (2,0) {\small$\mathfrak{su}(4)_0$};
        \draw[thick] (-0.7,0)--(-1.3,0);
        \draw[color=blue!60, fill=blue!5,thick] (-2,0) circle (0.65);
        \node at (-2,0) {\small$\mathfrak{su}(4)_0$};
        \draw[thick] (0,0.7)--(0,1.3);
        \draw[color=violet!60, fill=violet!5,thick] (0,2) circle (0.65);
        \node at (0,2) {\small$\mathfrak{su}(3)_0$};
        \draw[thick] (2.7,0)--(3.3,0);
        \draw[color=green!60, fill=green!5,thick] (4,0) circle (0.65);
        \node at (4,0) {\small$\mathfrak{su}(2)$};
        \draw[thick] (-2.7,0)--(-3.3,0);
         \draw[color=green!60, fill=green!5,thick] (-4,0) circle (0.65);
        \node at (-4,0) {\small$\mathfrak{su}(2)$};
        \draw[thick] (0,-0.7)--(0,-1.4);
        \node at (0,-2) {1};
               \draw[thick] (-0.5,-1.5)--(-0.5,-2.5)--(0.5,-2.5)--(0.5,-1.5)--cycle ;
        \end{tikzpicture}}
    \caption{Low-energy quiver gauge theory for M-theory on $X_{E_6}^x$.}
    \label{fig:E6yfigure}
    \end{figure}
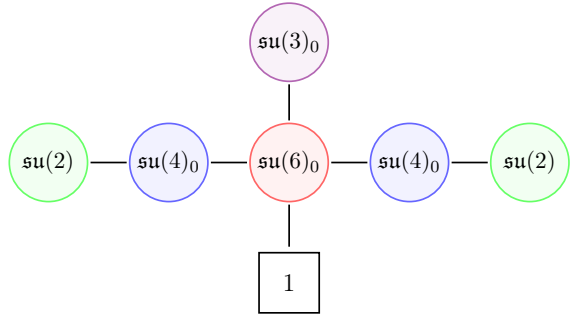
    
The Chern-Simons levels of each gauge node have been gleaned thanks to the fact that there exists a triangulation of the global resolved threefold $\tilde{\tilde{X}}_{E_6}^x$ (such as the one in Figure \ref{fig:E6example}) that is invariant under a $\mathbb{Z}/2$ reflection: we claim that such symmetry would \textit{not} be preserved if any of the Chern-Simons levels were different from zero. We present a rigorous proof of this fact in section \ref{sec:parity}.

\subsection{UV flavor symmetry}\label{sec:flavorUV}
We would now like to extract the flavor symmetry of the SCFT that arises as the UV completion of the quiver gauge theory presented above. To this end, we can employ the technique developed in \cite{Tachikawa,Yonekura}. For the full details of the construction we refer to the original works. In our context, the central result that is needed to extract the sought-after UV flavor symmetry can be stated as follows, provided that all the Chern-Simons levels and $\theta$ angle parameters of the nodes are equal to zero, as is in our cases: 
\begin{itemize}
    \item consider a quiver whose gauge nodes are joined according to the Dynkin diagram of $\mathfrak g$. Let $A_{ij}$ be the Cartan matrix of $\mathfrak g$, let $f_i$ be the number of flavors connected to $i$-th gauge node (with $i = 1,\ldots \text{rank}(\mathfrak g)$) and $\mathfrak{su}(r_i+1)$ be its gauge algebra.
    \item If: 
    \begin{equation}\label{cartan condition}
        \sum_j A_{ij} (r_j+1) - f_i  = 0 \quad \quad \forall i,
    \end{equation}
    then the total UV flavor symmetry $\mathcal{G}_{UV}^{\text{TOT}}$ contains \textit{at least} the following non-abelian factor:
    \begin{equation}
      \mathcal{G}_{UV}^{TOT} \supseteq  \mathcal{G}_{UV} = \mathfrak{g} \times \mathfrak{g}.
    \end{equation}
\end{itemize}
In our specific quiver depicted in Figure \ref{fig:E6yfigure}, a trivial computation shows that the UV flavor symmetry contains at least the factor:
\begin{equation}
     \mathcal{G}_{UV} = E_6 \times E_6.
\end{equation}
Thus, we have proven that:\\

\indent \textit{The 5d SCFT geometrically engineered by M-theory on the threefold $X_{E_6}^x$ exhibits at least a $E_6\times E_6$ flavor symmetry}.\\

\indent As we will see in section \ref{sec:quivers}, this theory is \textit{not} a direct descendant of the 6d $E_6\times E_6$ conformal matter theory. Thus, it genuinely earns the name of \textit{five-dimensional conformal matter}.\\
\indent In the next section we recap the general resolution procedure for the threefolds $X_{\mathfrak{g}}^{\bullet}$, with an eye more focused on the mathematical rigour.

\section{5d $(\mathfrak g,\mathfrak g)_{\bullet}$ conformal matter and  \texorpdfstring{$\boldsymbol{X}_{\mathfrak{g}}^{\bullet}$}{} singularities}\label{sec:generalrecipe}

In this section we discuss the general properties of the 5d $(\mathfrak g,\mathfrak g)_{\bullet}$ conformal matter theories. We begin in Section \ref{sec:basechange} by illustrating the mathematical details of the resolution of the threefolds $X_{\mathfrak{g}}^{\bullet}$ defined in \eqref{eq:threefoldsingdef}, giving a recipe that generalizes the procedure outlined in \cref{sec:exampleE6}. In Section \ref{sec:toriclocalmodels} we generalize the discussion of Section \ref{sec:exampleE6quiver} to describe the quivers arising along the 5d Coulomb branch of the 5d SCFT in terms of resolutions obtained by glueing toric local models. The resulting theories are described in Section \ref{sec:parity} where in particular we argue that they are parity invariant. The precise structure of the models is outlined in \ref{sec:quivers}, where we give all the 5d quiver gauge theory phases for all the 5d conformal matter theories of $(\mathfrak g,\mathfrak g)_{\bullet}$ type, we determine their flavor symmetries and we comment on their interplay with the 6d conformal matter theories.

\subsection{A partial resolution via base change}\label{sec:basechange}

For convenience, let us recall the definition of the $X_{\mathfrak{g}}^{\bullet}$ singularities:\footnote{ \ Recall that, for $\mathfrak{g}=A$, we only consider the $\bullet = x$ case.}
\begin{equation}
\label{eq:threefoldsingdefv2}
 X_{\mathfrak{g}}^{\bullet}\equiv \Set{ P_{\mathfrak{g}}(x,y,z)\rvert_{\bullet = uv} = 0} \subset \mathbb C^4,
\end{equation}
as well as their singular locus:
\begin{equation}
\label{eq:singularlocusv2}
\text{Sing}(X_{\mathfrak{g}}^{\bullet}) = \Set{ (x,y,z)\rvert_{\bullet = uv} = 0 } \subset \mathbb C^4.
\end{equation}

Since outside the stratum $u = v = 0$ of \eqref{eq:singularlocusv2}, we have a trivial family Du Val singularities, then any two crepant resolutions of $X_{\mathfrak{g}}^{\bullet}$ agree outside the origin $0 \in X_{\mathfrak{g}}^{\bullet}\subset \mathbb C^4$.
Since the singularities $X_{\mathfrak{g}}^{\bullet}$ are obtained by base-change of $Y_{\mathfrak{g}}$, a  way to obtain a partial crepant\footnote{ \ The fact that $\varepsilon$ is a crepant resolution is guaranteed by a result of \cite{reid1983minimal}.} resolution $\re{\varepsilon}:\tilde{X}_{\mathfrak g}^{\bullet} \to X_{\mathfrak{g}}^{\bullet}$ is to lift the base-change\footnote{ \ We remark that the subscript $\mathbb C^1$ under the $\times$ symbol in \eqref{eq:basechangeII} is a common notation for the algebraic notion of ``fibered product'' \cite{eisenbud2000geometry}.} that we used to obtain \eqref{eq:threefoldsingdefv2} from the \textit{singular} Du Val $Y_{\mathfrak g}$ to the \textit{resolved} Du Val $\tilde{Y}_{\mathfrak{g}}$:
\label{sec:resbasechange}
\begin{center}
\begin{equation}
\label{eq:basechangeII}
    \begin{tikzcd}
       \mathbb C^2 \underset{\mathbb C^1}{\times} \tilde Y_{\mathfrak{g}}\arrow{rr}\arrow[bend right=20,swap]{ddr}\arrow{dr}{\re{\varepsilon}} &&\tilde Y_{\mathfrak{g}} \arrow{d} \\   &\arrow{r}\mathbb C^2 \underset{\mathbb C^1}{\times} Y_{\mathfrak{g}}\arrow{d} \arrow{r} &  Y_{\mathfrak{g}}\arrow{d}{\pi_\bullet} & (x,y,z)\arrow[mapsto]{d} \\
        &\mathbb C^2 \arrow{r} & \mathbb C^1 & \bullet\\[-0.8cm]
        &(u,v)\arrow[mapsto]{r} & uv &
    \end{tikzcd}
    \end{equation}
\end{center}
where $\bullet\in\Set{x,y,z}$. As for \eqref{E6yexample}, the base-change \eqref{eq:basechangeII} boils down to substituting\footnote{ \ The map $\varepsilon$ is well-defined and unique by the universal property of the fibered-product.} the variable $\bullet$ with $uv$ in all the resolution maps of appendix \ref{sec:duvalres}. We saw that, substituting $\bullet = uv$ in the equation of  $Y_{\mathfrak{g}}$ corresponds to ``sticking'' a $Y_{\mathfrak{g}}$ singularity over each point $(u,v) \neq (0,0)$ of $uv = 0 \subset \mathbb C^2$. Similarly, \eqref{eq:basechangeII} corresponds to sticking a resolved $\tilde{Y}_{\mathfrak{g}}$ over each point $(u,v) \neq (0,0)$ of $uv = 0 \in \mathbb C^2$. This completely resolves the singularities of $X_{\mathfrak{g}}^{\bullet}$ outside of the origin $0 \in X_{\mathfrak{g}}^{\bullet} \subset \mathbb C^4$ (as presented in \eqref{eq:threefoldsingdef}). The threefold $\tilde{X}_{\mathfrak g}^{\bullet}$ has residual singularities supported on the preimage $\varepsilon^{-1}(0)$ of the origin $0 \in X_{\mathfrak{g}}^{\bullet}\subset \mathbb C^4$ via the partial resolution map $\varepsilon$.\\

\indent  One can explicitly check that $\varepsilon^{-1}(0)$ consists of a bunch of $\mathbb P^1$'s intersecting according to the Dynkin diagram\footnote{ \ This is true in the sense that different irreducible components of $\varepsilon^{-1}(0)$ intersect according to the Dynkin diagram of $\mathfrak{g}$ but we \textit{do not} extract from the Dynkin diagram the datum on the self-intersection of the $\mathbb P^1$'s entering in $\varepsilon^{-1}(0).$ The reason is that $\varepsilon^{-1}(0)$ is not included in an obvious way in any surface inside $\tilde{X}_{\mathfrak g}^{\bullet}$; furthermore many irreducible components of $\varepsilon^{-1}(0)$ are singular in $\tilde{X}_{\mathfrak g}^{\bullet}$, making their normal bundles ill-defined.} of $\mathfrak g$. Besides, the threefold $\tilde{X}_{\mathfrak g}^{\bullet}$ has a singularity of type $A$ over each $\mathbb P^1$ appearing in $\varepsilon^{-1}(0)$. We call $r_i \geq 0$ the rank of the $A$-type singularity\footnote{ \ As we will see there might be, according the base-change that we choose, some $\mathbb P^1 \subset \varepsilon^{-1}(0)$ over which $\tilde{X}_{\mathfrak g}^{\bullet}$ is, locally, a trivial fibration of the single-center Taub-NUT space $\left\{(x,y,z) \in \mathbb C^3| xy = z \right\}$ (isomorphic to $\mathbb C^2$). We denote the single-center Taub-NUT space as as $A_0$.}  appearing on the $i$-th irreducible component $\mathbb P^1_i \subset \varepsilon^{-1}(0)$.\\ 
\indent In general, these residual singularities can be resolved via a sequence of blow-ups with reduced centers:
\begin{equation}
\label{eq:piresolution}
\begin{tikzcd}[row sep = tiny]
 \pi: \tilde{\tilde{X}}^{\bullet}_{\mathfrak g} \arrow{r}&   \tilde{X}_{\mathfrak g}^{\bullet}
\end{tikzcd}.
\end{equation}
The fiber $\pi^{-1}(\mathbb P^1_i)$ can be described as follows:
\begin{itemize}
    \item on $p \in \mathbb P^1_i \setminus \cup_{j\neq i} \mathbb P^1_j$, the fiber $\varepsilon^{-1}(p)$ is the union of $r_i$ ``vertical'' $\mathbb P^1$'s coming from the resolution of the $Y_{A_{r_i}}$ singularity;
    \item at the intersection points $q_{ij} \equiv \mathbb P^1_i \cap \mathbb P^1_j$ the fiber $\varepsilon^{-1}(q_{ij})$ consists of $r_i+r_j+1$ ``vertical'' $\mathbb P^1$'s intersecting according to the Dynkin diagram of $A_{r_i+r_j+1}$\footnote{ \ This is precisely the same enhancement observed in \cite{KatzVafa}.}.
\end{itemize}
We can then conclude that $\pi^{-1}(\mathbb P^1_i)$ is isomorphic to $r_i$ (non-necessarily geometrically \cite{beauville_1996}) ruled compact surfaces, i.e.\ Hirzebruch surfaces $\mathbb  F_{n_j}$, with $j = 1, \ldots, r_i$, blown up at some points (producing multiple-fibers). The ruling of each of these surfaces is the restriction of $\pi$ and the reducible fibers (containing more than one $\mathbb P^1$) are contracted over the intersection points $q_{ij}$. The irreducible components of $\pi^{-1}(\mathbb P^1_{i})$ intersect according to the $A_{r_i}$ Dynkin diagram, where every node represents a compact surface ruled over $\mathbb P^1_i$ and every edge represents a $\mathbb P^1$ being the intersection of two of these compact surfaces. We stress that this configuration is peculiar to our specific crepant resolution $\pi \circ \varepsilon$.\\ \indent To compute subtle geometrical aspects of these ruled surfaces we can use that, outside $q_{ij}$, nothing fancy happens and, for each irreducible component of $\pi^{-1}(\mathbb P^1_i)$, we do not have multiple fibers. To understand neighborhoods of $\pi^{-1}(q_{ij})$ we can use that they can be \textit{locally} described by toric geometry (despite the full geometry $\tilde{X}_{\mathfrak g}^{\bullet}$ being non-toric). We embark on the toric analysis of $\pi^{-1}(q_{ij})$ in the next section.

\subsection{The resolution of two transversal families of \texorpdfstring{$A$}{} singularities}
\label{sec:toriclocalmodels}
The neighbourhoods of the points $q_{ij} \in \tilde{X}_{\mathfrak{g}}^{\bullet}$ can be locally described with the following toric models, already considered in \cite{KatzVafa}:
\begin{equation}
\label{eq:Xhk}
X_{hk}=\Set{(a,b,u,v)\in\mathbb C^4|uv-a^hb^k=0},
\end{equation}
for $h,k \ge 0$, with $q_{ij}$ corresponding to the origin of \eqref{eq:Xhk}. Thus, using \eqref{eq:Xhk} to resolve all the neighborhoods of $q_{ij} 
 \in \tilde{X}_{\mathfrak{g}}^{\bullet}$ and gluing together these local constructions, we will be able to completely smooth out the initial threefold $X_{\mathfrak{g}}^{\bullet}$ and glean the structure of the compact divisors blown-up on top of the origin.
\begin{example}
    For $h=k=0$ $X_{hk}$ is trivially smooth. Besides, for $h=k=1$ we find the conifold singularity.
\end{example}
Notice that $X_{hk}$ is a toric variety with singularities of type $A_{k-1}$ on the $x$ axis and of type $A_{h-1}$ on the $y$ axis.

The fan $\Sigma_{hk}$ of $X_{hk}$ has a  unique maximal cone $\sigma_{\max}$ generated by the following primitive vectors:
\begin{equation} 
\label{eq:fanXhk}
 w_1=\begin{pmatrix}
     0 \\ 0\\ 1 
\end{pmatrix},
\   w_2=\begin{pmatrix}
         0 \\ h\\ 1 
\end{pmatrix},
\  w_3=\begin{pmatrix}
  1 \\ k\\ 1
\end{pmatrix},
\  w_4=\begin{pmatrix}
    1 \\ 0\\ 1 
\end{pmatrix}.
\end{equation}
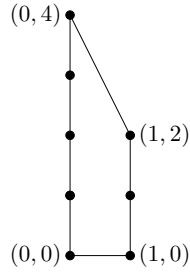
\begin{figure}[H]
    \centering
    \scalebox{0.8}{
    \begin{tikzpicture}
        \draw  (0,2)--(0,-2)--(1,-2)--(1,0)--cycle;
        \filldraw (0,2) circle (2pt);
        \filldraw (0,1) circle (2pt);
        \filldraw (0,0) circle (2pt);
        \filldraw (0,-1) circle (2pt);
        \filldraw (0,-2) circle (2pt);
        \filldraw (1,-1) circle (2pt);
        \filldraw (0,-1) circle (2pt);
        \filldraw (1,0) circle (2pt);
        \filldraw (1,-2) circle (2pt);
        \node[left] at (0,2) {\small $(0,4)$};
        \node[left] at (0,-2) {\small $(0,0)$};
        \node[right] at (1,-2) {\small $(1,0)$};
        \node[right] at (1,0) {\small $(1,2)$};
    \end{tikzpicture}}
    \caption{Example of the planar diagram $P_{hk}$  for $h=4$ and $k=2$.}
    \label{fig:polytope}
\end{figure}
$X_{hk}$ is Calabi-Yau (since all the $w_i$ lie on a plane) and its toric diagram $P_{hk}\subset \R^2$ is a lattice polygon  which encodes all the information about the geometry of $X_{hk}$ and its crepant resolutions (see \Cref{fig:polytope}). We now want to outline general features of these toric diagrams.  
We can readily check that the lattice polygon $P_{hk}$ has
\begin{itemize}
\item no internal points,
\item $k+1$ lattice points on the edge of $P_{hk}$ corresponding to the pair $(w_3,w_4)$,
\item  $h+1$ lattice points on the edge of $P_{hk}$ corresponding to the pair $(w_1,w_2)$.
\end{itemize}

Consequently, there are $h+k-2$ non-compact divisors after the resolution, and they are, on the $(u,v) \neq 0$ component of the singular locus, the exceptional loci of the resolution of the trivial families of $A_{h-1}$ and $A_{k-1}$ singularities.
 
Every crepant resolution of $X_{hk}$ corresponds to a triangulation of $P_{hk}$ which has, as set of vertices, exactly the marked points on $P_{hk}$. This triangulation produces $h+k-1$ $\mathbb{P}^1$'s on the origin $a=b=u=v=0$, arranged as a $A_{h+k-2}$ Dynkin diagram, as noticed in \cite{KatzVafa}.\\
\begin{example} For $h=2$ and $k =1$,  we have only three possible crepant resolutions which correspond to the triangulations of $P_{21}$ in \Cref{fig:Q01}.
\begin{figure}[H]
        \centering
        \begin{tikzpicture}
            \draw (0,0)--(1,0)--(1,0.5)--(0,1)--cycle;
        \filldraw (0,0) circle (1pt);
        \filldraw (1,0) circle (1pt);
        \filldraw (0,1) circle (1pt);
        \filldraw (1,0.5) circle (1pt);
        \filldraw (0,0.5) circle (1pt);
        \draw (0,0.5) --(1,0);
        \draw (0,1) --(1,0);
        \draw[dashed,<->] (1.5,0.5)--(2.5,0.5);
        \node[above] at  (2,0.5) {\small flop};
        \draw[dashed,<->] (4.5,0.5)--(5.5,0.5);
        \node[above] at  (5,0.5) {\small flop};
            \draw (2+1,0)--(3+1,0)--(3+1,0.5)--(2+1,1)--cycle;
        \filldraw (2+1,0) circle (1pt);
        \filldraw (3+1,0) circle (1pt);
        \filldraw (2+1,1) circle (1pt);
        \filldraw (3+1,0.5) circle (1pt);
        \filldraw (2+1,0.5) circle (1pt);
        \draw (2+1,0.5) --(3+1,0);
        \draw (2+1,0.5) --(3+1,0.5);
            \draw (6,0)--(7,0)--(7,0.5)--(6,1)--cycle;
        \filldraw (6,0) circle (1pt);
        \filldraw (7,0) circle (1pt);
        \filldraw (6,1) circle (1pt);
        \filldraw (7,0.5) circle (1pt);
        \filldraw (6,0.5) circle (1pt);
        \draw (6,0.5) --(7,0.5);
        \draw (6,0) --(7,0.5);
        \end{tikzpicture}
        \caption{Possible triangulations of $P_{21}$.}
        \label{fig:Q01}
    \end{figure}
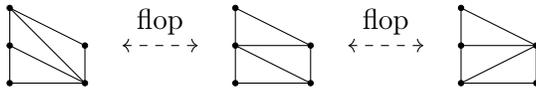
\end{example}

Notice that there always exists a triangulation of $P_{hk}$ that is invariant w.r.t.\ a $\mathbb Z/2$-transformation that reflects the edges of each vertical line $P_{hk}$ w.r.t.\ the center of the vertical line:

\begin{figure}[H]
\centering
    \scalebox{0.8}{
  \begin{tikzpicture}
        \draw  (0,2)--(0,-2)--(1,-2)--(1,0)--cycle;
        \draw  (0,0)--(1,-1);
        \draw  (0,1)--(1,-1);
        \draw  (0,-1)--(1,-1);
        \draw  (0,1)--(1,0);
        \draw  (0,-1)--(1,-2);
        \filldraw (0,2) circle (2pt);
        \filldraw (0,1) circle (2pt);
        \filldraw (0,0) circle (2pt);
        \filldraw (0,-1) circle (2pt);
        \filldraw (0,-2) circle (2pt);
        \filldraw (1,-1) circle (2pt);
        \filldraw (0,-1) circle (2pt);
        \filldraw (1,0) circle (2pt);
        \filldraw (1,-2) circle (2pt);
        \node[left] at (0,-2) {\small $A$};
        \node[right] at (1,-2) {\small $B$};
        \node[right] at (1,-1) {\small $C$};
        \node[right] at (1,0) {\small $D$};
        \node[left] at (0,2) {\small $E$};
        \node[left] at (0,1) {\small $F$};
        \node[left] at (0,0) {\small $G$};
        \node[left] at (0,-1) {\small $H$};
        \draw[thick,dashed,->] (3,0) to (6,0);
        \node[above] at (4.5,0.2) {$\mathbb{Z}/2$ reflection};
        \draw  (8,2)--(8,-2)--(9,-2)--(9,0)--cycle;
        \draw  (8,0)--(9,-1);
        \draw  (8,1)--(9,-1);
        \draw  (8,-1)--(9,-1);
        \draw  (8,1)--(9,0);
        \draw  (8,-1)--(9,-2);
        \filldraw (8,2) circle (2pt);
        \filldraw (8,1) circle (2pt);
        \filldraw (8,0) circle (2pt);
        \filldraw (8,-1) circle (2pt);
        \filldraw (8,-2) circle (2pt);
        \filldraw (9,-1) circle (2pt);
        \filldraw (8,-1) circle (2pt);
        \filldraw (9,0) circle (2pt);
        \filldraw (9,-2) circle (2pt);
        \node[left] at (8,-2) {\small $E$};
        \node[right] at (9,-2) {\small $D$};
        \node[right] at (9,-1) {\small $C$};
        \node[right] at (9,0) {\small $B$};
        \node[left] at (8,2) {\small $A$};
        \node[left] at (8,1) {\small $H$};
        \node[left] at (8,0) {\small $G$};
        \node[left] at (8,-1) {\small $F$};
    \end{tikzpicture}}
     \caption{$\mathbb Z/2$-transformation that leaves invariant the toric diagram. The transformation exchanges $A \leftrightarrow E$, $F \leftrightarrow H$, $B \leftrightarrow D$. Non-trivially, the number of edges ending on a pair of vertices exchanged by the $\mathbb Z/2$ action is the same.}
    \label{fig:Z2sym}
\end{figure}
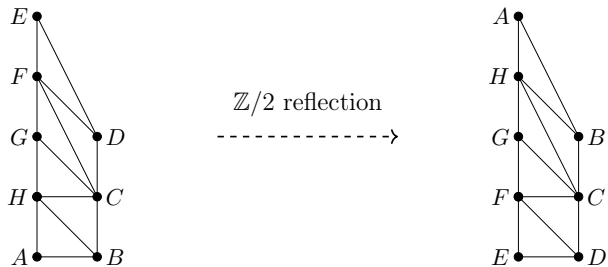

As we will see in the Section \ref{sec:parity}, we can employ this fact to extract the Chern-Simons levels of the nodes of the quiver that arises as the IR dual of the SCFT coming from M-theory on $X^{\bullet}_{\mathfrak{g}}$.\\ 
\indent We can now proceed to glue the local toric models to reproduce the full $\tilde{X}_{\mathfrak g}^{\bullet}$. The procedure is exactly analogous to the one explained in section \ref{sec:exampleE6}: each edge of the Dynkin diagram of the $\mathfrak g$ singularity corresponds to a certain $q_{ij}$, and hence to a certain local toric model. Then, we torically\footnote{ \ By this, we mean that the transition functions are defined by binomial relations between the coordinates.} glue together the toric diagrams associated to the $A_{\text{rank}(\mathfrak g)-1}$ subalgebra of $\mathfrak g$ (as in the leftmost part of Figures \ref{fig:E6nodots} and \ref{fig:E6example}). Finally we glue, to the trivalent node of $\mathfrak g$, the toric diagram of the local toric model associated to the intersection between the trivalent and the upmost node of the Dynkin diagram of $\mathfrak g$, for $\mathfrak{g}=D,E$ (corresponding to the rightmost part of Figures \ref{fig:E6nodots} and \ref{fig:E6example}). For the $\mathfrak{g}=A$ case the gluing is the obvious one.\\
\indent In the next section we recap the procedure to extract the physical data of the 5d SCFT from the geometric perspective outlined above.

\subsection{Quiver gauge-theory and parity invariance}\label{sec:parity}
M-theory reduced on a threefold $X_{\mathfrak{g}}^{\bullet}$ gives rise to a 5d $\mathcal{N}=1$ SCFT. Upon mass deformation such SCFT flows to weakly-coupled quiver gauge theory, with only $\mathfrak{su}$ nodes.\\
\indent In section \ref{sec:basechange} we introduced the partial resolution $\varepsilon: \tilde{X}_{\mathfrak g}^{\bullet} \to X_{\mathfrak g}^{\bullet}$ defined via the base-change \eqref{eq:basechangeII}. The five-dimensional degrees of freedom are trapped in the submanifold  $\varepsilon^{-1}(0) \subset \tilde{X}_{\mathfrak g}^{\bullet}$ contracted on $0 \in X_{\mathfrak{g}}^{\bullet}$ and can be described using a  five-dimensional quiver capturing the physics of M-theory on this resolved phase of the $X_{\mathfrak{g}}^{\bullet}$ singularity. In order to write down such quiver explicitly, we gather the following data:
\begin{itemize}
\item the quiver has $i = 1,..., \text{rank}(\mathfrak g)$  $\mathfrak{su}(r_i+1)$ gauge nodes, connected according to the $\mathfrak g$ Dynkin diagram. The numbers $r_i$ can be extracted directly by looking at the base-change of the resolution maps presented in \cref{sec:duvalres}, as concretely explained in the recipe of section \ref{sec:exampleE6}.
\item flavor groups (if present) are associated to non-compact families of $A_k$ singularities (or single-center Taub-NUT spaces) fibered over non-compact lines that intersect the outmost nodes of $\varepsilon^{-1}(0)$. 
\item The (minimal) UV flavor symmetry of the quiver gauge theory can be extracted with the techniques of \cite{Tachikawa,Yonekura}, as outlined in section \ref{sec:flavorUV}.
\item To completely specify the quiver we also need to compute the Chern-Simons levels of each node. To do so it is \textit{not} necessary to analyze the detailed geometry of the ruled surfaces $\pi^{-1}(\mathbb P^1_i)$ (with $\pi$ the blowup map presented in equation \ref{eq:piresolution}): as we have argued in section \ref{sec:flavorUV}, to have enhanced $\mathfrak{g}\times \mathfrak{g}$ at the SCFT point it is enough \cite{Yonekura} to check that the Chern-Simons levels vanish for each gauge node, or, equivalently, that the five-dimensional SCFT is parity invariant. If the SCFT comes from a $(p,q)$-web construction, the parity of the 5d SCFT is interpreted as the parity of the type IIB spacetime. In particular, the parity is equivalent to the invariance of the $(p,q)$-web w.r.t the reflection around the origin of the $(p,q)$-plane. $X_{\mathfrak{g}}^{\bullet}$ is not globally dual to a $(p,q)$-web, but the neighborhoods of the points $q_{ij}$ are described by the local models \eqref{eq:Xhk} that display the required $\mathbb Z/2$ reflection symmetry. One can check that these local $\mathbb Z/2$ reflections glue together into a $\mathbb Z/2$ automorphism that exchanges $u$ and $v$
\begin{equation}
\label{eq:parityautI}
\begin{tikzcd}[row sep = tiny]
  \tilde{X}_{\mathfrak{g}}^{\bullet} \arrow{r}{\Psi}&  \tilde{X}_{\mathfrak{g}}^{\bullet}  \\
(u,v,\ldots)\arrow[mapsto]{r} &(v,u,\ldots)
\end{tikzcd}
\end{equation}
At this point, $\Psi$ is an order-two automorphism of the partially resolved threefold $\tilde{X}_{\mathfrak{g}}^{\bullet}$. In order to promote it to the (complete) resolution $\tilde{\tilde{X}}_{\mathfrak{g}}^{\bullet}$ we need to perform the blowup sequence $\pi$, defined in \eqref{eq:piresolution}, in such a way as to treat the coordinates $u$ and $v$ ``democratically''.\footnote{ \ By this we mean that the centers of the various blowups  defining $\pi$ have to be subvarieties of $\tilde{X}_{\mathfrak{g}}$ preserved by \ $\Psi$.} At the level of the toric diagram, this amounts to choosing a triangulation of the local toric model \eqref{eq:Xhk} that is invariant w.r.t. to the reflection along the vertical axis of the toric diagram. 
We can then promote $\Psi$ to an automorphism of $\tilde{\tilde{X}}_{\mathfrak{g}}^{\bullet}$ ensuring that the five-dimensional SCFT is parity invariant and, hence, that the Chern-Simons levels of each $\mathfrak{su}$ node of the quiver are zero. It can be easily proven that a parity-invariant triangulation exists for \textit{all} our models $\tilde{\tilde{X}}_{\mathfrak{g}}^{\bullet}$, leading us to the conclusion that:

\textit{Given a threefold $X_{\mathfrak{g}}^{\bullet}$, there always exists a complete resolution  $\tilde{\tilde{X}}_{\mathfrak{g}}^{\bullet}$ that respects the $\mathbb{Z}/2$ reflection as defined above. Therefore, all the gauge nodes in the related 5d quiver gauge theory can be chosen to have vanishing Chern-Simons levels.}

\end{itemize}

In the next section we employ the above recipe to exhibit the quivers and the UV flavor symmetries for the threefolds $X_{\mathfrak{g}}^{\bullet}$, showing that they engineer 5d conformal matter.




\subsection{Low-energy quiver theories and their properties}\label{sec:quivers}
In this section, we explicitly show the low-energy gauge theory quivers corresponding to the geometries $X_{A_k}^{\bullet}$ (for every $k>1$ and $\bullet = x$, which is completely equivalent to $\bullet = y$), $X_{D_k}^{\bullet}$ (for every $k$), $X_{E_6}^{\bullet}$, $X_{E_7}^{\bullet}$, $X_{E_8}^{\bullet}$, with $\bullet=x,y,z$. We then use these quivers to extract the UV flavor symmetry of the considered conformal matter theories and to exclude the presence of electric one-form symmetries. Finally, we comment on the 6d uplift of the aforementioned theories. In Appendix \ref{sec:prepotential} we go through consistency checks that ensure that the quivers yield a sensible UV SCFT completion, computing their prepotential for an example, according to \cite{Intriligator_1997}.

\subsubsection{List of low-energy quiver gauge theories}\label{sec:quiverslist}
\indent \textit{All the nodes of the quivers in this section are of $\mathfrak{su}(n_i)$ type, with $n_i$ the number written inside each node. As proven in section \ref{sec:parity}, all nodes have vanishing Chern-Simons levels (or vanishing $\theta$ angle, for $\mathfrak{su}(2)$ nodes)}.\\
\indent We note that, due to the presence of fundamental hypermultiplets, there is no one-form electric symmetry for the considered theories. The fact that the electric one-form symmetry is explicitly broken by the presence of fundamental matter also excludes magnetic two-form symmetries. 

\noindent\begin{minipage}{0.1\textwidth}
    $\boldsymbol{\bullet X_{A_{2j+1}}^x}$
    \end{minipage}
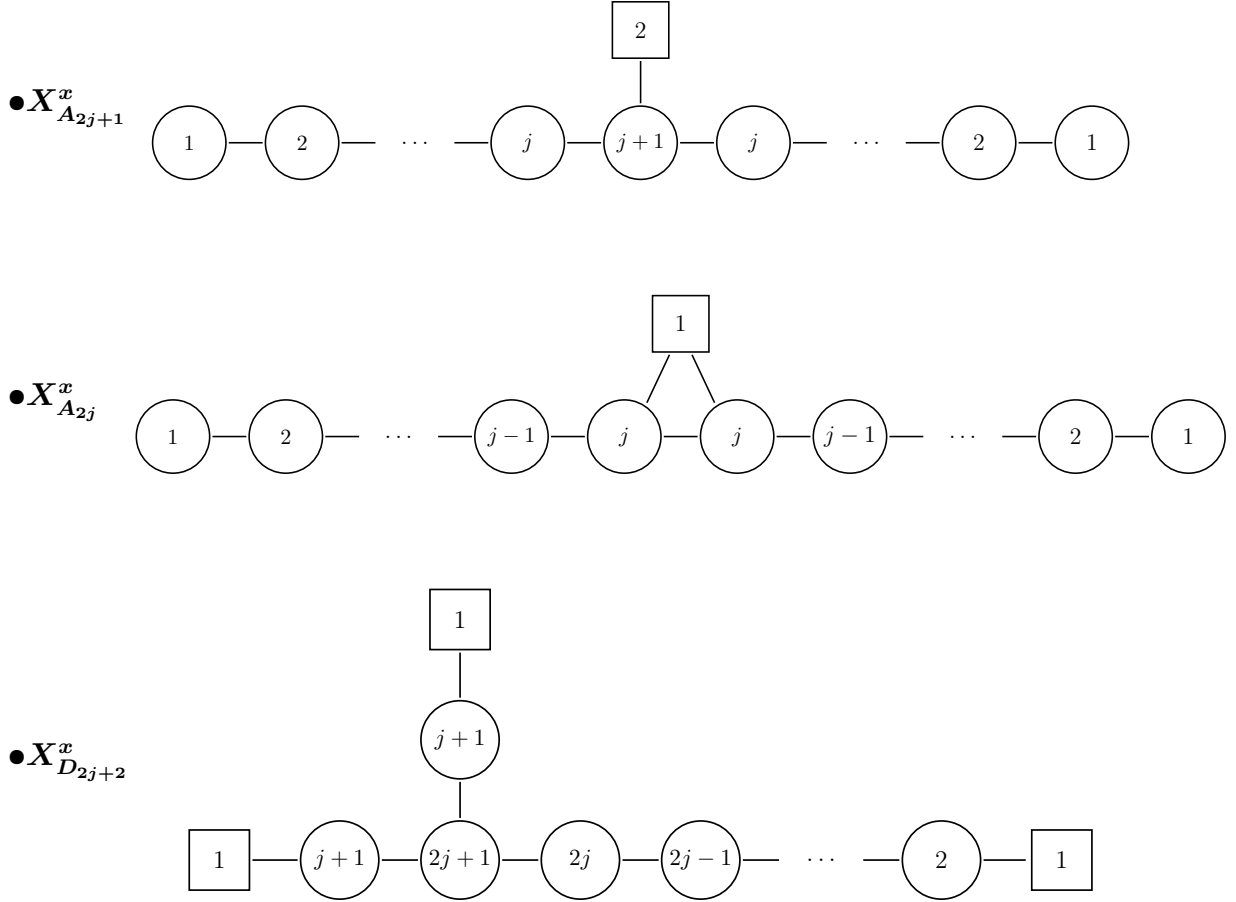
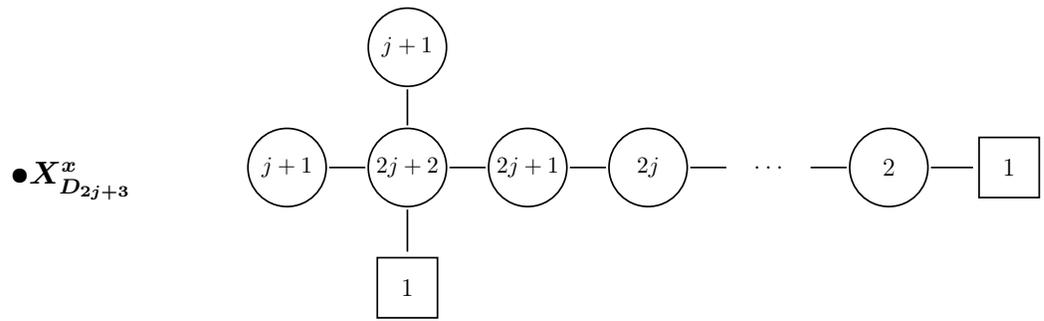
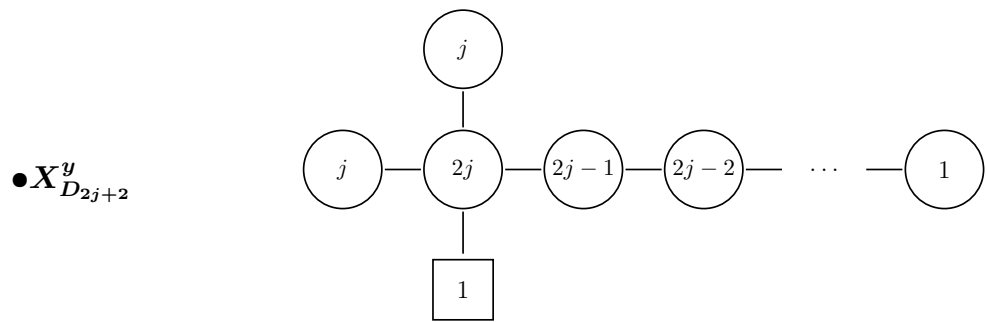
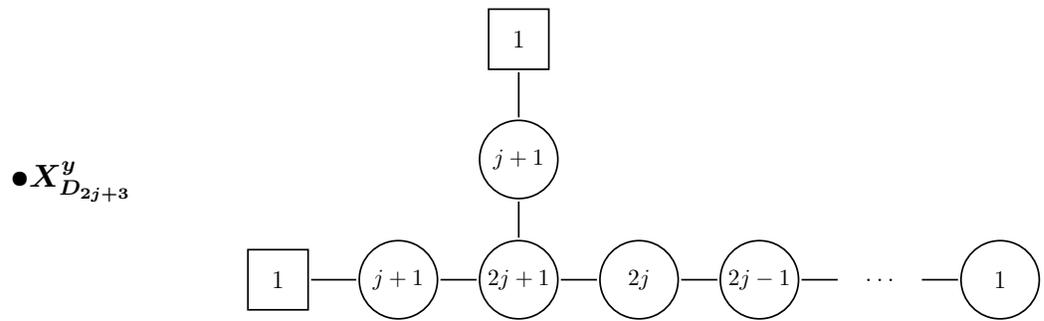
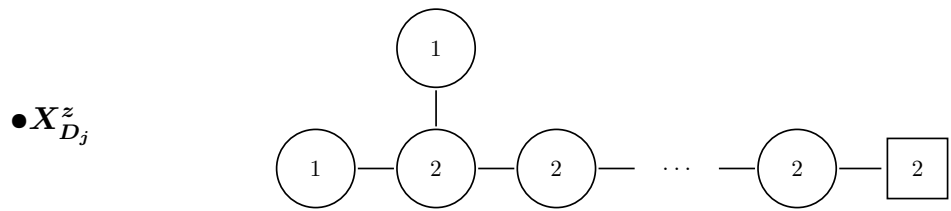
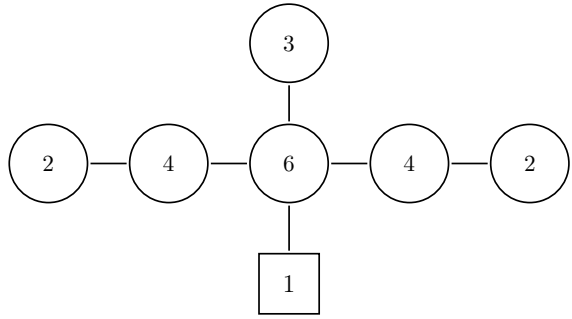
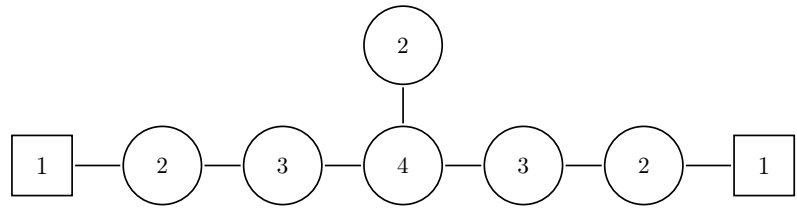
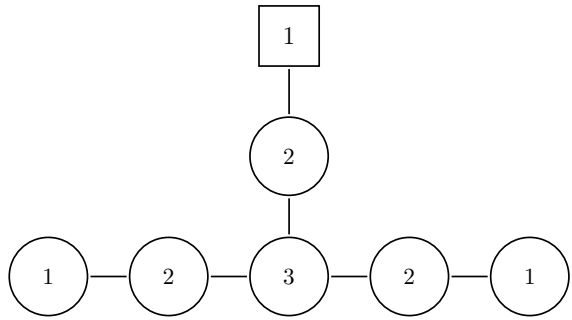
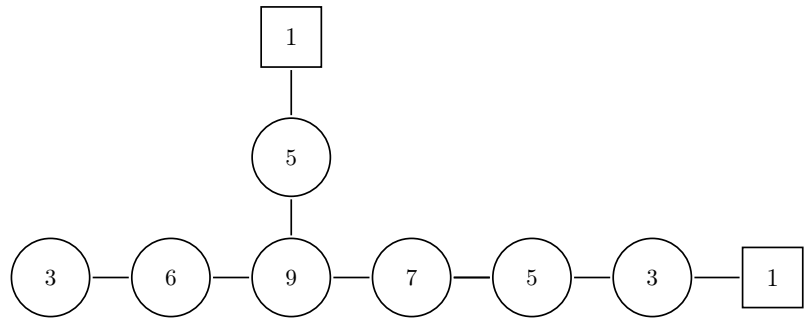
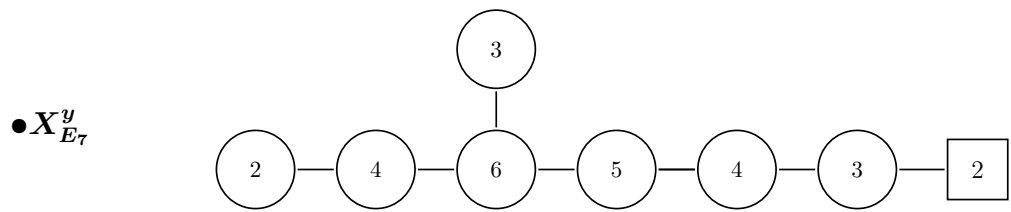
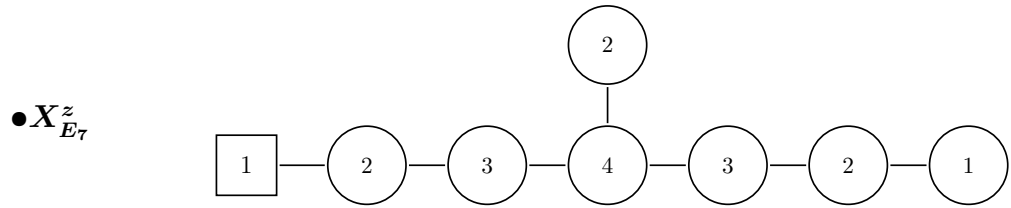
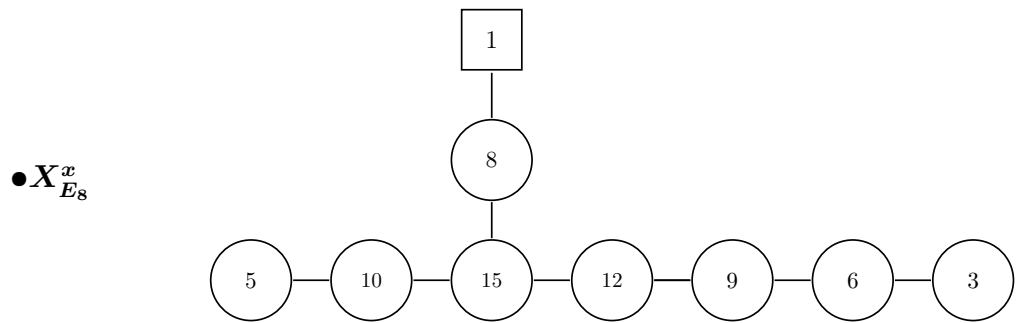
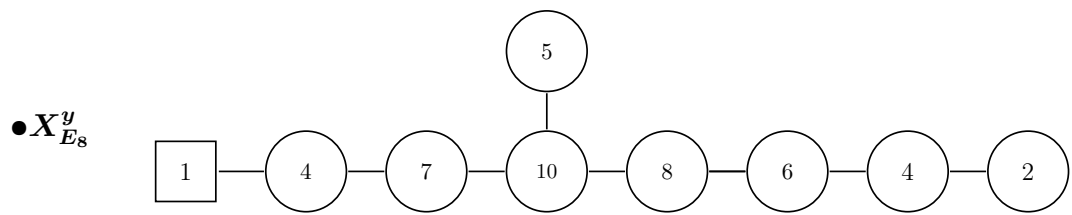
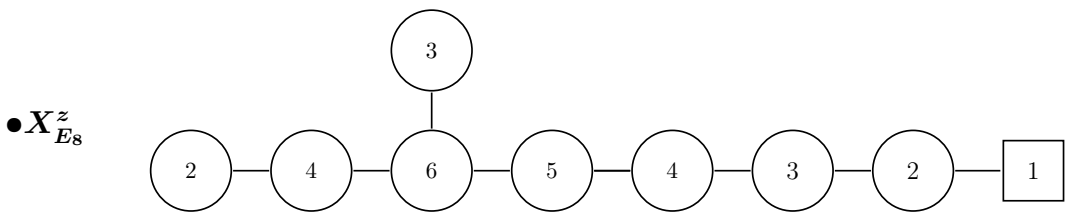
\begin{minipage}{0.86\textwidth}
       \begin{figure}[H]
    \centering
    \scalebox{0.75}{
}
    \caption*{}
    \label{Akxdisparifig}
    \end{figure}
    \end{minipage}\\

\noindent\begin{minipage}{0.1\textwidth}
    $\boldsymbol{\bullet X_{A_{2j}}^x}$
    \end{minipage}
\begin{minipage}{0.86\textwidth}
       \begin{figure}[H]
    \centering
    \scalebox{0.75}{%
}
    \caption*{}
    \label{Akxparifig}
    \end{figure}
    \end{minipage}\\

\noindent\begin{minipage}{0.1\textwidth}
    $\boldsymbol{\bullet X_{D_{2j+2}}^x}$
    \end{minipage}
\begin{minipage}{0.86\textwidth}
       \begin{figure}[H]
    \centering
    \scalebox{0.8}{%
}
    \caption*{}
    \label{Dkxparifig}
    \end{figure}
    \end{minipage}\\

    \noindent\begin{minipage}{0.1\textwidth}
    $\boldsymbol{\bullet X_{D_{2j+3}}^x}$
    \end{minipage}
\begin{minipage}{0.86\textwidth}
       \begin{figure}[H]
    \centering
    \scalebox{0.8}{%
}
    \caption*{}
    \label{Dkxdisparifig}
    \end{figure}
    \end{minipage}\\

\noindent\begin{minipage}{0.1\textwidth}
    $\boldsymbol{\bullet X_{D_{2j+2}}^y}$
    \end{minipage}
\begin{minipage}{0.86\textwidth}
       \begin{figure}[H]
    \centering
    \scalebox{0.8}{%
}
    \caption*{}
    \label{Dkyparifig}
    \end{figure}
    \end{minipage}\\

\noindent\begin{minipage}{0.1\textwidth}
    $\boldsymbol{\bullet X_{D_{2j+3}}^y}$
    \end{minipage}
\begin{minipage}{0.86\textwidth}
       \begin{figure}[H]
    \centering
    \scalebox{0.8}{%
}
    \caption*{}
    \label{Dkydisparifig}
    \end{figure}
    \end{minipage}\\

\begin{minipage}{0.1\textwidth}
    $\boldsymbol{\bullet X_{D_j}^z}$
    \end{minipage}
\begin{minipage}{0.8\textwidth}
       \begin{figure}[H]
    \centering
    \scalebox{0.8}{
    %
}
    \caption*{}
    \label{D4zfig}
    \end{figure}
    \end{minipage}\\

  \begin{minipage}{0.1\textwidth}
    $\boldsymbol{\bullet X_{E_6}^x}$
    \end{minipage}
\begin{minipage}{0.8\textwidth}
       \begin{figure}[H]
    \centering
    \scalebox{0.8}{
    %
}
    \caption*{}
    \label{E6xfig}
    \end{figure}
    \end{minipage}\\

\begin{minipage}{0.1\textwidth}
    $\boldsymbol{\bullet X_{E_6}^y}$
    \end{minipage}
\begin{minipage}{0.8\textwidth}
       \begin{figure}[H]
    \centering
    \scalebox{0.8}{
    %
}
    \caption*{}
    \label{E6yfig}
    \end{figure}
    \end{minipage}\\

\begin{minipage}{0.1\textwidth}
    $\boldsymbol{\bullet X_{E_6}^z}$
    \end{minipage}
\begin{minipage}{0.8\textwidth}
       \begin{figure}[H]
    \centering
    \scalebox{0.8}{
    %
}
    \caption*{}
    \label{E6zfig}
    \end{figure}
    \end{minipage}\\

\begin{minipage}{0.1\textwidth}
    $\boldsymbol{\bullet X_{E_7}^x}$
    \end{minipage}
\begin{minipage}{0.8\textwidth}
       \begin{figure}[H]
    \centering
    \scalebox{0.8}{
    %
}
    \caption*{}
    \label{E7xfig}
    \end{figure}
    \end{minipage}\\

\begin{minipage}{0.1\textwidth}
    $\boldsymbol{\bullet X_{E_7}^y}$
    \end{minipage}
\begin{minipage}{0.8\textwidth}
       \begin{figure}[H]
    \centering
    \scalebox{0.8}{
    %
}
    \caption*{}
    \label{E7yfig}
    \end{figure}
    \end{minipage}\\

\begin{minipage}{0.1\textwidth}
    $\boldsymbol{\bullet X_{E_7}^z}$
    \end{minipage}
\begin{minipage}{0.8\textwidth}
       \begin{figure}[H]
    \centering
    \scalebox{0.8}{
    %
}
    \caption*{}
    \label{E7zfig}
    \end{figure}
    \end{minipage}\\

\begin{minipage}{0.1\textwidth}
    $\boldsymbol{\bullet X_{E_8}^x}$
    \end{minipage}
\begin{minipage}{0.8\textwidth}
       \begin{figure}[H]
    \centering
    \scalebox{0.8}{
    %
}
    \caption*{}
    \label{E8xfig}
    \end{figure}
    \end{minipage}\\

\begin{minipage}{0.1\textwidth}
    $\boldsymbol{\bullet X_{E_8}^y}$
    \end{minipage}
\begin{minipage}{0.8\textwidth}
       \begin{figure}[H]
    \centering
    \scalebox{0.8}{
    %
}
    \caption*{}
    \label{E8yfig}
\end{figure}
\end{minipage}\\

\begin{minipage}{0.1\textwidth}
    $\boldsymbol{\bullet X_{E_8}^z}$
    \end{minipage}
\begin{minipage}{0.8\textwidth}
       \begin{figure}[H]
    \centering
    \scalebox{0.8}{
    %
}
    \caption*{}
    \label{E8zfig}
\end{figure}
\end{minipage}\\

\subsubsection{UV flavor symmetry}\label{sec:flavor symmetry}

We would now like to extract the flavor symmetry of the SCFTs that arise as the UV completion of the quiver gauge theories presented above. To this end, we can employ the technique developed in \cite{Tachikawa,Yonekura} and recapped in section \ref{sec:flavorUV}.\\
\indent A straightforward computation shows that condition \eqref{cartan condition} is satisfied in \textit{all} the quivers listed in this section. In physical terms, this means that the UV flavor symmetry contains  as a subgroup at least the symmetry given by the direct product of two copies of the Dynkin diagram corresponding to the gauge nodes.\\
\indent Let us summarize this result in \cref{UV flavor table}, also writing down the total rank of the expected UV flavour symmetry, thus accounting also for abelian factors. The total rank can be readily computed by keeping track of one topological instanton symmetry $\mathfrak{u}(1)_T$ for each gauge node, and one factor for each hypermultiplet in the quiver, and a $\mathfrak{su}(n_i)$ contribution, with $i=1,\ldots,\nu$, for each of the flavor nodes, labelled by $n_i$ in the low-energy quiver descriptions of the previous section. Notice that theories with the same flavor rank are not equivalent, as they have different gauge nodes, as can be seen directly by looking at the quivers written in the previous section.\\
\indent Furthermore, it is immediate to observe that threefolds $X_{\mathfrak{g}}^{\bullet}$ with \textit{different} label $\bullet =x,y,z$, corresponding to \textit{different} low-energy 5d quiver gauge theories, contain in their UV flavor symmetry the \textit{same} factor $\mathfrak{g}\times\mathfrak{g}$. This implies that 5d conformal matter comes in different ``species'', labelled precisely by $\bullet$.\footnote{ \ Notice that the only interacting conformal matter theory for the $A$ cases comes from the choice $\bullet = x$ (which is equivalent to $\bullet = y$). As we have previously observed, one can formally define a ``5d conformal matter theory'' for $\bullet = z$, provided that one bears in mind that it is just a theory of bifundamental hypermultiplets.} This yields a striking difference with respect to 6d conformal matter, and produces a rich duality structure that we will examine in section \ref{sec:exceptional quivers}.\\
To summarize, it holds:\\

\indent \textit{The theories $X_{\mathfrak{g}}^{\bullet}$ enjoy a UV flavor symmetry which is at least:}\footnote{ \ Notice that the \textit{rank} of the total UV flavor symmetry cannot be higher than that of \eqref{total flavor}. This implies that non-abelian rank-preserving enhancements cannot be ruled out.}
\begin{equation}\label{total flavor}
    \mathcal{G}_{UV} = \mathfrak{g}\times \mathfrak{g} \times \mathfrak{su}(n_1)\times\cdots \times \mathfrak{su}(n_{\nu}) \times \mathfrak{u}(1)^{\nu-1},
\end{equation}
with $\nu$ the number of flavor nodes in the quiver description corresponding to $X_{\mathfrak{g}}^{\bullet}$. The appearance of the extra $\mathfrak{u}(1)$ factors was already pointed out by means of a careful tracking of flavor symmetry enhancement, from the field theoretic viewpoint, in \cite{Tachikawa:2015mha}.

\renewcommand{\arraystretch}{1.3}
\begin{table}[H]\centering
\begin{equation}
\begin{array}{|c|c|c|}
\hline 
 \textbf{Singularity} & \boldsymbol{\mathcal{G}_{UV}}  & \textbf{Rank of }\boldsymbol{\mathcal{G}_{UV}^{\text{TOT}}} \\
\hline
\hline
X_{A_{2j+1}}^x & A_{2j+1}\times A_{2j+1} \times \mathfrak{su}(2)& 4j+3\\
X_{A_{2j}}^x & A_{2j}\times A_{2j} & 4j\\
X_{D_{2j+2}}^x & D_{2j+2}\times D_{2j+2} \times \mathfrak{u}(1)^2& 4j+6\\
X_{D_{2j+3}}^x &  D_{2j+3} \times D_{2j+3} \times \mathfrak{u}(1)& 4j+7 \\
X_{D_{2j+2}}^y &  D_{2j+2}\times D_{2j+2} & 4j+4 \\
X_{D_{2j+3}}^y &  D_{2j+3}\times D_{2j+3}\times \mathfrak{u}(1) & 4j+7\\
X_{D_j}^z &  D_j\times D_j \times \mathfrak{su}(2) & 2j+1\\
X_{E_6}^x & E_6\times E_6 & 12\\
X_{E_6}^y & E_6\times E_6 \times \mathfrak{u}(1)& 13\\
X_{E_6}^z & E_6\times E_6 & 12\\
X_{E_7}^x & E_7\times E_7 \times \mathfrak{u}(1)& 15\\
X_{E_7}^y & E_7\times E_7\times \mathfrak{su}(2) & 15\\
X_{E_7}^z & E_7\times E_7 & 14\\
X_{E_8}^x & E_8\times E_8 & 16\\
X_{E_8}^y & E_8\times E_8 & 16\\
X_{E_8}^z & E_8\times E_8 & 16\\
\hline
\end{array}\nonumber
\end{equation}
\caption{UV flavor enhancement for $X_{\mathfrak{g}}^{\bullet}$. Quivers with the same flavor rank are not equivalent, as they have manifestly different rank of the gauge nodes, as can be seen by their presentation in the previous section.}\label{UV flavor table}
\end{table}

\subsubsection{Relation to 6d conformal matter}

Let us add a quick remark, that can be readily noticed by looking at the quivers depicted above.\\
\indent It can be seen that, if we take the quiver gauge theories corresponding to $X_{D_4}^y$, $X_{D_5}^y$, $X_{E_6}^z$, $X_{E_6}^y$, $X_{E_7}^z$, $X_{E_8}^z$ and gauge the flavor nodes, we get precisely the quivers written down by Tachikawa in \cite{Tachikawa}, corresponding to affine $E_6, E_6, E_6, E_7, E_7, E_8$ Dynkin diagrams, respectively. This corresponds to a symmetry enhancement in the UV:
\begin{equation}
    \begin{array}{cccc}
    \boldsymbol{X_{D_4}^y}: & \mathcal{G}_{UV} = D_4\times D_4 & \xrightarrow{\text{gauging}} & \widehat{E_6\times E_6}\\
    \boldsymbol{X_{D_5}^y}: & \mathcal{G}_{UV} = D_5\times D_5 & \xrightarrow{\text{gauging}} & \widehat{E_6\times E_6}\\
    \boldsymbol{X_{E_6}^z}: &\mathcal{G}_{UV} = E_6\times E_6 & \xrightarrow{\text{gauging}} & \widehat{E_6\times E_6}\\
    \boldsymbol{X_{E_6}^y}: & \mathcal{G}_{UV} = E_6\times E_6&\xrightarrow{\text{gauging}} &  \widehat{E_7\times E_7}\\
    \boldsymbol{X_{E_7}^z}: & \mathcal{G}_{UV} = E_7\times E_7&\xrightarrow{\text{gauging}} &  \widehat{E_7\times E_7}\\
    \boldsymbol{X_{E_8}^z}:  &\mathcal{G}_{UV} = E_8\times E_8 & \xrightarrow{\text{gauging}} & \widehat{E_8\times E_8}\\
    \end{array},
\end{equation}
where $\widehat{\mathfrak{g} \times \mathfrak{g}}$ denotes the flavor algebra $\hat{\mathfrak{g}} \times \hat{\mathfrak{g}}$ with the flavor symmetries associated to the affine nodes identified.\footnote{ \ In particular, we note that the rank of the flavor enhanced $\widehat{\mathfrak{g} \times \mathfrak{g}}$ flavor symmetry is $2 \text{rank}(\mathfrak{g}) + 1$.}
The appearance of affine flavor symmetries is decisively hinting at the fact that these theories have an intrinsically six-dimensional origin. In fact, the choices $y=uv$ for $D_4$, $D_5$ and $E_6$, and $z = uv$ for $E_6,E_7,E_8$, correspond to the ``F-theoretic base change'', as the resulting threefold is explicitly in the form of an elliptic fibration.\footnote{ \ We note that, as in the $X_{E_6}^y$ case, the threefold might not be in Weierstrass form, but still being an elliptic fibration over $y \in \mathbb C$.} Indeed we can consider the six-dimensional theories obtained from F-theory compactified on the threefolds $X_{D_4}^y$, $X_{D_5}^y$, $X_{E_6}^y$, $X_{E_6}^z$, $X_{E_7}^z$, $X_{E_8}^z$, that now are hypersurfaces in a weighted projective space. We can then reduce these theories on a circle transverse to the Calabi-Yau, obtaining 5d Kaluza-Klein theories. Decoupling the Kaluza-Klein modes we obtain the quivers in section \ref{sec:quivers} with gauged flavor nodes. In order to reconcile with our M-theory setup, we must hence remove the points at infinity of the projective space: this has the neat effect of decompactifying the divisors that ultimately correspond to the flavor nodes in the M-theory picture.\footnote{ \ In the $X^{z}_{\mathfrak{g}}$ cases with $\mathfrak{g}=E_6,E_7,E_8$, this operation simply amounts to ungauging the \textit{affine} node of $\mathfrak{g}$.}\\
\indent The theories that we have just described are nothing but the \textit{6d conformal matter} theories introduced in \cite{conformalmatter}, and whose 5d descendants were written down explicitly in \cite{Apruzzi:2019opn}.\footnote{ \ Our quivers and the ones of \cite{Apruzzi:2019opn} can be readily seen to agree, provided that one exchanges (after gauging the affine node) our $\mathfrak{su}(1)$ gauge nodes with a node with 2 flavors, as pointed out in \cite{Bergman:2014kza,Hayashi:2014hfa}.} We remark here that, instead, the other examples we presented in \cref{sec:quivers} possess a genuine $\mathfrak{g}\times \mathfrak{g}$ UV flavor symmetry in 5d, with $\mathfrak{g} = D_4,D_5,E_6,E_7,E_8$, and can be easily constructed via a M-theoretic approach, while lacking an immediately transparent F-theoretic interpretation. This is because, in that case, the base-change we performed in \eqref{eq:threefoldsingdef} is not compatible with the structure of elliptic fibration of $Y_{\mathfrak g}$. We can also give a somewhat more heuristic explanation as to why these other base changes do not admit an F-theory uplift: if we gauge the flavor nodes of their quivers, we simply do \textit{not} obtain the Dynkin diagram of a Lie algebra, be it finite or affine, but instead we produce some more complicated quiver.\\
\indent Nevertheless, it has been conjectured that all 5d $\mathcal{N}=1$ SCFTs descend from 6d $(1,0)$ theories reduced on a circle with holonomies \cite{Jefferson_2018}. Thus, also our examples should share such parentage. It would be interesting to apply the techniques of \cite{Bhardwaj:2019xeg} in the reverse direction, flopping curves inside the compact divisors (i.e. integrating \textit{in} BPS particles) and compactifying non-compact divisors (namely integrating \textit{in} BPS strings) in order to climb the descendant hierarchy to reach a candidate 5d KK marginal theory, which would naturally uplift to 6d.\footnote{\ We thank Sakura Sch\"afer-Nameki for suggesting this idea to us.}

\section{Exceptional linear quivers and 5d dualities}\label{sec:exceptional quivers}
After having shown that M-theory on  $X_{\mathfrak g}^{\bullet}$ engineers five-dimensional bifundamental $\mathfrak{g} \times \mathfrak{g}$ matter, the next step is to use this construction to geometrically engineer quivers with nodes of type $\mathfrak g$. Of course, we are particularly interested in the cases $\mathfrak{g}=D,E$, in which the bifundamental matter gets properly promoted to conformal matter (as well as the $A$ case with the choice $x = uv$). In this section, we will give the general rules to gauge together the conformal matter theories that we introduced in the previous sections. The most general quiver that we can obtain from conformal matter is a linear quiver, with decorated edges each being of type $(\mathfrak g,\mathfrak g)_{\bullet_i}$, with $i = 1,...,n_{\text{edges}}$. We start by considering linear quivers with conformal matter edges all of the same type, and then we gradually generalize to linear quivers with different kinds of conformal matter edges. 
\subsection{Linear generalized quiver with edges of the same type}
\label{sec:linearquiversI}
In this section, we are going to glue two conformal matter theories of the same type $(\mathfrak g,\mathfrak g)_{\bullet}$ via a gauging procedure. We will use this to build generalized linear quivers with all edges being conformal matter theories of type $\bullet$.\\
\indent In order to do so, we recall that each $\mathfrak g$ factor of the UV flavor symmetry corresponds, in $X_{\mathfrak{g}}^{\bullet}$, to one of the irreducible components of the lines of $Y_{\mathfrak g}$ singularities on $u v = 0 \subset \mathbb C^2$. Let $\mathfrak g_{u}$ (resp. $\mathfrak g_{v}$) be the flavor symmetry factor associated to the $v = 0$ line (resp. $u = 0$ line) of $Y_{\mathfrak{g}}$ singularities. To gauge $\mathfrak{g}_{u}$ we have to compactify the $u$ coordinate into a $\mathbb P^1$, leaving the $v$ direction non-compact. More precisely, we glue two affine spaces  $\mathbb C^2_{u,v}$ and $\mathbb C^2_{u',v'}$ to get the total space of a line bundle $\mathcal{O}_{\mathbb P^1}(\re{-n})$\footnote{ \ We write ``$n$'' in red to emphasize that it is the same that appears in equation \eqref{eq:gaugingflavor}.}, where $(u,u')$ span the base-space $\mathbb P^1$ and $(v,v')$ span the fiber. \\

\indent Let us concentrate, for ease of exposition, on the $\mathfrak{g}=E_6$ and $\bullet = x$ case (everything proceeds similarly for $\bullet = y$ or $\bullet = z$). At the level of threefolds we have two singularities $X_{E_6}^{x}\subset \mathbb C^4_{y,z,u,v}$  and $\left(X_{E_6}^{x}\right)'\subset \mathbb C^4_{y,z,u',v'}$, explicitly given by:
\begin{equation}\label{gaugingE6}
    \begin{cases}
        (uv)^2+y^3+z^4=0\\
        (u'v')^2+y^3+z^4=0\\
    \end{cases}.
\end{equation}
We then patch them together with the following transition function:
\begin{equation}
\label{eq:gaugingflavor}
\begin{tikzcd}[row sep = tiny]
 \chi: U \subset X_{E_6}^{x}\arrow{r}&U' \subset \left(X_{E_6}^{x}\right)'\\
(y,z,u,v)\arrow[mapsto]{r} &(y,z,u',v') = (y,z,\frac{1}{u}, u^{\re{n}}v).
\end{tikzcd}
\end{equation}
where $U$  and $U'$ are, respectively, the open subsets of $X_{E_6}^{x}$   and $\left(X_{E_6}^{x}\right)'$ with $u \neq 0$ and $u' \neq 0$. 
We note that, to consistently patch $X_{E_6}^{x}$ and $\left(X_{E_6}^{x}\right)'$, we want $u'v' = u^{n-1} v$ to coincide with $uv$, hence we have to pick $n=2$. Consequently, gauging the singular $E_6$ lines respectively parameterized by $u$ and $u'$ of two bifundamental $E_6\times E_6$ 5d matter amounts to partly compactifying $(u,v)$ to a resolved Du Val singularity of type $A_1$ (by gluing them with another $\mathbb C^2$ spanned by $(u',v')$).  The construction outlined in \eqref{eq:gaugingflavor} applies also to the cases $\bullet = y$ or $\bullet = z$, opportunely permuting the roles of $(x,y,z)$. Of course, the reasoning can be generalized to any $\mathfrak{g} = ADE$ in lieu of $E_6$. The system \eqref{gaugingE6} thus corresponds to the quiver:
 \begin{figure}[H]
 \centering
\scalemath{0.9}{\begin{tikzpicture}
        \draw[thick] (0,0) circle (0.7);
        \node at (0,0) {$E_6$};
        \draw[thick] (0.8,0)--(2.2,0);
        \draw[thick] (-0.8,0)--(-2.2,0);
        \node at (1.45,0.3) {\scalebox{0.8}{ $(E_6,E_6)_x$}};
        \node at (-1.45,0.3) {\scalebox{0.8}{$(E_6,E_6)_x$}};
        \draw[thick] (2.3,0.7)--(3.7,0.7)--(3.7,-0.7)--(2.3,-0.7)--cycle;
        \draw[thick] (-2.3,0.7)--(-3.7,0.7)--(-3.7,-0.7)--(-2.3,-0.7)--cycle;
        \node at (3,0) {$E_6$};
        \node at (-3,0) {$E_6$};
        \end{tikzpicture}}
    \caption{Quiver with gauge nodes of type $E_6$ and edges of type $(E_6,E_6)_x$.}
     \label{fig:dualE6}
\end{figure}
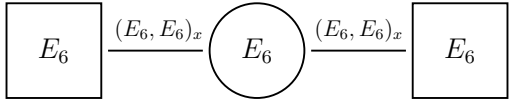  
where we indicate with $(E_6,E_6)_x$ the conformal matter lying at the collision between $E_6$ singularities corresponding to each edge of the quiver. We emphasize the following remark, which is a novel aspect of the theories we are considering in this section:\\

\textit{In a usual quiver gauge theory the nodes of the quiver are enough to identify the nature of the bifundamental matter associated to the edges linking quiver nodes. On the other hand, for the low-energy 5d quiver emerging from M-theory on \eqref{gaugingE6}, we need to specify the ``species'' $\bullet$ of the edges $(E_6,E_6)_\bullet$ of the quiver nodes (corresponding to different types of collision of $Y_{E_6}$ singularities in the threefold).\\}

As a result, the gauge rank of the quiver in Figure \ref{fig:dualE6} must\footnote{ \ We gleaned the data to compute \eqref{eq:rankcounting} by counting the number of the compact divisors in the threefold geometry, as explained in section \ref{sec:generalrecipe} and shown in the quivers of section \ref{sec:quivers}.} be computed as follows:
\begin{equation}
\label{eq:rankcounting}
    \text{rk}_{TOT} = \underbrace{\text{rk}(E_6)}_{\text{gauge node}}+\underbrace{2\text{rk}\big[(E_6,E_6)_x\big]}_{\text{edges}} = 6+2\cdot 15 = 36.
\end{equation}

\indent We can now contract the $u \in \mathbb P^1$ line\footnote{ \ The contraction map is $(u,v) \to (v u^2,v,uv)=(U,V,x)$.} producing a singularity, that we call $X_{E_6,1}^{x}$, defined by two equations in $\mathbb C^5$: 
\begin{equation}
    \label{eq: A1gsing}
    \begin{cases}
    x^2 +y^3+z^4= 0\\
    UV = x^2
    \end{cases},
\end{equation}
with coordinates $(x,y,z,U,V)$. The subscript ``1'' in $X_{E_6,1}^{x}$ refers to the second equation, which is a $A_1$ singularity. We have already given a quiver gauge theory interpretation of \eqref{eq: A1gsing}, in Figure \ref{fig:dualE6}. It turns out, though, that there exists a \textit{dual} quiver gauge theory corresponding to \eqref{eq: A1gsing}: to see it, notice that it is very similar to the complete intersection
\begin{equation}
    \label{eq: A1gsingold}
    \begin{cases}
    x^2 +y^3+z^4= 0\\
    UV = x
    \end{cases},
\end{equation}
which is nothing but the threefold $X_{E_6}^x$ that we have already vastly analyzed in sections \ref{sec:exampleE6} and \ref{sec:exampleE6quiver}. Therefore, we can extract the corresponding 5d low-energy quiver of \eqref{eq: A1gsing} by employing the hypersurface equations in Table \ref{tab:localmodelsE6z}, taking care of squaring their right-hand side, as now we have $UV=x^{\re{\mathbf{2}}}$, as opposed to $UV=x$ in the $X_{E_6}^x$ case. Proceeding in this way we can explicitly construct the quiver, obtaining the duality:
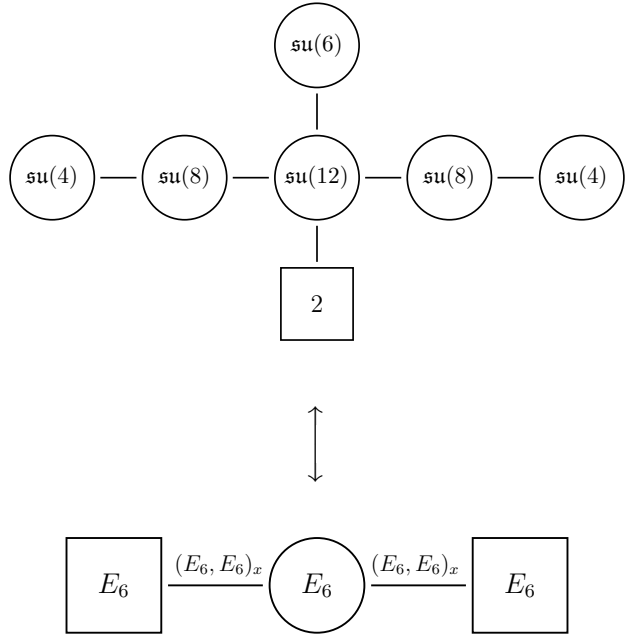
\begin{figure}[H]
\centering
\begin{subfigure}[c]{0.8\textwidth}
\centering
  \scalemath{0.8}{  \begin{tikzpicture}
        \draw[thick] (0,0) circle (0.7);
        \node at (0,0) {\small$\mathfrak{su}(12)$};
        \draw[thick] (0.8,0)--(1.4,0);
        \draw[thick] (2.2,0) circle (0.7);
        \node at (2.2,0) {\small$\mathfrak{su}(8)$};
        \draw[thick] (-0.8,0)--(-1.4,0);
        \draw[thick] (-2.2,0) circle (0.7);
        \node at (-2.2,0) {\small$\mathfrak{su}(8)$};
        \draw[thick] (0,0.8)--(0,1.4);
        \draw[thick] (0,2.2) circle (0.7);
        \node at (0,2.2) {\small$\mathfrak{su}(6)$};
        \draw[thick] (3,0)--(3.6,0);
        \draw[thick] (4.4,0) circle (0.7);
         \draw[thick] (-3,0)--(-3.6,0);
        \node at (4.4,0) {\small$\mathfrak{su}(4)$};
         \draw[thick] (-4.4,0) circle (0.7);
        \node at (-4.4,0) {\small$\mathfrak{su}(4)$};
         \draw[thick] (0,-0.8)--(0,-1.4);
            \draw[thick] (-0.6,-1.5)--(-0.6,-2.7)--(0.6,-2.7)--(0.6,-1.5)--cycle;
        \node at (0,-2.1) {2};
        \end{tikzpicture}}
    \caption*{}
    \end{subfigure}
    \begin{subfigure}[c]{0.8\textwidth}
    \centering
    $\displaystyle\left\updownarrow\vphantom{\int_A^B}\right.$ 
    \end{subfigure}
       \begin{subfigure}{0.8\textwidth}
       \centering
       \vspace{0.7cm}
\scalemath{0.9}{\begin{tikzpicture}
        \draw[thick] (0,0) circle (0.7);
        \node at (0,0) {$E_6$};
        \draw[thick] (0.8,0)--(2.2,0);
        \draw[thick] (-0.8,0)--(-2.2,0);
        \node at (1.45,0.3) {\scalebox{0.8}{$(E_6,E_6)_x$}};
        \node at (-1.45,0.3) {\scalebox{0.8}{$(E_6,E_6)_x$}};
        \draw[thick] (2.3,0.7)--(3.7,0.7)--(3.7,-0.7)--(2.3,-0.7)--cycle;
        \draw[thick] (-2.3,0.7)--(-3.7,0.7)--(-3.7,-0.7)--(-2.3,-0.7)--cycle;
        \node at (3,0) {$E_6$};
        \node at (-3,0) {$E_6$};
        \end{tikzpicture}}
    \end{subfigure}
    \caption{A new 5d duality.}
     \label{fig:newduality}
\end{figure}  

Notice that, according to \eqref{total flavor}, the theory in Figure \ref{fig:newduality} enjoys a flavor symmetry which is at least of the form
\begin{equation}\label{A1 sing flavor}
    \mathcal{G}_{UV} = E_6\times E_6 \times \mathfrak{su}(2).
\end{equation}
Such symmetry can be motivated both from the field-theoretic point of view, thanks to the work of \cite{Tachikawa:2015mha} and \cite{Yonekura}, and the geometric side. As regards the latter, indeed, one can manifestly see that \eqref{eq: A1gsing} has three lines of singularities: two are of $E_6$ type, and one is of $A_1$ type. This is in perfect agreement with \eqref{A1 sing flavor}. \\

The previous procedure can be generalized, maintaining the Calabi-Yau condition, to obtain quivers shaped as the Dynkin diagram of type $A_{n-1}$, with nodes of type $\mathfrak g$, and flavor nodes of type $\mathfrak g$ attached to the rightmost and leftmost nodes of the $A_{n-1}$ Dynkin diagram (following the convention of Figure \ref{tab:DynkinADE}).  The resulting singularity $X_{\mathfrak{g},n-1}^{\bullet}$ is a complete intersection in $\mathbb C^5$ defined by the following  two equations:
\begin{equation}
    \label{eq:Ggsingularity}
    \begin{cases}
    P_{\mathfrak g}(x,y,z)=0\\
    UV=(\bullet)^{n}
    \end{cases},
\end{equation}
where $P_{\mathfrak{g}}$ are equations defining the $ADE$ singularity of type $\mathfrak{g}$ according to Table \ref{tab:singularities}, $\bullet$ is one among $x,y,z$ and the coordinates of the $\mathbb C^5$ ambient space are $x,y,z,U,V$.\\
\indent One might now be tempted to generalize the construction considering as second equation of \eqref{eq:Ggsingularity} some Du Val singularity different from $A_{n-1}$: the gluing condition \eqref{eq:gaugingflavor} would naively admit also these cases. However, one can check that for all the cases in which one picks a $D_k, E_6, E_7, E_8$ singularity as second equation of \eqref{eq:Ggsingularity}, then such generalization is not at finite distance in the moduli space of the Calabi-Yau threefold \cite{Gukov_2000,Xie:2015rpa}.\footnote{ \ Something interesting might come from choosing case $\mathfrak g = D_k$ and replacing $ P_{D_{k}}=0$ with the second equation of \eqref{eq:Ggsingularity}. In this case, the bound that indicates whether the singularity is at finite distance in the moduli space of Calabi-Yau is exactly saturated and we can not readily rule out these singularities.}
Finally, we note that, imposing $n=1$ in \eqref{eq:Ggsingularity} (with $A_0$ the single-center Taub-NUT), we re-obtain the bifundamental matter that we treated in the previous sections: $X_{\mathfrak g,0}^{\bullet} = X_{\mathfrak{g}}^{\bullet}$.\\
\indent As we have already pointed out in Figure \ref{fig:newduality}, the roles of $\mathfrak g$ and $A_{n-1}$ in \eqref{eq:Ggsingularity} can be exchanged, giving the following 5d duality:\\

\textit{Let $(\mathfrak g,\mathfrak g)_{\bullet}$ be the bifundamental matter associated to $X_{\mathfrak g}^{\bullet}$.  Let $Q^{\bullet}_{n-1,\mathfrak g}$ be a quiver of type $A_{n-1}$, with gauge nodes of type $\mathfrak g$, flavor nodes of type $\mathfrak g$ attached to the rightmost and leftmost nodes of the quiver and edges of type $(\mathfrak g,\mathfrak g)_{\bullet}$. Namely, we are referring to the quiver in Figure \ref{fig:generalduality}}.\\
       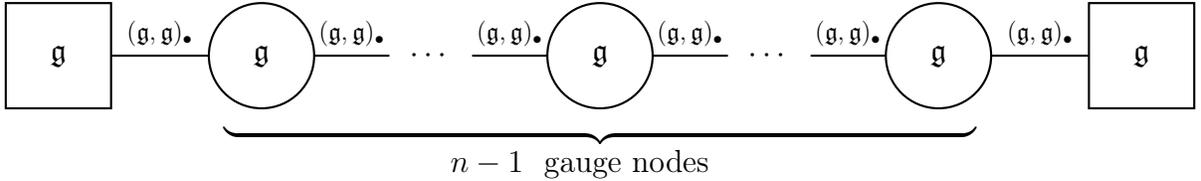
\begin{figure}[H]
    \centering
\scalemath{1}{\begin{tikzpicture}
        \draw[thick] (0,0) circle (0.7);
        \node at (0,0) {$\mathfrak{g}$};
        \draw[thick] (0.7,0)--(1.7,0);
        \draw[thick] (-0.7,0)--(-1.7,0);
        \node at (1.2,0.3) {\scalebox{0.8}{$(\mathfrak g, \mathfrak g)_\bullet$}};
        \node at (-1.2,0.3) {\scalebox{0.8}{$(\mathfrak g, \mathfrak g)_\bullet$}};
          \node at (2.25,0) {$\cdots$};
          \node at (-2.25,0) {$\cdots$};
          \draw[thick] (2.8,0)--(3.8,0);
          \draw[thick] (-2.8,0)--(-3.8,0);
          \node at (3.3,0.3) {\scalebox{0.8}{$(\mathfrak g, \mathfrak g)_\bullet$}};
          \node at (-3.3,0.3) {\scalebox{0.8}{$(\mathfrak g, \mathfrak g)_\bullet$}};
          \draw[thick] (4.5,0) circle (0.7);
          \node at (4.5,0) {$\mathfrak{g}$};
          \draw[thick] (-4.5,0) circle (0.7);
          \node at (-4.5,0) {$\mathfrak{g}$};
          \draw[thick] (5.2,0)--(6.5,0);
          \draw[thick] (-5.2,0)--(-6.5,0);
          \node at (5.85,0.3) {\scalebox{0.8}{$(\mathfrak g, \mathfrak g)_\bullet$}};
          \node at (-5.85,0.3) {\scalebox{0.8}{$(\mathfrak g, \mathfrak g)_\bullet$}};
        
        \draw[thick] (6.5,0.7)--(7.9,0.7)--(7.9,-0.7)--(6.5,-0.7)--cycle;
        \node at (7.2,0) {$\mathfrak{g}$};
        \draw[thick] (-6.5,0.7)--(-7.9,0.7)--(-7.9,-0.7)--(-6.5,-0.7)--cycle;
        \node at (-7.2,0) {$\mathfrak{g}$};
       \node at (0,-1.3) {$\underbrace{\hspace{10 cm}}_{\scalebox{1}{\quad \qquad $n-1$ \text{ gauge nodes}}}$};
        \end{tikzpicture}}
    \caption{Quiver of type $Q_{n-1,\mathfrak{g}}^{\bullet}$.}
    \label{fig:generalduality}
    \end{figure}

\textit{Let $Q^{\bullet}_{\mathfrak g, n-1}$ be the quiver obtained taking the quiver of the $X_{\mathfrak g}^{\bullet}$ theory of section \ref{sec:quivers} and replacing, for each gauge and flavor node, $\mathfrak{su}(N_i)$ with $ \to \mathfrak{su}(n N_i)$. Then, $Q^{\bullet}_{n-1,\mathfrak g}$ and $Q^{\bullet}_{\mathfrak g, n-1}$ are both mass deformations of M-theory on $X^{\bullet}_{\mathfrak g,n-1}$}.\\

\indent Notice that when $\mathfrak{g}$ is of $A$ type (and the edges are \textit{not} of $(A,A)_x$ type) we reproduce a duality which is transparent from the
Type IIB $(p,q)$-brane web language, as it amounts to an S-duality transformation that rotates the web by $\pi/2$.\\
\indent In Figures \ref{fig:dualityE6x}, \ref{fig:dualityE6y} and \ref{fig:dualityE6z} we represent the duality for the cases $X_{E_6,n-1}^{\bullet}$, with $\bullet=x,y,z$ (indicating with $(E_6,E_6)_\bullet$ the conformal matter edges). As we have previously pointed out in a specific example, in order to completely characterize the quivers in Figures \ref{fig:dualityE6x}, \ref{fig:dualityE6y} and \ref{fig:dualityE6z}, namely for the low-energy duals arising from M-theory on \eqref{eq:Ggsingularity}, in the cases $\mathfrak{g}=D,E$, we \textit{must} specify the ``species'' $(\mathfrak g,\mathfrak g)_\bullet$  of the edges of the quiver nodes (corresponding to different types of collision of $Y_{\mathfrak{g}}$ singularities in the threefold), thus decorating the quiver with additional data.\\
\indent The UV flavor symmetry enjoyed by such theories is at least of the form:
\begin{itemize}
    \item Figure 16: $\mathcal{G}_{UV} = E_6\times E_6 \times \mathfrak{su}(n)$,
    \item Figure 17: $\mathcal{G}_{UV} = E_6\times E_6 \times \mathfrak{su}(n)\times \mathfrak{su}(n)\times \mathfrak{u}(1)$,
    \item Figure 18: $\mathcal{G}_{UV} = E_6\times E_6 \times \mathfrak{su}(n)$.
\end{itemize}
As it happened in the \eqref{A1 sing flavor} example, one can manifestly observe the non-abelian part of these symmetries from the geometric point of view, noticing that they appear as non-compact lines of singularities of the appropriate type.

  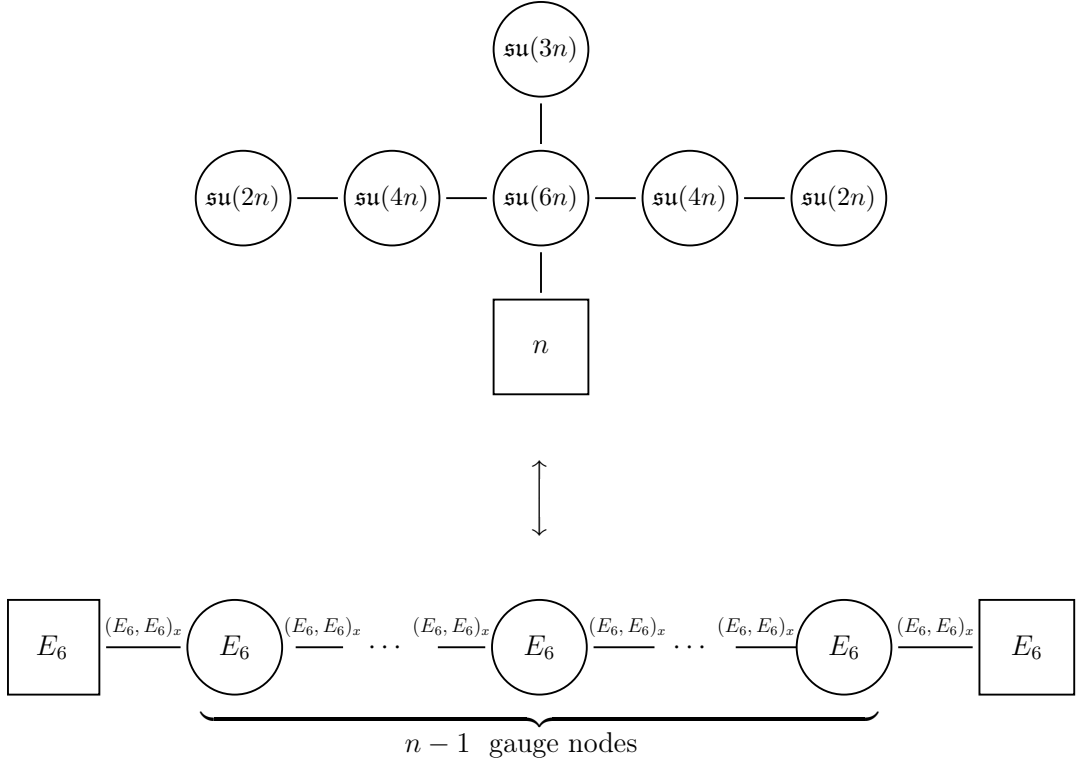
\begin{figure}[H]
\centering
\begin{subfigure}[c]{.9\textwidth}
\centering
  \scalemath{0.9}{  \begin{tikzpicture}
        \draw[thick] (0,0) circle (0.7);
        \node at (0,0) {\small$\mathfrak{su}(6n)$};
        \draw[thick] (0.8,0)--(1.4,0);
        \draw[thick] (2.2,0) circle (0.7);
        \node at (2.2,0) {\small$\mathfrak{su}(4n)$};
        \draw[thick] (-0.8,0)--(-1.4,0);
        \draw[thick] (-2.2,0) circle (0.7);
        \node at (-2.2,0) {\small$\mathfrak{su}(4n)$};
        \draw[thick] (0,0.8)--(0,1.4);
        \draw[thick] (0,2.2) circle (0.7);
        \node at (0,2.2) {\small$\mathfrak{su}(3n)$};
        \draw[thick] (3,0)--(3.6,0);
        \draw[thick] (4.4,0) circle (0.7);
         \draw[thick] (-3,0)--(-3.6,0);
        \node at (4.4,0) {\small$\mathfrak{su}(2n)$};
         \draw[thick] (-4.4,0) circle (0.7);
        \node at (-4.4,0) {\small$\mathfrak{su}(2n)$};
         \node at (0,-2.2) {$n$};
            \draw[thick] (0.7,-1.5)--(0.7,-2.9)--(-0.7,-2.9)--(-0.7,-1.5)--cycle;
        \draw[thick] (0,-0.8)--(0,-1.4);
        \end{tikzpicture}}
    \caption*{}     
    \end{subfigure}
  \begin{subfigure}[c]{0.9\textwidth}
    \centering    $\displaystyle\left\updownarrow\vphantom{\int_A^B}\right.$ 
    \end{subfigure}
       \begin{subfigure}[c]{.9\textwidth}
       \centering
       \vspace{0.8cm}
\scalemath{0.9}{\hspace{-0.08cm}\begin{tikzpicture}
        \draw[thick] (-0.1,0) circle (0.7);
        \node at (-0.1,0) {$E_6$};
        \draw[thick] (0.7,0)--(1.6,0);
        \draw[thick] (-0.9,0)--(-1.6,0);
        \node at (1.2,0.3) {\scalebox{0.7}{$(E_6,E_6)_x$}};
        \node at (-1.4,0.3) {\scalebox{0.7}{$(E_6,E_6)_x$}};
          \node at (2.15,0) {$\cdots$};
          \node at (-2.35,0) {$\cdots$};
          \draw[thick] (2.8,0)--(3.7,0);
          \draw[thick] (-3,0)--(-3.7,0);
          \node at (3.1,0.3) {\scalebox{0.7}{$(E_6,E_6)_x$}};
          \node at (-3.3,0.3) {\scalebox{0.7}{$(E_6,E_6)_x$}};
          \draw[thick] (4.4,0) circle (0.7);
          \node at (4.4,0) {$E_6$};
          \draw[thick] (-4.6,0) circle (0.7);
          \node at (-4.6,0) {$E_6$};
          \draw[thick] (5.2,0)--(6.3,0);
          \draw[thick] (-5.4,0)--(-6.5,0);
          \node at (5.75,0.3) {\scalebox{0.7}{$(E_6,E_6)_x$}};
          \node at (-5.95,0.3) {\scalebox{0.7}{$(E_6,E_6)_x$}};
        
        \draw[thick] (6.4,0.7)--(7.8,0.7)--(7.8,-0.7)--(6.4,-0.7)--cycle;
        \node at (7.1,0) {$E_6$};
        \draw[thick] (-6.6,0.7)--(-7.95,0.7)--(-7.95,-0.7)--(-6.6,-0.7)--cycle;
        \node at (-7.3,0) {$E_6$};
       \node at (-0.1,-1.3) {$\underbrace{\hspace{10 cm}}_{\scalebox{1}{\quad \qquad $n-1$ \text{ gauge nodes}}}$};
        \end{tikzpicture}}
    \end{subfigure}
    \caption{Duality between quivers of type $Q_{E_6,n-1}^{x}$ and quivers of type $Q_{n-1,E_6}^{x}$.}
     \label{fig:dualityE6x}
\end{figure}

  \begin{figure}[H]
\centering
\begin{subfigure}[c]{.9\textwidth}
\centering
  \scalemath{0.9}{  \begin{tikzpicture}
        \draw[thick] (0,0) circle (0.7);
        \node at (0,0) {\small$\mathfrak{su}(4n)$};
        \draw[thick] (0.8,0)--(1.4,0);
        \draw[thick] (2.2,0) circle (0.7);
        \node at (2.2,0) {\small$\mathfrak{su}(3n)$};
        \draw[thick] (-0.8,0)--(-1.4,0);
        \draw[thick] (-2.2,0) circle (0.7);
        \node at (-2.2,0) {\small$\mathfrak{su}(3n)$};
        \draw[thick] (0,0.8)--(0,1.4);
        \draw[thick] (0,2.2) circle (0.7);
        \node at (0,2.2) {\small$\mathfrak{su}(2n)$};
        \draw[thick] (3,0)--(3.6,0);
        \draw[thick] (4.4,0) circle (0.7);
         \draw[thick] (-3,0)--(-3.6,0);
        \node at (4.4,0) {\small$\mathfrak{su}(2n)$};
         \draw[thick] (-4.4,0) circle (0.7);
        \node at (-4.4,0) {\small$\mathfrak{su}(2n)$};
         \draw[thick] (5.2,0)--(5.8,0);
            \draw[thick] (5.9,0.7)--(7.3,0.7)--(7.3,-0.7)--(5.9,-0.7)--cycle;
            \draw[thick] (-5.2,0)--(-5.8,0);
            \draw[thick] (-5.9,0.7)--(-7.3,0.7)--(-7.3,-0.7)--(-5.9,-0.7)--cycle;
        \node at (6.6,0) {$n$};
        \node at (-6.6,0) {$n$};
        \end{tikzpicture}}
    \caption*{}
     
    \end{subfigure}
  \begin{subfigure}[c]{0.9\textwidth}
    \centering
    $\displaystyle\left\updownarrow\vphantom{\int_A^B}\right.$ 
    \end{subfigure}
       \begin{subfigure}[c]{.9\textwidth}
       \centering
 \vspace{0.8cm}
\scalemath{0.9}{\begin{tikzpicture}
        \draw[thick] (0,0) circle (0.7);
        \node at (0,0) {$E_6$};
        \draw[thick] (0.8,0)--(1.6,0);
        \draw[thick] (-0.8,0)--(-1.6,0);
        \node at (1.3,0.3) {\scalebox{0.7}{$(E_6,E_6)_y$}};
        \node at (-1.3,0.3) {\scalebox{0.7}{$(E_6,E_6)_y$}};
          \node at (2.25,0) {$\cdots$};
          \node at (-2.25,0) {$\cdots$};
          \draw[thick] (2.9,0)--(3.7,0);
          \draw[thick] (-2.9,0)--(-3.7,0);
          \node at (3.2,0.3) {\scalebox{0.7}{$(E_6,E_6)_y$}};
          \node at (-3.2,0.3) {\scalebox{0.7}{$(E_6,E_6)_y$}};
          \draw[thick] (4.5,0) circle (0.7);
          \node at (4.5,0) {$E_6$};
          \draw[thick] (-4.5,0) circle (0.7);
          \node at (-4.5,0) {$E_6$};
          \draw[thick] (5.3,0)--(6.4,0);
          \draw[thick] (-5.3,0)--(-6.4,0);
          \node at (5.85,0.3) {\scalebox{0.7}{$(E_6,E_6)_y$}};
          \node at (-5.85,0.3) {\scalebox{0.7}{$(E_6,E_6)_y$}};
        
        \draw[thick] (6.5,0.7)--(7.9,0.7)--(7.9,-0.7)--(6.5,-0.7)--cycle;
        \node at (7.2,0) {$E_6$};
        \draw[thick] (-6.5,0.7)--(-7.9,0.7)--(-7.9,-0.7)--(-6.5,-0.7)--cycle;
        \node at (-7.2,0) {$E_6$};
       \node at (0,-1.3) {$\underbrace{\hspace{10 cm}}_{\scalebox{1}{\quad \qquad $n-1$ \text{ gauge nodes}}}$};
        \end{tikzpicture}}
    \end{subfigure}
    \caption{Duality between quivers of type $Q_{E_6,n-1}^{y}$ and quivers of type $Q_{n-1,E_6}^{y}$.}
     \label{fig:dualityE6y}
\end{figure}

 \begin{figure}[H]
\centering
\begin{subfigure}[c]{.9\textwidth}
\centering
  \scalemath{0.9}{  \begin{tikzpicture}
        \draw[thick] (0,0) circle (0.7);
        \node at (0,0) {\small$\mathfrak{su}(3n)$};
        \draw[thick] (0.8,0)--(1.4,0);
        \draw[thick] (2.2,0) circle (0.7);
        \node at (2.2,0) {\small$\mathfrak{su}(2n)$};
        \draw[thick] (-0.8,0)--(-1.4,0);
        \draw[thick] (-2.2,0) circle (0.7);
        \node at (-2.2,0) {\small$\mathfrak{su}(2n)$};
        \draw[thick] (0,0.8)--(0,1.4);
        \draw[thick] (0,2.2) circle (0.7);
        \node at (0,2.2) {\small$\mathfrak{su}(2n)$};
        \draw[thick] (3,0)--(3.6,0);
        \draw[thick] (4.4,0) circle (0.7);
         \draw[thick] (-3,0)--(-3.6,0);
        \node at (4.4,0) {\small$\mathfrak{su}(n)$};
         \draw[thick] (-4.4,0) circle (0.7);
        \node at (-4.4,0) {\small$\mathfrak{su}(n)$};
         \node at (0,4.4) {$n$};
            \draw[thick] (0.7,3.7)--(0.7,5.1)--(-0.7,5.1)--(-0.7,3.7)--cycle;
        \draw[thick] (0,3)--(0,3.6);
        \end{tikzpicture}}
    \caption*{}
    \end{subfigure}
 \begin{subfigure}[c]{0.9\textwidth}
    \centering
    $\displaystyle\left\updownarrow\vphantom{\int_A^B}\right.$ 
    \end{subfigure}
       \begin{subfigure}[c]{.9\textwidth}
       \centering
       \vspace{0.8cm}
\scalemath{0.9}{\begin{tikzpicture}
        \draw[thick] (0,0) circle (0.7);
        \node at (0,0) {$E_6$};
        \draw[thick] (0.8,0)--(1.6,0);
        \draw[thick] (-0.8,0)--(-1.6,0);
        \node at (1.3,0.3) {\scalebox{0.7}{$(E_6,E_6)_z$}};
        \node at (-1.3,0.3) {\scalebox{0.7}{$(E_6,E_6)_z$}};
          \node at (2.25,0) {$\cdots$};
          \node at (-2.25,0) {$\cdots$};
          \draw[thick] (2.9,0)--(3.7,0);
          \draw[thick] (-2.9,0)--(-3.7,0);
          \node at (3.2,0.3) {\scalebox{0.7}{$(E_6,E_6)_z$}};
          \node at (-3.2,0.3) {\scalebox{0.7}{$(E_6,E_6)_z$}};
          \draw[thick] (4.5,0) circle (0.7);
          \node at (4.5,0) {$E_6$};
          \draw[thick] (-4.5,0) circle (0.7);
          \node at (-4.5,0) {$E_6$};
          \draw[thick] (5.3,0)--(6.4,0);
          \draw[thick] (-5.3,0)--(-6.4,0);
          \node at (5.85,0.3) {\scalebox{0.7}{$(E_6,E_6)_z$}};
          \node at (-5.85,0.3) {\scalebox{0.7}{$(E_6,E_6)_z$}};
        
        \draw[thick] (6.5,0.7)--(7.9,0.7)--(7.9,-0.7)--(6.5,-0.7)--cycle;
        \node at (7.2,0) {$E_6$};
        \draw[thick] (-6.5,0.7)--(-7.9,0.7)--(-7.9,-0.7)--(-6.5,-0.7)--cycle;
        \node at (-7.2,0) {$E_6$};
       \node at (0,-1.3) {$\underbrace{\hspace{10 cm}}_{\scalebox{1}{\quad \qquad $n-1$ \text{ gauge nodes}}}$};
        \end{tikzpicture}}
    \end{subfigure}
    \caption{Duality between quivers of type $Q_{E_6,n-1}^{z}$ and quivers of type $Q_{n-1,E_6}^{z}$.}
     \label{fig:dualityE6z}
\end{figure}
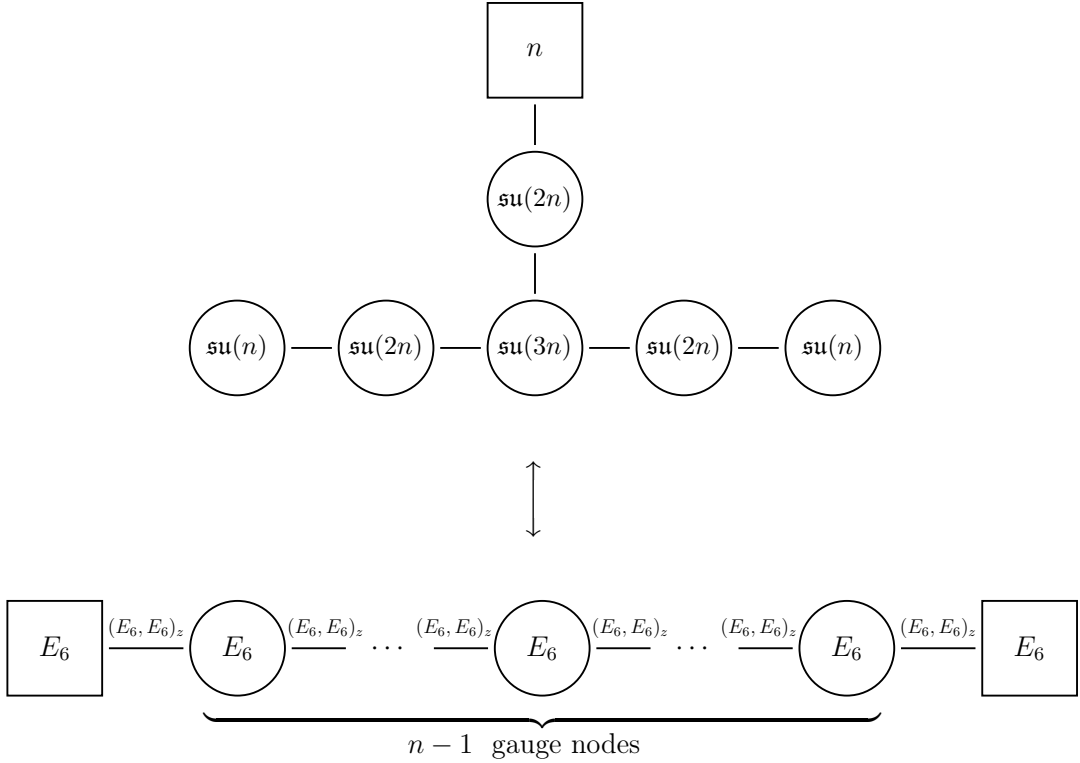  

We conclude this section by noticing that, as we just gauged together two conformal matter theories, one would be naturally tempted to gauge together more than two conformal matter theories. It is easy to see that this does not produce a SCFT, or, geometrically, a Calabi-Yau threefold. Indeed, to glue together two $(\mathfrak g,\mathfrak g)_{\bullet}$ factors we started with two equations (as in \eqref{gaugingE6}) of the Du Val singularity $Y_{\mathfrak g}$ with, respectively  $\bullet \equiv R = uv$  and $\bullet \equiv R' = u'v'$. These equations were glued together imposing $u' = \frac{1}{u}$ and $v' = u^2 v$. The logic behind this construction is that, in the first chart, to each irreducible factor of $R$ corresponds a line of $Y_\mathfrak g$. The line at $u = 0$ is described by varying $v$, and is non-compact as $v$ is the variable describing the fibers of the $\mathcal O(-2)$ bundle. Vice versa, the factor $v =0$ is spanned by the coordinate $u$ and is compact as $u' = 1/u$. In the other chart, a similar argument applies: the line at $v'=0$ is compact while $u' = 0$ corresponds to the fiber at $u' = 0$ of the $\mathcal O(-2)$ bundle. We could try to generalize this by gluing together $N$ lines of $Y_{\mathfrak g}$ by considering 
\begin{eqnarray}
\label{eq:threefoldsingdefmultiple}
 \begin{cases}
 0=P_{\mathfrak{g}}(x,y,z)\rvert_{\bullet = R_N(u,v)}, \\
 0=P_{\mathfrak{g}}(x,y,z)\rvert_{\bullet = R'_N(u',v')},
\end{cases}
\end{eqnarray}
with $u = 1/u'$ and $v = u^k v'$, with $R_N(u,v)$  being  
\begin{equation}
    \label{eq:gauginofthree}
    R_{N} = \underbrace{u(u-u_1)...(u-u_{N-2})}_{N-1 \text{ factors}} v.
\end{equation}
and $R_{N}' = R_{N}(1/u',u'^k v')$. We note that we have, for \eqref{eq:gauginofthree}, $N-1$ non-compact lines of $Y_\mathfrak g$  at $u= u_j$, for $j=1,...,N-2$, and at $u=0$. The explicit expression for $R_N'$ is 
\begin{equation}
    \label{eq:gauginofnII}
    R'_N(u',v') = u'^{k-N+1}\left(u'-\frac{1}{u_1}\right)...\left(u'-\frac{1}{u_{N-2}}\right) v',
\end{equation}
where the factors $u' = \frac{1}{u_j}$ correspond to the same non-compact lines of $Y_{\mathfrak g}$ associated with $u = u_j$. To get in total $N$ non-compact lines of $Y_{\mathfrak g}$, we then need to impose $k=N$ obtaining, at $u' = 0$ (or, equivalently, at $u \to \infty$), a new non-compact line of $Y_{\mathfrak g}$.\\
\indent Summing up, we just showed that the only way to glue together $N$ non-compact lines of $Y_{\mathfrak g}$ along the compact line $v = v' = 0$ is to consider $(u,v)$ as coordinates on $\mathcal O(-N)$. This construction does not yield a Calabi-Yau variety: assuming $\bullet \neq x$, from the equations in \cref{tab:singularities} the holomorphic volume form of \eqref{eq:threefoldsingdefmultiple} is
\begin{equation}
\label{eq:poincareresidueI}
\Omega = \frac{dy \wedge du \wedge dv}{2x} = (u')^{N-2}\frac{dy \wedge du' \wedge dv'}{2x}
\end{equation}
and hence degenerates at $u' = 0$ for $N \neq 2$. The argument works similarly for the $\bullet = x$ case. Hence, we conclude that we \textit{cannot} gauge together more than two $(\mathfrak g,\mathfrak g)_{\bullet}$ conformal matter theories. Nonetheless, this does not exclude ``trifundamental'' conformal matter theories, or trinions, constructed in some other fashion. We will have more to say on this topic in upcoming work.

\subsection{Further generalizations and 5d dualities}
In this section we gradually introduce two natural generalizations of the 5d dualities outlined above.\\

\noindent\textbf{5d dualities and the conifold}\\
For the sake of clarity, consider the following example: suppose that we wish to ``glue'' a singularity of type $X_{E_6}^z$ with one of type $X_{E_6}^y$. This cannot happen on a $\mathbb{P}^1$ with normal bundle $\mathcal{O}_{\mathbb{P}^1}(-2)$. This setup can, however, be realized by gluing the two singularities along a $\mathbb{P}^1$ with normal bundle $\mathcal{O}_{\mathbb{P}^1}(-1)\oplus\mathcal{O}_{\mathbb{P}^1}(-1) $, which is the normal bundle of the $\mathbb{P}^1$ inflated by the small crepant resolution of the conifold. To see how this comes about, write down the complete intersection of the threefolds as:
\begin{equation}\label{threefoldconifold}
    \begin{cases}
     x^2+(uv)^3+w^4=0\\
    x^2+w'^3+(u'v')^4=0\\
    \end{cases},
\end{equation}
where the coordinates $(u,v,w)$ and $(u',v',w')$ are in principle independent. Now we proceed with the gluing: we impose that the two patches $\mathbb{C}^3_{u,v,w}$ and $\mathbb{C}^3_{u',v',w'}$ are related precisely as the two charts on the conifold resolution \cite{Candelas:1989js}:
\begin{equation}
\label{eq:gaugingconifold}
\begin{tikzcd}[row sep = tiny]
 \chi: \mathbb{C}^3_{u,v,w}\arrow{r}& \mathbb{C}^3_{u',v',w'}\\
(u,v,w)\arrow[mapsto]{r} &(u',v',w') = (wv,\frac{1}{v},uv).
\end{tikzcd}
\end{equation}
This implies that the blowdown of \eqref{threefoldconifold} reduces to the system:
\begin{equation}\label{systemconifold}
    \begin{cases}
        x^2+y^3+z^4 = 0\\
        UV = yz \\
    \end{cases},
\end{equation}
namely the intersection of a $E_6$ singularity with the conifold equation. Reasoning as in the previous section, we can hence write down the following 5d duality:
\begin{figure}[H]
\centering
\begin{subfigure}[c]{.9\textwidth}
\centering
  \scalemath{0.9}{  \begin{tikzpicture}
        \draw[thick] (0,0) circle (0.7);
        \node at (0,0) {\small$\mathfrak{su}(7)$};
        \draw[thick] (0.8,0)--(1.4,0);
        \draw[thick] (2.2,0) circle (0.7);
        \node at (2.2,0) {\small$\mathfrak{su}(5)$};
        \draw[thick] (-0.8,0)--(-1.4,0);
        \draw[thick] (-2.2,0) circle (0.7);
        \node at (-2.2,0) {\small$\mathfrak{su}(5)$};
        \draw[thick] (0,0.8)--(0,1.4);
        \draw[thick] (0,2.2) circle (0.7);
        \node at (0,2.2) {\small$\mathfrak{su}(4)$};
        \draw[thick] (3,0)--(3.6,0);
        \draw[thick] (4.4,0) circle (0.7);
         \draw[thick] (-3,0)--(-3.6,0);
        \node at (4.4,0) {\small$\mathfrak{su}(3)$};
         \draw[thick] (-4.4,0) circle (0.7);
        \node at (-4.4,0) {\small$\mathfrak{su}(3)$};
         \draw[thick] (5.2,0)--(5.8,0);
            \draw[thick] (5.9,0.7)--(7.3,0.7)--(7.3,-0.7)--(5.9,-0.7)--cycle;
            \draw[thick] (-5.2,0)--(-5.8,0);
            \draw[thick] (-5.9,0.7)--(-7.3,0.7)--(-7.3,-0.7)--(-5.9,-0.7)--cycle;
            \draw[thick] (0.7,3.7)--(0.7,5.1)--(-0.7,5.1)--(-0.7,3.7)--cycle;
            \draw[thick] (0,3.0)--(0,3.6);
        \node at (0,4.4) {1};
        \node at (6.6,0) {1};
        \node at (-6.6,0) {1};
        \end{tikzpicture}}
    \caption*{}
     
    \end{subfigure}
  \begin{subfigure}[c]{0.9\textwidth}
    \centering
    $\displaystyle\left\updownarrow\vphantom{\int_A^B}\right.$ 
    \end{subfigure}
       \begin{subfigure}[c]{.9\textwidth}
\centering
\vspace{0.8cm}
\scalemath{0.9}{\begin{tikzpicture}
        \draw[thick] (0,0) circle (0.7);
        \node at (0,0) {$E_6$};
        \draw[thick] (0.8,0)--(2.2,0);
        \draw[thick] (-0.8,0)--(-2.2,0);
        \node at (1.45,0.3) {\scalebox{0.8}{ $(E_6,E_6)_y$}};
        \node at (-1.45,0.3) {\scalebox{0.8}{$(E_6,E_6)_z$}};
        \draw[thick] (2.3,0.7)--(3.7,0.7)--(3.7,-0.7)--(2.3,-0.7)--cycle;
        \draw[thick] (-2.3,0.7)--(-3.7,0.7)--(-3.7,-0.7)--(-2.3,-0.7)--cycle;
        \node at (3,0) {$E_6$};
        \node at (-3,0) {$E_6$};
        \end{tikzpicture}}
    \end{subfigure}
    \caption{Duality between quiver of shape $E_6$ with $\mathfrak{su}$ nodes and linear quiver with an edge of type $(E_6,E_6)_x$ and one of type $(E_6,E_6)_y$.}
     \label{fig:dualconifold}
\end{figure}
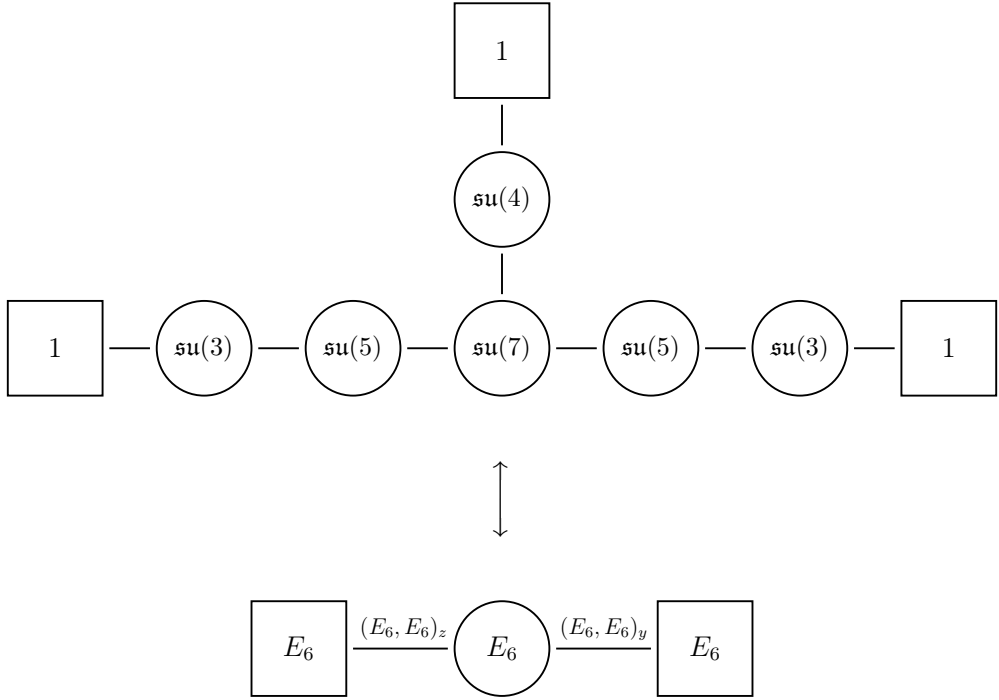  
The corresponding UV flavor symmetry, following \eqref{total flavor}, is:
\begin{equation}
    \mathcal{G}_{UV} = E_6\times E_6 \times \mathfrak{u}(1)^2.
\end{equation}
\noindent We can further generalize the system \eqref{systemconifold} to:
\begin{equation}\label{systemconifoldgen}
    \begin{cases}
        x^2+y^3+z^4 = 0\\
        UV = y^nz^k \\
    \end{cases},
\end{equation}
which corresponds to the 5d duality:
 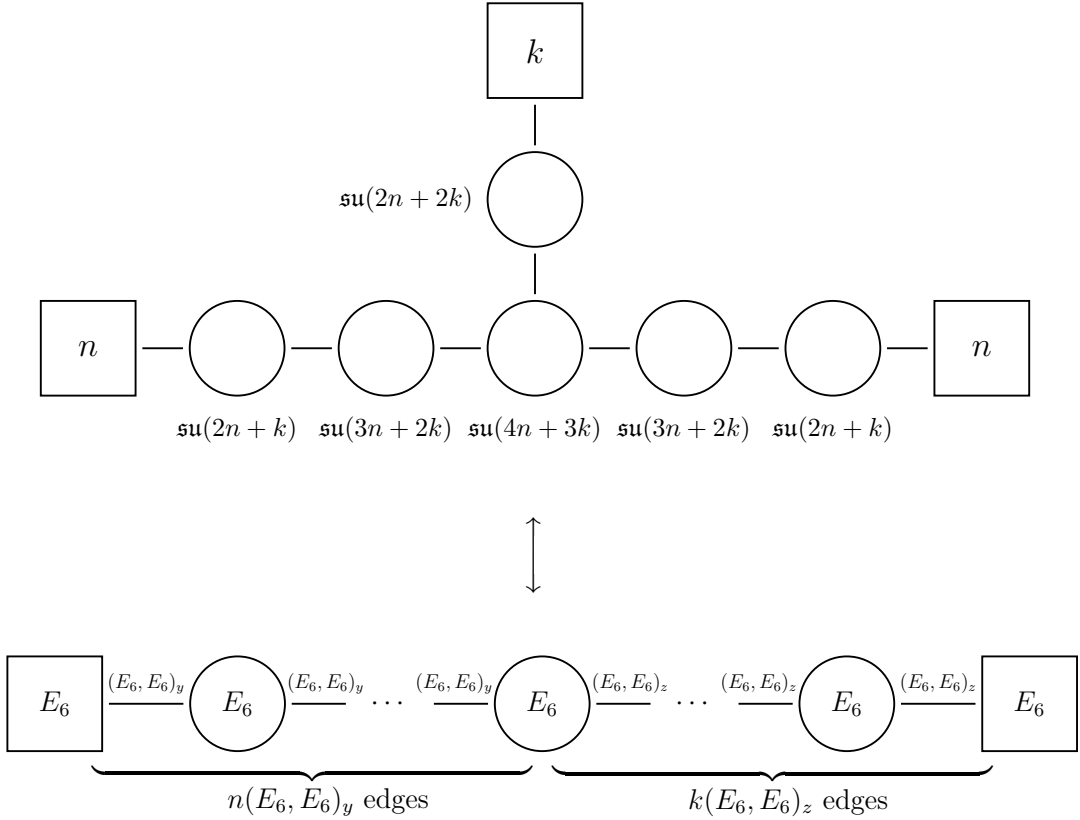
\begin{figure}[H]
\centering
\begin{subfigure}[c]{.9\textwidth}
\centering
  \scalemath{0.9}{  \begin{tikzpicture}
        \draw[thick] (0,0) circle (0.7);
        \node at (0,-1.2) {\small$\mathfrak{su}(4n+3k)$};
        \draw[thick] (0.8,0)--(1.4,0);
        \draw[thick] (2.2,0) circle (0.7);
        \node at (2.2,-1.2) {\small$\mathfrak{su}(3n+2k)$};
        \draw[thick] (-0.8,0)--(-1.4,0);
        \draw[thick] (-2.2,0) circle (0.7);
        \node at (-2.2,-1.2) {\small$\mathfrak{su}(3n+2k)$};
        \draw[thick] (0,0.8)--(0,1.4);
        \draw[thick] (0,2.2) circle (0.7);
        \node at (-1.9,2.2) {\small$\mathfrak{su}(2n+2k)$};
        \draw[thick] (3,0)--(3.6,0);
        \draw[thick] (4.4,0) circle (0.7);
         \draw[thick] (-3,0)--(-3.6,0);
        \node at (4.4,-1.2) {\small$\mathfrak{su}(2n+k)$};
         \draw[thick] (-4.4,0) circle (0.7);
        \node at (-4.4,-1.2) {\small$\mathfrak{su}(2n+k)$};
         \draw[thick] (5.2,0)--(5.8,0);
            \draw[thick] (5.9,0.7)--(7.3,0.7)--(7.3,-0.7)--(5.9,-0.7)--cycle;
            \draw[thick] (-5.2,0)--(-5.8,0);
            \draw[thick] (-5.9,0.7)--(-7.3,0.7)--(-7.3,-0.7)--(-5.9,-0.7)--cycle;
            \draw[thick] (0.7,3.7)--(0.7,5.1)--(-0.7,5.1)--(-0.7,3.7)--cycle;
            \draw[thick] (0,3.0)--(0,3.6);
        \node at (0,4.4) {\scalebox{1.2}{$k$}};
        \node at (6.6,0) {\scalebox{1.2}{$n$}};
        \node at (-6.6,0) {\scalebox{1.2}{$n$}};
        \end{tikzpicture}}
    \caption*{}
     
    \end{subfigure}
  \begin{subfigure}[c]{0.9\textwidth}
    \centering
    $\displaystyle\left\updownarrow\vphantom{\int_A^B}\right.$ 
    \end{subfigure}
       \begin{subfigure}[c]{.9\textwidth}
       \centering
       \vspace{0.8cm}
\scalemath{0.9}{\begin{tikzpicture}
        \draw[thick] (0,0) circle (0.7);
        \node at (0,0) {$E_6$};
        \draw[thick] (0.8,0)--(1.6,0);
        \draw[thick] (-0.8,0)--(-1.6,0);
        \node at (1.3,0.3) {\scalebox{0.7}{$(E_6,E_6)_z$}};
        \node at (-1.3,0.3) {\scalebox{0.7}{$(E_6,E_6)_y$}};
          \node at (2.25,0) {$\cdots$};
          \node at (-2.25,0) {$\cdots$};
          \draw[thick] (2.9,0)--(3.7,0);
          \draw[thick] (-2.9,0)--(-3.7,0);
          \node at (3.2,0.3) {\scalebox{0.7}{$(E_6,E_6)_z$}};
          \node at (-3.2,0.3) {\scalebox{0.7}{$(E_6,E_6)_y$}};
          \draw[thick] (4.5,0) circle (0.7);
          \node at (4.5,0) {$E_6$};
          \draw[thick] (-4.5,0) circle (0.7);
          \node at (-4.5,0) {$E_6$};
          \draw[thick] (5.3,0)--(6.4,0);
          \draw[thick] (-5.3,0)--(-6.4,0);
          \node at (5.85,0.3) {\scalebox{0.7}{$(E_6, E_6)_z$}};
          \node at (-5.85,0.3) {\scalebox{0.7}{$(E_6,E_6)_y$}};
        
        \draw[thick] (6.5,0.7)--(7.9,0.7)--(7.9,-0.7)--(6.5,-0.7)--cycle;
        \node at (7.2,0) {$E_6$};
        \draw[thick] (-6.5,0.7)--(-7.9,0.7)--(-7.9,-0.7)--(-6.5,-0.7)--cycle;
        \node at (-7.2,0) {$E_6$};
       \node at (-3.4,-1.3) {$\underbrace{\hspace{6.5cm}}_{\scalebox{1}{\quad \qquad $n (E_6,E_6)_y$ \text{edges}}}$};
       \node at (3.4,-1.3) {$\underbrace{\hspace{6.5cm}}_{\scalebox{1}{\quad \qquad $k (E_6,E_6)_z$ \text{edges}}}$};
        \end{tikzpicture}}
    \end{subfigure}
    \caption{Duality between quiver of shape $E_6$ with $\mathfrak{su}$ nodes and linear quiver with mixed $E_6$ edges of type $(E_6,E_6)_y$ and $(E_6,E_6)_z$, with a total of $n+k-1$ gauge nodes.}
     \label{fig:dualconifoldgen}
\end{figure}  
Indeed, we have proven in section \ref{sec:toriclocalmodels} that the small crepant resolution of the singularity $UV=y^nz^k$, appearing as the second equation of \eqref{systemconifoldgen}, inflates $n+k-1$ $\mathbb{P}^1$'s on top of the origin $U=V=y=z=0$. These $\mathbb{P}^1$'s correspond to the gauge nodes of the rightmost quiver of Figure \ref{fig:dualconifoldgen}.  Besides, notice that in the rightmost quiver of Figure \ref{fig:dualconifoldgen}  we have conformal matter edges of two different ``species'', namely of type $(E_6,E_6)_z$ and $(E_6,E_6)_y$.\\
\indent It is easy to see that the gauge ranks of the quivers on the two sides of the dualities in Figure \ref{fig:dualconifold} and \ref{fig:dualconifoldgen} match, as a consistency check of the duality. Recall from section \ref{sec:quivers} that the edges $(E_6,E_6)_z$ and $(E_6,E_6)_y$ have rank $5$ and $10$, respectively. The quiver at the bottom has rank:
\begin{equation}
    \text{rk}(\text{Quiver}_B) = 16n+11k-6.
\end{equation}
The quiver at the top has rank:
\begin{equation}
    \text{rk}(\text{Quiver}_T) = \underbrace{k\cdot(5)}_{(E_6,E_6)_z \text{ edges}}+\underbrace{n\cdot(10)}_{(E_6,E_6)_y \text{ edges}}+\underbrace{(n+k-1)\cdot 6}_{E_6 \text{ nodes}} = 16n+11k-6,
\end{equation}
and we see that the two ranks match for any $n$ and $k$ greater than 0. This is a consequence of the fact that, in the geometric picture of section \ref{sec:generalrecipe}, the rank of the nodes on the left-hand side of Figure \ref{fig:dualconifoldgen} arise from the ranks of the edges on the right, so that the matching is automatically guaranteed.

\indent Let us conclude this section with two remarks:\\

\indent \textit{Remark 1: It can be readily checked, with the techniques reviewed in section \ref{sec:flavorUV} and \ref{sec:flavor symmetry}, that the quivers in Figure \ref{fig:dualconifoldgen} possess a flavor symmetry that enhances to at least $E_6\times E_6 \times \mathfrak{su}(n) \times \mathfrak{su}(n) \times \mathfrak{su}(k) \times \mathfrak{u}(1)^2$ in the UV, for all $n$ and $k$ greater than 0. This is in agreement with the geometric picture offered by \eqref{systemconifoldgen}, that spots lines of singularities corresponding to the non-abelian factors.}\\

\indent \textit{Remark 2: The complete intersection \eqref{systemconifoldgen} can be alternatively taken by swapping the variables $(y,z)$ with either $(x,y)$ or $(x,z)$. Extracting the corresponding quivers is straightforward. Again, it can be checked that the flavor symmetry enhances to at least $E_6\times E_6$ in the UV. Naturally, one can also substitute the first equation of \eqref{systemconifoldgen} with any Du Val singularity of type $A_k$, $D_k$ or $E_7, E_8$, yielding further quivers and 5d dualities in an algorithmic manner, retracing the same steps of the $E_6$ case.\footnote{ \ As usual, in the $A_k$ case one can only mix two ``species'' of conformal matter theories: the genuine interacting 5d SCFT arising from $\bullet = x$ and the well-known non-interacting bifundamental theory with $\bullet = z$.} These dualities have been collected in \cref{sec:quiverdualities}.}\\

Furthermore, notice that the quivers at the top of Figure \ref{fig:dualconifoldgen} were already obtained in \cite{Collinucci:2020jqd}, with comparable techniques. \\
\indent Finally, it is easy to check, with techniques similar to those used at the end of \cref{sec:linearquiversI}, that even using different kinds of conformal matter, it is not possible to gauge the diagonal combination of more than two conformal matter theories on the same curve. \\

\noindent\textbf{5d dualities and a fourfold singularity}\\
A straightforward generalization of \eqref{systemconifold} is the complete intersection:
\begin{equation}\label{systemfourfold}
    \begin{cases}
        x^2+y^3+z^4 = 0\\
        UV = xyz \\
    \end{cases}.
\end{equation}
We wonder which is the resolved phase of this system, and whether it admits an interpretation in terms of a quiver with $E_6$ gauge nodes, as we observed in a similar instance in the previous section. Let us see how we can elucidate this point.\\
\indent Simply considering the second equation of \eqref{systemfourfold} we get a toric fourfold hypersurface of $\mathbb C^5$ that possesses three lines of singularities intersecting at the origin, each taking the shape of a conifold. Carrying out a small complete crepant resolution, it is easy to show that the preimage of the origin through the resolution is constituted by two $\mathbb{P}^1$'s, intersecting in the shape of a $A_2$ Dynkin diagram and with normal bundle $\mathcal{O}(0)\oplus\mathcal{O}(-1)\oplus\mathcal{O}(-1)$. Indeed, on the resolved phase we can write the complete intersection \eqref{systemfourfold} as:
\begin{equation}\label{threefoldfourfold}
    \begin{cases}
    (uv)^2+x^3+w^4=0\\
     x'^2+(u'v')^3+w'^4=0\\
    x''^2+w''^3+(u''v'')^4=0\\
    \end{cases}.
\end{equation}
where the patches $\mathbb{C}^4_{x,u,v,w}$, $\mathbb{C}^4_{x',u',v',w'}$ and $\mathbb{C}^4_{x'',u'',v'',w''}$ are related by the transition functions on the two $\mathbb{P}^1$'s with normal bundle $\mathcal{O}(0)\oplus\mathcal{O}(-1)\oplus\mathcal{O}(-1)$:
\begin{equation}
\label{eq:gaugingfourfold}
\begin{tikzcd}[row sep = tiny]
 \chi: \mathbb{C}^4_{x,u,v,w}\arrow{r}& \mathbb{C}^3_{x',u',v',w'}\\
(x,u,v,w)\arrow[mapsto]{r} &(x',u',v',w') = (vu,\frac{1}{u},xu,w).
\end{tikzcd}
\end{equation}
\begin{equation*}
\begin{tikzcd}[row sep = tiny]
 \tilde{\chi}: \mathbb{C}^4_{x',u',v',w'}\arrow{r}& \mathbb{C}^3_{x'',u'',v'',w''}\\
(x',u',v',w')\arrow[mapsto]{r} &(x'',u'',v'',w'') = (x',w'v',\frac{1}{v'},u'v').
\end{tikzcd}
\end{equation*}
The system \eqref{threefoldfourfold} contains three equations of type, respectively, $X_{E_6}^{x}$, $X_{E_6}^{y}$ and $X_{E_6}^{z}$. Thus, proceeding as in the previous section, we can write down a new 5d duality described by \eqref{systemfourfold} in the singular phase (again, the ranks can be quickly checked to be matching, as ensured by the geometric construction of section \ref{sec:generalrecipe}):
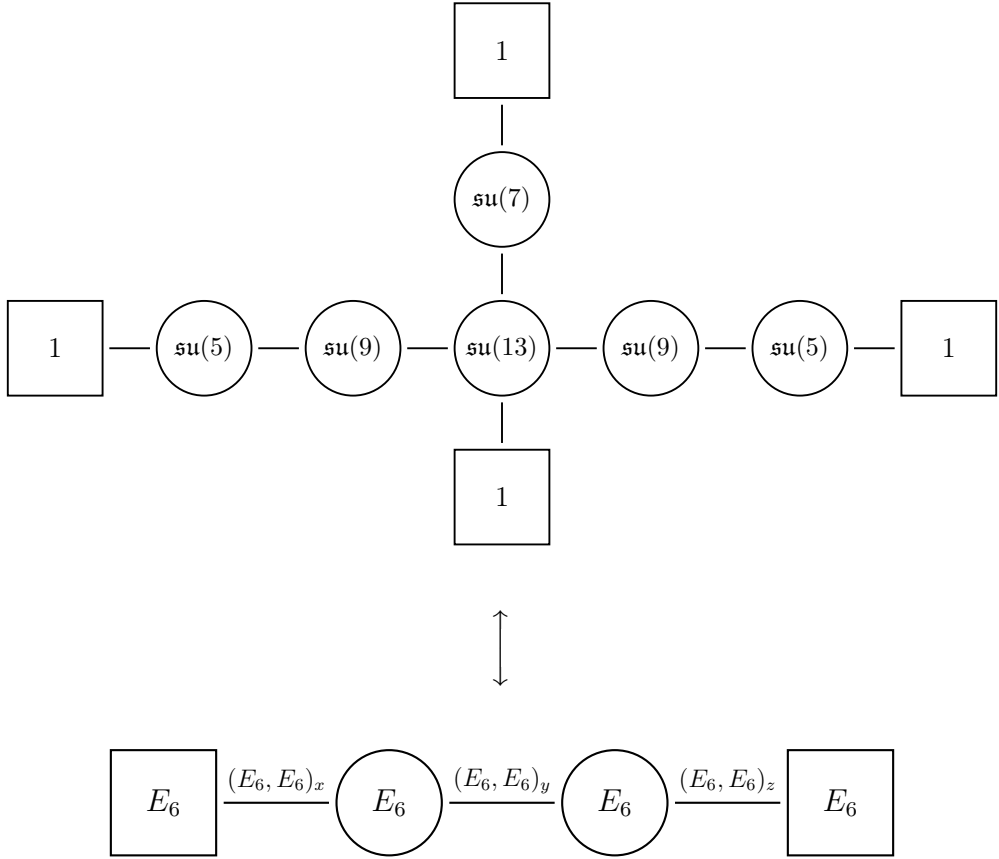
\begin{figure}[H]
\centering
\begin{subfigure}[c]{.9\textwidth}
\centering
  \scalemath{0.9}{  \begin{tikzpicture}
        \draw[thick] (0,0) circle (0.7);
        \node at (0,0) {\small$\mathfrak{su}(13)$};
        \draw[thick] (0.8,0)--(1.4,0);
        \draw[thick] (2.2,0) circle (0.7);
        \node at (2.2,0) {\small$\mathfrak{su}(9)$};
        \draw[thick] (-0.8,0)--(-1.4,0);
        \draw[thick] (-2.2,0) circle (0.7);
        \node at (-2.2,0) {\small$\mathfrak{su}(9)$};
        \draw[thick] (0,0.8)--(0,1.4);
        \draw[thick] (0,2.2) circle (0.7);
        \node at (0,2.2) {\small$\mathfrak{su}(7)$};
        \draw[thick] (3,0)--(3.6,0);
        \draw[thick] (4.4,0) circle (0.7);
         \draw[thick] (-3,0)--(-3.6,0);
        \node at (4.4,0) {\small$\mathfrak{su}(5)$};
         \draw[thick] (-4.4,0) circle (0.7);
        \node at (-4.4,0) {\small$\mathfrak{su}(5)$};
         \draw[thick] (5.2,0)--(5.8,0);
         \draw[thick] (0,-0.8)--(0,-1.4);
            \draw[thick] (5.9,0.7)--(7.3,0.7)--(7.3,-0.7)--(5.9,-0.7)--cycle;
            \draw[thick] (-5.2,0)--(-5.8,0);
            \draw[thick] (-5.9,0.7)--(-7.3,0.7)--(-7.3,-0.7)--(-5.9,-0.7)--cycle;
            \draw[thick] (0.7,3.7)--(0.7,5.1)--(-0.7,5.1)--(-0.7,3.7)--cycle;
            \draw[thick] (0.7,-1.5)--(0.7,-2.9)--(-0.7,-2.9)--(-0.7,-1.5)--cycle;
            \draw[thick] (0,3.0)--(0,3.6);
        \node at (0,4.4) {1};
        \node at (6.6,0) {1};
        \node at (-6.6,0) {1};
        \node at (0,-2.2) {1};
        \end{tikzpicture}}
    \caption*{}
     
    \end{subfigure}
  \begin{subfigure}[c]{0.9\textwidth}
    \centering
    $\displaystyle\left\updownarrow\vphantom{\int_A^B}\right.$ 
    \end{subfigure}
       \begin{subfigure}[c]{.9\textwidth}
    \centering
       \vspace{0.8cm}
\scalemath{1}{\begin{tikzpicture}
        \draw[thick] (-0.4,0) circle (0.7);
        \draw[thick] (2.6,0) circle (0.7);
        \node at (-0.4,0) {$E_6$};
        \draw[thick] (0.4,0)--(1.8,0);
        \draw[thick] (3.4,0)--(4.8,0);
        \draw[thick] (-1.2,0)--(-2.6,0);
        \node at (1.1,0.3) {\scalebox{0.8}{$(E_6,E_6)_y$}};
        \node at (-1.9,0.3) {\scalebox{0.8}{$(E_6,E_6)_x$}};
        \node at (4.1,0.3) {\scalebox{0.8}{$(E_6,E_6)_z$}};
        \draw[thick] (4.9,0.7)--(6.3,0.7)--(6.3,-0.7)--(4.9,-0.7)--cycle;
        \draw[thick] (-2.7,0.7)--(-4.1,0.7)--(-4.1,-0.7)--(-2.7,-0.7)--cycle;
        \node at (2.6,0) {$E_6$};
        \node at (-3.4,0) {$E_6$};
        \node at (5.6,0) {$E_6$};
        \end{tikzpicture}}
    \end{subfigure}
    \caption{Duality between quiver of shape $E_6$ with $\mathfrak{su}$ nodes and linear quiver with mixed edges of type $(E_6,E_6)_x$, $(E_6,E_6)_y$ and $(E_6,E_6)_z$.}
     \label{fig:dualfourfold}
\end{figure}  
Notice that in Figure \ref{fig:dualfourfold} we have three distinct species of $(E_6,E_6)_\bullet$ edges. Of course, the roles of $x,y,z$ can be swapped, obtaining equivalent quivers, and leaving the quiver at the top unchanged. As usual, the UV flavor symmetry is at least $E_6 \times E_6 \times \mathfrak{u}(1)^3$.\\

\indent This example leads us to the most general case, encoding all the 5d dualities previously presented in this work:
\begin{equation}\label{systemfourfoldgen}
    \begin{cases}
        P_{\mathfrak{g}}(x,y,z)= 0\\
        UV = x^hy^nz^k \\
    \end{cases},
\end{equation}
where $P_{\mathfrak{g}}(x,y,z)$ is a Du Val singularity and $h,n,k$ are integers \textit{greater} or equal to zero.\footnote{ \ Although at least one of them \textit{must not} vanish.}
Writing down the corresponding quivers for all cases is straightforward on both sides of the duality, following the recipe we have just outlined.\footnote{ \ As previously mentioned, one feels the urge to further generalize the second equation of \eqref{systemfourfoldgen} to a (possibly deformed) singularity of type $D$ or $E$. Unfortunately, in this case the complete intersection \eqref{systemfourfoldgen} would not respect the finite distance condition of \cite{Gukov_2000}. Physically, this condition, for quasi-homogeneous singularities, is also crucial to guarantee the existence of an R-symmetry for the five-dimensional theory obtained from the geometric engineering, and hence to obtain $\mathcal N = 1$ supersymmetry in 5d.}\\
\indent We conclude by noticing that three or more species of different conformal matters can not be glued together on a single compact $\mathbb P^1$, in the spirit of \cref{sec:linearquiversI}.\\

\subsection{Complete list of the five-dimensional dualities}
\label{sec:quiverdualities}
In the previous sections, we have analyzed various instances of 5d SCFTs admitting a low-energy quiver gauge theory description, arising from M-theory on the singular threefolds \eqref{systemfourfoldgen}.
A notable subcase is given by $h=n=0$ and $k=1$ (and analogous permutations of $h,n,k$). This choice yields the 5d conformal matter that has played a paramount part in this work.\\ 
\indent In this section we collect the general expressions governing the new 5d dualities arising from \eqref{systemfourfoldgen} for \textit{generic} $h,n,k$ (taken to be integers equal or greater than 0). These are dualities between low-energy quiver gauge theories of shape $\mathfrak{g}$ with $\mathfrak{su}$ nodes and linear quivers with nodes of type $A,D,E$, with edges expressed in terms of 5d conformal matter.
In the tables below we list these new 5d dualities, first giving the general shape of the quivers with $\mathfrak{su}(m_i)$ nodes, and then specifying a vector that encodes the number $m_i$. We further write down the UV flavor symmetry $\mathcal{G}_{UV}$ enjoyed by the theories, for $h,n,k\neq 0$ (as remarked in the main text, the total UV flavor symmetry will in general contain $\mathcal{G}_{UV}$, and exactly match its rank).\\
\indent In practice, given a choice of $\mathfrak{g}$ and of the exponents $h,n,k$ in \eqref{systemfourfoldgen}, one proceeds in two steps in order to construct the 5d duality:
\begin{itemize}
    \item To write down the low-energy quiver side with $\mathfrak{su}$ nodes, one uses as basic building blocks the quivers exhibited in section \ref{sec:quivers}, arising from the geometries $X_{\mathfrak{g}}^{\bullet}$. Given a labelling of the nodes of these quivers, we can specify the corresponding $\mathfrak{su}(m_i)$ group for each node with a vector $v_{\mathfrak{g}}^{\bullet}$, whose entries are nothing but the $m_i$.
    Then, one obtains the vector $v_{\mathfrak{g}}^{(h,n,k)}$ that specifies the quiver \eqref{systemfourfoldgen} as follows. For each factor of $\bullet = x,y,z$ on the right-hand side of the equation $UV = x^hy^nz^k$,  we \textit{linearly} add the vectors of the basic building blocks, with the appropriate multiplicity given by the exponents $h,n,k$:
    \begin{equation}\label{vector addition}
        v_{\mathfrak{g}}^{(h,n,k)} = h v_{\mathfrak{g}}^x + n v_{\mathfrak{g}}^y + k v_{\mathfrak{g}}^z.
    \end{equation}
    Formula \eqref{vector addition} is an immediate consequence of the blowup maps for the algebras $\mathfrak{g}$ presented in appendix \ref{sec:duvalres}. Let us add a crucial physical remark:\\
    \textit{By construction, all the low-energy quiver gauge theories defined by the vector $v_{\mathfrak{g}}^{(h,n,k)}$ flow to a 5d SCFT in the UV sporting at least $\mathfrak{g}\times\mathfrak{g}$ flavor symmetry.}
    \item The other side of the duality, involving edges with 5d conformal matter, can be written down as a \textit{linear} quiver with two outermost flavor nodes of type $\mathfrak{g}$. Between them there are (arranged in a linear fashion):
    \begin{itemize}
        \item $h$ edges of type $(\mathfrak g,\mathfrak g)_x$,
        \item $n$ edges of type $(\mathfrak g,\mathfrak g)_y$,
        \item $k$ edges of type $(\mathfrak g,\mathfrak g)_z$,
        \item $h+n+k-1$ gauge nodes of type $\mathfrak{g}$.
    \end{itemize}
    For $\mathfrak{g}=A$ we consider only $(\mathfrak{g},\mathfrak{g})_x$ edges (or equivalently $(\mathfrak{g},\mathfrak{g})_y$ edges), as the $(\mathfrak{g},\mathfrak{g})_z$ one is simply a bifundamental theory with no compact divisors.
    As we have mentioned in section \ref{sec:exceptional quivers}, \textit{the order of the edges does not matter}, namely we can freely exchange edges of different conformal matter species (as long as we keep the number that each species appears fixed), since this operation does not change the corresponding singular geometry.
\end{itemize}
Let's clarify this apparently involved construction by listing the explicit dualities for $\mathfrak{g} = A_k,D_k,E_6,E_7,E_8$. We specify the shape of the quivers, along with a labelling of their nodes, and the vectors $v_{\mathfrak{g}}^{x},v_{\mathfrak{g}}^{y},v_{\mathfrak{g}}^{z},v_{\mathfrak{g}}^{(h,n,k)}$. Vector entries pertaining to gauge nodes of type $\mathfrak{su}(m_i)$ contain the number $m_i$. For vector entries referring to a flavor node, we specify the number of flavors, say $\#$, and add a subscript $\#_f$ to distinguish them from gauge ones. A null entry in the vector means that the node is not present.\\
\indent We note that, as in the case of the $(\mathfrak g,\mathfrak g)_\bullet$ theories, also the generalized quivers built from bifundamental conformal matter display no electric one-form symmetry nor magnetic two-form symmetry in five dimensions. Indeed, as we show in the remaining part of this section, all the generalized quivers can be dualized to quivers with nodes of type $\mathfrak{su}$ that always display fundamental matter that explicitly breaks the central symmetry of the gauge nodes. 

\newpage

\vspace{-1cm}
\centering
$
}
    \caption*{}
    \label{A2jplus1labelling}
    \end{figure}
    \end{minipage}\\

\begin{minipage}{0.22\textwidth}
   \textbf{Vectors:}
    \end{minipage}
\begin{minipage}{0.75\textwidth}
       $\begin{array}{l}
            v_{A_{2j+1}}^x = (1,2\ldots,j,j+1,j,\ldots,2,1,2_f)\\
            \begin{aligned}
            v_{A_{2j+1}}^{(h,n,k)} = h v_{A_{2j+1}}^x
            \end{aligned}\\
       \end{array}$
    \end{minipage}\\

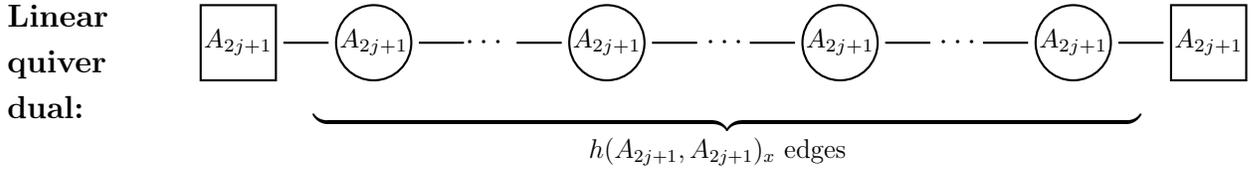
\begin{minipage}{0.15\textwidth}
\textbf{Linear quiver dual:}
    \end{minipage}
\begin{minipage}{0.83\textwidth}
 \begin{figure}[H]
       \centering
       \vspace{0.3cm}
\scalemath{1}{%
}
    \end{figure}
    \end{minipage}\\
\vspace{0.3cm}

\begin{minipage}{0.16\textwidth}
\textbf{UV flavor symmetry: } 
    \end{minipage}
\begin{minipage}{0.8\textwidth}
 $\quad \mathcal{G}_{UV} = A_{2j+1}\times A_{2j+1} \times \mathfrak{su}(2h)$
    \end{minipage}\\

\vspace{0.6cm}
\centering
$%
}
    \caption*{}
    \label{A2jlabelling}
    \end{figure}
    \end{minipage}\\

\begin{minipage}{0.22\textwidth}
   \textbf{Vectors:}
    \end{minipage}
\begin{minipage}{0.75\textwidth}
       $\begin{array}{l}
            v_{A_{2j}}^x = (1,2\ldots,j-1,j,j,j-1,\ldots,2,1,1_f)\\
            \begin{aligned}
            v_{A_{2j}}^{(h,n,k)} = h v_{A_{2j}}^x
            \end{aligned}\\
       \end{array}$
    \end{minipage}\\

\begin{minipage}{0.15\textwidth}
\textbf{Linear quiver dual:}
    \end{minipage}
\begin{minipage}{0.83\textwidth}
 \begin{figure}[H]
       \centering
       \vspace{0.3cm}
\scalemath{1}{%
}
    \end{figure}
    \end{minipage}\\
\vspace{0.3cm}

\begin{minipage}{0.16\textwidth}
\textbf{UV flavor symmetry: } 
    \end{minipage}
\begin{minipage}{0.8\textwidth}
 $\quad \mathcal{G}_{UV} = A_{2j}\times A_{2j} \times \mathfrak{su}(h)$
    \end{minipage}\\

\newpage

\vspace{1.6cm}
\centering
$%
}
    \caption*{}
    \label{D4labelling}
    \end{figure}
    \end{minipage}\\

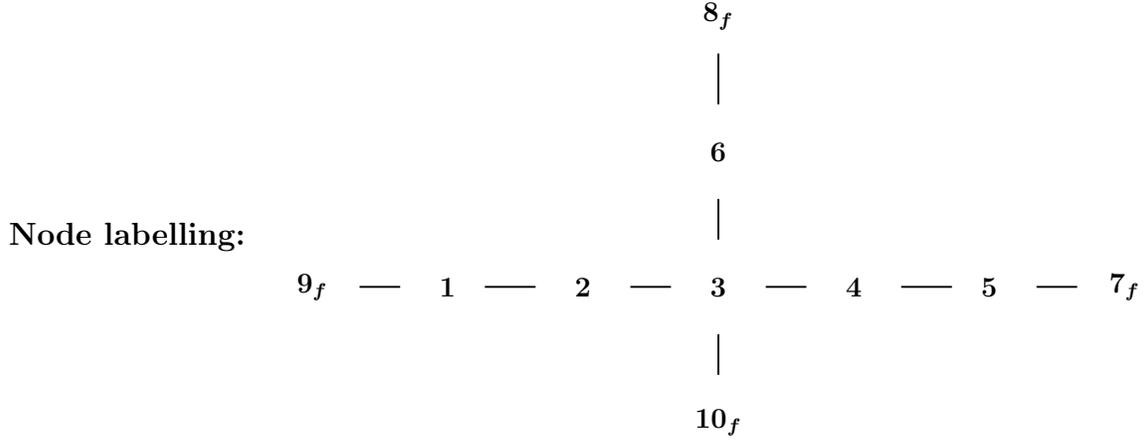
\begin{minipage}{0.22\textwidth}
   \textbf{Vectors:}
    \end{minipage}
\begin{minipage}{0.75\textwidth}
       $\begin{array}{l}
            v_{D_{2j+2}}^x = (j+1,j+1,2j+1,2j,\ldots,2,1_f,1_f,1_f,0_f)\\
            v_{D_{2j+2}}^y = (j,j,2j,2j-1,\ldots,1,0_f,0_f,0_f,1_f)\\
            v_{D_{2j+2}}^z = (1,1,2,2,\ldots,2,2_f,0_f,0_f,0_f)\\
            \begin{aligned}
            v_{D_{2j+2}}^{(h,n,k)} = h v_{D_{2j+2}}^x + n v_{D_{2j+2}}^y + k v_{D_{2j+2}}^z
            \end{aligned}\\
       \end{array}$
    \end{minipage}\\

\begin{minipage}{0.15\textwidth}
\textbf{Linear quiver dual:}
    \end{minipage}
\begin{minipage}{0.83\textwidth}
 \begin{figure}[H]
       \centering
       \vspace{0.8cm}
\scalemath{1}{%
}
    \end{figure}
    \end{minipage}\\
\vspace{0.5cm}

\begin{minipage}{0.16\textwidth}
\textbf{UV flavor symmetry: } 
    \end{minipage}
\begin{minipage}{0.8\textwidth}
 $\quad \mathcal{G}_{UV} = D_{2j+2}\times D_{2j+2} \times \mathfrak{su}(h+2k) \times \mathfrak{su}(h) \times \mathfrak{su}(h) \times \mathfrak{su}(n)  \times \mathfrak{u}(1)^3$
    \end{minipage}\\

    \newpage

\vspace{1.5cm}
\centering
$%
}
    \caption*{}
    \label{D5labelling}
    \end{figure}
    \end{minipage}\\

\begin{minipage}{0.12\textwidth}
   \textbf{Vectors:}
    \end{minipage}
\begin{minipage}{0.86\textwidth}
       $\begin{array}{l}
            v_{D_{2j+3}}^x = (j+1,j+1,2j+2,2j+1,\ldots,2,1_f,0_f,0_f,1_f)\\
            v_{D_{2j+3}}^y = (j+1,j+1,2j+1,2j,\ldots,1,0_f,1_f,1_f,0_f)\\
            v_{D_{2j+3}}^z = (1,1,2,2,\ldots,2,2_f,0_f,0_f,0_f)\\
            \begin{aligned}
            v_{D_{2j+3}}^{(h,n,k)} = h v_{D_{2j+3}}^x + n v_{D_{2j+3}}^y + k v_{D_{2j+3}}^z
            \end{aligned}\\
       \end{array}$
    \end{minipage}\\

\begin{minipage}{0.15\textwidth}
\textbf{Linear quiver dual:}
    \end{minipage}
\begin{minipage}{0.83\textwidth}
 \begin{figure}[H]
       \centering
       \vspace{0.8cm}
\scalemath{1}{%
}
    \end{figure}
    \end{minipage}\\

\vspace{0.5cm}
\begin{minipage}{0.16\textwidth}
\textbf{UV flavor symmetry: } 
    \end{minipage}
\begin{minipage}{0.8\textwidth}
 $\quad \mathcal{G}_{UV} = D_{2j+3}\times D_{2j+3} \times \mathfrak{su}(h+2k) \times \mathfrak{su}(n) \times \mathfrak{su}(n) \times \mathfrak{su}(h)  \times \mathfrak{u}(1)^3$
    \end{minipage}\\

\newpage

\centering
$%
}
    \caption*{}
    \label{E6labelling}
    \end{figure}
    \end{minipage}\\

\begin{minipage}{0.12\textwidth}
   \textbf{Vectors:}
    \end{minipage}
\begin{minipage}{0.86\textwidth}
       $\begin{array}{l}
            v_{E_6}^x = (2,4,6,4,2,3,0,0,0,1_f)\\
            v_{E_6}^y = (2,3,4,3,2,2,1_f,0,1_f,0)\\
            v_{E_6}^z = (1,2,3,2,1,2,0,1_f,0,0)\\
            \begin{aligned}
            v_{E_6}^{(h,n,k)} = (2h+2n+k,4h+3n+2k,6h+4n+3k,4h+3n+2k,2h+2n+k,\\
            3h+2n+2k,n_f,k_f,n_f,h_f)
            \end{aligned}\\
       \end{array}$
    \end{minipage}\\

\begin{minipage}{0.15\textwidth}
\textbf{Linear quiver dual:}
    \end{minipage}
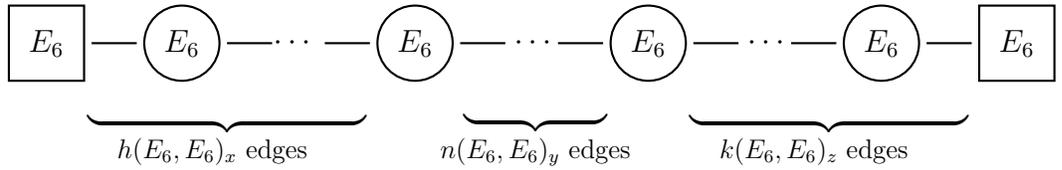
\begin{minipage}{0.83\textwidth}
 \begin{figure}[H]
       \centering
       \vspace{0.8cm}
\scalemath{1}{%
}
    \end{figure}
    \end{minipage}\\

\vspace{0.5cm}
\begin{minipage}{0.16\textwidth}
\textbf{UV flavor symmetry: } 
    \end{minipage}
\begin{minipage}{0.8\textwidth}
 $\quad \mathcal{G}_{UV} = E_6\times E_6 \times \mathfrak{su}(n) \times \mathfrak{su}(k) \times \mathfrak{su}(n) \times \mathfrak{su}(h)  \times \mathfrak{u}(1)^3$
    \end{minipage}\\

\newpage
\vspace{1.6cm}
\centering
$%
}
    \caption*{}
    \label{E7labelling}
    \end{figure}
    \end{minipage}\\

\begin{minipage}{0.12\textwidth}
   \textbf{Vectors:}
    \end{minipage}
\begin{minipage}{0.86\textwidth}
       $\begin{array}{l}
            v_{E_7}^x = (3,6,9,7,5,3,5,1_f,1_f,0)\\
            v_{E_7}^y = (2,4,6,5,4,3,3,2_f,0,0)\\
            v_{E_7}^z = (2,4,7,4,2,1,3,0,0,1_f)\\
            \begin{aligned}
            v_{E_7}^{(h,n,k)} = (3h+2n+2k,6h+4n+4k,9h+6n+7k,7h+5n+4k,5h+4n+2k,\\
            3h+3n+k,5h+3n+3k,(h+2n)_f,h_f,k_f)
            \end{aligned}\\
       \end{array}$
    \end{minipage}\\

\begin{minipage}{0.15\textwidth}
\textbf{Linear quiver dual:}
    \end{minipage}
\begin{minipage}{0.83\textwidth}
 \begin{figure}[H]
       \centering
       \vspace{0.8cm}
\scalemath{1}{%
}
    \end{figure}
    \end{minipage}\\

\vspace{0.5cm}
\begin{minipage}{0.16\textwidth}
\textbf{UV flavor symmetry: } 
    \end{minipage}
\begin{minipage}{0.8\textwidth}
 $\quad \mathcal{G}_{UV} = E_7\times E_7 \times \mathfrak{su}(h+2n) \times \mathfrak{su}(h) \times \mathfrak{su}(k) \times \mathfrak{u}(1)^2$
    \end{minipage}\\

\newpage
\centering
$%
}
    \caption*{}
    \label{E8labelling}
    \end{figure}
    \end{minipage}\\

\begin{minipage}{0.12\textwidth}
   \textbf{Vectors:}
    \end{minipage}
\begin{minipage}{0.86\textwidth}
       $\begin{array}{l}
            v_{E_8}^x = (5,10,15,12,9,6,3,8,0,1_f,0)\\
            v_{E_8}^y = (4,7,10,8,6,4,2,5,0,0,1_f)\\
            v_{E_8}^z = (2,4,6,5,4,3,2,3,1_f,0,0)\\
            \begin{aligned}
            v_{E_8}^{(h,n,k)} = (5h+4n+2k,10h+7n+4k,15h+10n+6k,12h+8n+5k,9h+6n+4k,\\
            6h+4n+3k,3h+2n+2k,8h+5n+3k,k_f,h_f,n_f)
            \end{aligned}\\
       \end{array}$
    \end{minipage}\\

\begin{minipage}{0.15\textwidth}
\textbf{Linear quiver dual:}
    \end{minipage}
\begin{minipage}{0.83\textwidth}
 \begin{figure}[H]
      \centering
       \vspace{0.8cm}
\scalemath{1}{%
}
    \end{figure}
    \end{minipage}\\

\vspace{0.5cm}
\begin{minipage}{0.16\textwidth}
\textbf{UV flavor symmetry: } 
    \end{minipage}
\begin{minipage}{0.8\textwidth}
 $\quad \mathcal{G}_{UV} = E_8\times E_8 \times \mathfrak{su}(k) \times \mathfrak{su}(h) \times \mathfrak{su}(n)  \times \mathfrak{u}(1)^2$
    \end{minipage}\\

\justifying

\newpage

\section{Conclusions and Outlook}\label{sec:conclusions}
In this work we have introduced new 5d SCFTs featuring ``5d conformal matter'' with $\mathfrak{g}\times\mathfrak{g}$ flavor symmetry, for $\mathfrak{g}=A_k,D_k,E_6,E_7,E_8$, following the spirit of M-theory geometric engineering. This has led us to identify different \textit{species} of $\mathfrak{g}\times\mathfrak{g}$ 5d conformal matter, that must therefore be decorated with an additional label, in order to keep track of the threefold that engineered them in the first place. Theories of different species (but with the same $\mathfrak{g}\times\mathfrak{g}$ flavor in the UV) have different Coulomb branch dimensions, and thus are genuinely inequivalent.\\
\indent Crucially, the 5d conformal matter theories that we have constructed \textit{do not} straightforwardly descend from a simple KK reduction of their 6d cousins, except for selected cases that we have highlighted in the text.\\
\indent Furthermore, the theories we have introduced in this work provide the basic building blocks to construct \textit{linear} quiver gauge theories with gauge nodes of type $A$, $D$ and $E$, much alike the well-known 6d conformal matter theories. The additional feature, in the 5d context, is that such linear quivers can be built concatenating elementary building blocks of \textit{different} 5d conformal matter species, along with more standard cases with conformal matter of the same species. The entirety of these linear quiver gauge theories has a neat interpretation in terms of M-theory geometric engineering on a class of threefolds built as complete intersections: thanks to this geometric interpretation, we introduced a plethora of new 5d dualities between linear quiver gauge theories with conformal matter of type $\mathfrak{g}\times\mathfrak{g}$ (of all possible species) and quivers of type $\mathfrak{g}$ with $\mathfrak{su}$ nodes. This expands and fully generalizes the results of \cite{Collinucci:2020jqd,Ohmori:2015pia}, providing a unified perspective on these works.\\
\indent Plenty of aspects pertaining to 5d conformal matter remain to be investigated. In particular, the following directions are subject of current investigation \cite{DMDZGSinprep}:
\begin{itemize}
\item Classification program for 5d conformal matter of type $\mathfrak{g}\times\mathfrak{g}'$, with $\mathfrak{g}\neq\mathfrak{g}'$. These can be easily engineered through M-theory compactified on hypersurface threefolds with collisions of suitable non-compact singularities,  such as:
\begin{equation*}
    \begin{array}{ll}
    D_k \times D_n: \quad\quad & x^2+uvy^2+u^{k-1}v^{n-1}=0 \\
    D_4 \times E_7: \quad\quad & x^2 + z u^2v^3 + z^3=0 \\
        E_6 \times E_7: \quad\quad & x^2+y^3+u^3v^3y+u^4v^5=0 \\
         E_6 \times E_8: \quad\quad & x^2+y^3+u^4v^5=0 \\
         E_7 \times E_8: \quad\quad & x^2+y^3+u^3v^4y+u^5v^5=0 \\
    \end{array}
\end{equation*}
Their resolution requires slightly different (yet straightforward) methods compared to the $\mathfrak{g}=\mathfrak{g}'$ examples presented in this work, but notice that all the mixed $D-E$ and $E-E$ examples descend from 6d conformal matter, as they are in Weierstrass form. It would be interesting to understand more thoroughly if there exists $\mathfrak{g}\times\mathfrak{g}'$ conformal matter of intrinsic 5d origin: hints in this direction can be found in a geometric construction that we will present in upcoming work.
\item Determine whether the structure of 5d conformal matter theories is more akin to its 6d parents, or to its 4d descendants; namely if we can construct, respectively, only \textit{linear} 5d quiver gauge theories built out of bifundamental conformal matter, or also more general quivers, with e.g.\ trivalent nodes. This entails figuring out whether a 5d trinion theory displaying at least a factor $\mathfrak g \times \mathfrak g \times \mathfrak g$ with $\mathfrak g$ of type $D_k$ or $E_6,E_7,E_8$ exists. In five-dimensions, toric orbifold trinions with gauge algebra $A_N \times A_N \times A_N$ can be built as $\mathbb C^3/\left(\mathbb Z_{N+1} \times \mathbb Z_{N+1}\right)$ \cite{Benini:2009gi}, suggesting that this might be possible also for other types of ADE algebras. Indeed, one can check that, for example, the singularity
\begin{equation}
    \label{eq:trinionD4z}
    x^2  + z u^2 v^2(u-v)^2 + z^3= 0 
\end{equation}
satisfies the \cite{Gukov_2000} criterium for being at finite distance in the moduli space of Calabi-Yau threefolds, and sports three lines of $D_4$ singularities, suggesting a flavor symmetry with at least a $D_{4} \times D_{4} \times D_{4}$ subgroup.
\item Determine the circle reduction to four spacetime dimensions of the 5d conformal matter theories, for all \textit{species}. It turns out that all such 5d theories admit a class $\mathcal{S}$ construction that describes their 4d descendant. A related aspect is to explicitly compute the Higgs branches of the 5d conformal matter SCFTs, understanding the action of the $\mathfrak{g} \times \mathfrak{g}$ flavor groups. At the moment, this looks like an interesting challenge on its own, since the dynamics of M-theory on non-isolated singularities is still to be completely clarified. Nonetheless, its quaternionic dimension can be explicitly computed and matched with the class $\mathcal{S}$ descendants. We will present all these aspects in full details in future work.
\item  Finally, on a related note, one might want to precisely track down the influence of T-branes on 5d conformal matter theories: similarly to what happens in the 6d case, our threefold geometries can in general be supplemented by T-brane data, classified by nilpotent orbits of the appropriate ADE algebra. This can have a bipartite effect: on the one hand, the flavor group of 5d conformal matter theories $\mathfrak{g}\times \mathfrak{g}$ might be partially broken; in a complementary fashion, switching on a T-brane background for the gauge nodes might disrupt part of the gauge symmetry. As usual, both cases keep the geometry intact, and only affect symmetries and moduli spaces dimensions. Notice that the effect of T-branes on the flavor group can be neatly given a Type IIA interpretation in terms of bound states of D6-branes for a subset of the $\mathfrak{g}=A$ cases \cite{DeMarco:2021try,DeMarco:2022dgh,Bourget:2023wlb} (namely, the ones that only yield bifundamental matter with trivial Coulomb branch).
Pursuing this path for $D\times D$ and $E\times E$ 5d conformal matter theories would constitute a further interesting direction for future inquiry.
\end{itemize}

\section*{Acknowledgements}

MDZ thanks David R. Morrison and Jonathan J. Heckman for several illuminating discussions in the period 2015-17 which inspired this work, as well as the Simons Collaboration on Special Holonomy in Geometry, Analysis, and Physics for many inspiring discussions in the period 2018-2020, which served as further motivation for this paper. We thank Bobby Acharya, Sergio Benvenuti, Amihay Hanany, Lotte Hollands, Marco Fazzi, Alessandro Tomasiello, Sakura Sch\"afer-Nameki and Roberto Valandro for discussions. The work of MDZ and MDM has received funding from the European Research Council (ERC) under the European Union’s Horizon 2020 research and innovation program (grant agreement No. 851931). MDZ also acknowledges support from the Simons Foundation Grant \#888984 (Simons Collaboration on Global Categorical Symmetries). AS wishes to thank Accademia dei Lincei for financial support, as well as the Department of Mathematics and the Department of Physics of Uppsala University for hospitality during the realization of this work. AS further acknowledges funding from the Simons Collaboration on Special Holonomy. MDM thanks the ``Fondazione Angelo Della Riccia'' for financial support and acknowledges the kind hospitality of the Department of Mathematics of Uppsala University during the realization of this work.
MG was partially supported by the project PRIN 2020 \lq\lq Squarefree Gr\"obner degenerations, special varieties and related topics\rq\rq~(MUR, project number 2020355B8Y).
MDZ thanks the organizers of the \textit{XIII Workshop on Geometric Correspondences of Gauge Theories}, held at the Insitute for Geometry and Physics in Trieste, Italy, in June 2023 as well as the SwissMap research station in Les Diableret for hospitality during the completion of this work.

\newpage

\appendix\label{appendix A}

\section{Atlases for the crepant resolutions of Du Val Singularities}
  \label{sec:duvalres}
Although the theory of Du Val singularities is classic and their crepant
resolution is very well understood \cite{reidyoungpersonsing}, we present
here special atlases on the resolutions which will help us in the
computations of the main text. We tackle the $A_k,D_4,D_5,E_6,E_7,E_8$
cases explicitly (all the other $D_k$ cases can be examined analogously).
We will start recalling the definition of crepant resolution of a
singularity. 
\begin{definition}
Let $X$ be a quasi-projective variety and let $\omega_X$ be its canonical bundle. A crepant resolution of singularities is the datum of a birational  projective morphism $\varepsilon:\tilde X \rightarrow X$ such that $\tilde X$ is smooth and $\omega_{\tilde X}= \varepsilon^* \omega_X$. Sometimes, with abuse of notation, we will call resolution the variety $\tilde X$.

We will say that a projective-birational morphism is a partial crepant resolution if $\omega_{\tilde X}= \varepsilon^* \omega_X$, without asking $\tilde X$ to be smooth.
\end{definition}

\subsection{Crepant resolution of \texorpdfstring{$A_k$}{Ak} singularities}
\label{sec:resAkres}
Let us start with the $A_k$ case. Let $Y_{A_k}$ be the surface in \Cref{tab:singularities}, rewritten as:
\begin{equation}
    xy =z^{k+1}.
\end{equation}
Its crepant resolution $\tilde{Y}_{A_k}$ is covered by $k+1$ charts $U_j\cong \mathbb C^2$ for $j=1,\ldots,k+1$. Let $\varphi:\tilde{Y}_{A_k}\rightarrow {Y}_{A_k}$ be the resolution morphism, let $\varphi_j$ be its restriction to the $j$-th chart, for $j=1,\ldots,k+1$, and let $a_j,b_j$ be affine coordinates on $U_j$. We can write (see \cite{COSTELLAZIONI} for more details) $\varphi_j$ as follows:
\begin{equation}\label{eq:resAk}
    \begin{tikzcd}[row sep=tiny]
      U_j\arrow{r}{\varphi_j=\varphi|_{U_j}} &Y_{A_k}\\
      \mbox{\scriptsize$\left(a_j,b_j\right)$}\arrow[mapsto]{r} &\mbox{\scriptsize{$(x,y,z)=\left(a_j^{k-j+2}b_j^{k+1-j},a_j^{j-1}b_j^{j},a_jb_j\right)$}}
\end{tikzcd}
\end{equation}
 The transition functions between  $(a_j,b_j)$ and $(a_{j+1},b_{j+1})$,  are obtained from \eqref{eq:resAk} and read as:
    \begin{equation}
        \label{eq:transitionfunctionsAk}
        a_{j} = \frac{1}{b_{j+1}}, \quad b_{j} = a_{j+1}b^2_{j+1}.
    \end{equation}
As expected from a resolved Du Val singularity, from
\eqref{eq:transitionfunctionsAk} we recognize the transition functions of
$\mathcal O_{\mathbb P^1}(-2)$, with $(a_j,b_{j+1})$ spanning the basis of
the line bundle and $(b_j,a_{j+1})$ spanning the fiber direction.

It is worth, for the resolution of the remaining Du Val singularities,
to further comment on \eqref{eq:transitionfunctionsAk}. Let us consider, for
example, the $Y_{A_2}$ singularity. Its resolution $\tilde{Y}_{A_2}$ is covered by three charts,
$U_1,U_2,U_3$, all isomorphic to $\mathbb C^2$, glued as \eqref{eq:transitionfunctionsAk}. We depict such charts in Figure \ref{fig:A2charts}.
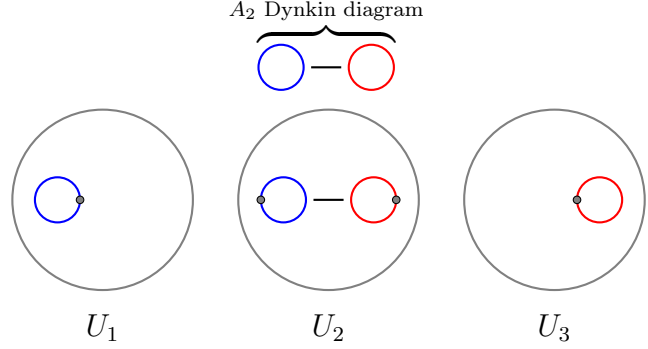
\begin{figure}[H]
    \centering
    \hspace{-0.5cm}$\overbrace{\begin{tikzpicture}
        \draw[thick,red] (-2.6,-0.3) circle (0.3);
        \draw[thick,blue] (-3.8,-0.3) circle (0.3);
        \draw[thick] (-3,-0.3)--(-3.4,-0.3);
    \end{tikzpicture}}^{A_2 \text{ Dynkin diagram}}$\\ \vspace{0.2cm}
    \begin{tikzpicture}
        \draw[thick,gray] (0,0) circle (1.2);
        \draw[thick,blue] (-0.6,0) circle (0.3);
        \draw[fill=gray] (-0.3,0) circle (0.05);
        \node at (0,-1.7) {$U_1$};
    \end{tikzpicture}
    \hspace{0.3cm}
    \begin{tikzpicture}
        \draw[thick,gray] (0,0) circle (1.2);
        \draw[thick,blue] (-0.6,0) circle (0.3);
        \draw[fill=gray] (-0.9,0) circle (0.05);
        \draw[thick] (-0.2,0)--(0.2,0);
        \draw[thick,red] (0.6,0) circle (0.3);
        \draw[fill=gray] (0.9,0) circle (0.05);
        \node at (0,-1.7) {$U_2$};
    \end{tikzpicture}
    \hspace{0.3cm}
    \begin{tikzpicture}
         \draw[thick,gray] (0,0) circle (1.2);
        \draw[thick,red] (0.6,0) circle (0.3);
        \draw[fill=gray] (0.3,0) circle (0.05);
        \node at (0,-1.7) {$U_3$};
    \end{tikzpicture}
    \hspace{0.3cm}
    \caption{Our preferred choice for an atlas covering the resolution of $A_2$. Each large gray circle represents a chart. Inside the large circles we indicate, using a color-code, which $\mathbb P^1$'s of the resolution of the $A_2$ Dynkin diagram are visible in the  considered chart. Gray dots represent points of the $\mathbb{P}^1$'s that are at infinity in the given chart.}
    \label{fig:A2charts}
\end{figure}

The transition functions
for $\tilde{Y}_{A_2}$ are
\begin{equation}
  \label{eq:transitionfunctionA2}
  a_1 = \frac{1}{b_2}, \quad b_1 = a_2 b_2^2 
\end{equation}
between the charts $U_1$ and $U_2$ and
\begin{equation}
  \label{eq:transitionfunctionA2II}
  a_2 = \frac{1}{b_3}, \quad b_2 = a_3 b_3^2 
\end{equation}
between\footnote{ \ We note that the charts $U_j$ are dense in the resolved
  $\tilde{Y}_{A_2}$, hence $U_1 \cap U_3 \neq \emptyset$. The transition function between $U_1$ and $U_3$ can be
  obtained solving \eqref{eq:transitionfunctionA2} to get $(a_2,b_2)$ in
terms of $(a_1,b_1)$, and plugging the resulting expressions \eqref{eq:transitionfunctionA2II}.} the charts $U_2$ and $U_3$.
$U_1 \cup U_2$ coincides with the total space of the normal bundle
$\mathcal O_{\mathbb P^1}(-2)$ of the blue $\mathbb P^1_{\text{blue}}$ in
Figure \ref{fig:A2charts} and the zero section is obtained
setting $b_1 = a_2 = 0$ and leaving $a_1$ and $b_2$ free to vary. In order
to obtain the total space of  $Y_{A_2}$, we need to compactify
one of the fibers of $U_1 \cup U_2 \cong   \mathcal O_{\mathbb
  P^1}(-2)$: this can be achieved using \eqref{eq:transitionfunctionA2}. \eqref{eq:transitionfunctionA2} tells us that the fiber that we want
to compactify is the one over the point\footnote{ \ We can show this as
  follows: for each fixed fiber over a point $b_2 \in \mathbb P^1_{\text{blue}}$, the point at
  infinity of the fiber is obtained sending $a_2 \to \infty$, or,
  equivalently, $b_3 \to 0$. By using the second equation of
  \eqref{eq:transitionfunctionA2II}, we see that this
  is possible just for $b_2= 0$ (since $b_3 =0$ gives $b_2 =0)$ by \eqref{eq:transitionfunctionA2II}).  Hence, the only fiber of $U_1 \cup
  U_2 \cong \mathbb P^1_{\text{blue}}$ that gets compactified (becoming
  $\mathbb P^1_{\text{red}}$ in Figure \ref{fig:A2charts}) is the one
  over $b_2 = 0 \in \mathbb P^1_{\text{blue}}$.} $b_2 = 0$ in the $\mathbb P^1_{\text{blue}}$. If we want to
compactify another fiber over, say, $b_2 = p \neq 0 \in \mathbb P^1_{\text{blue}}$
we simply shift $b_2 \to b_2 - p$ in the second equation of
\eqref{eq:transitionfunctionA2II}.
\\
In the remaining part of this appendix we will resolve the $Y_{\mathfrak
  g}$ singularities with the following recipe:
\begin{enumerate}
\item by standard blowup techniques, we will first find a partial resolution
  $\varphi_{\text{part}}: Y^{\text{part}}_{\mathfrak g}\to Y_{\mathfrak g}$
  that leaves
  a residual $A_k$ singularity.  We will then compose
  $\varphi_{\text{part}}$ with \eqref{eq:resAk} to get a (full) resolution
  $\varphi: \tilde{Y}_{\mathfrak g} \to Y_{\mathfrak g}$ of $Y_{\mathfrak g}$. In particular, one of
  the charts $U_1 \ni (a_1,b_1) $ of $\tilde{Y}_{\mathfrak g}$ will be
  isomorphic to $\mathbb C^2$ and we will have an explicit expression for $\varphi\rvert_{U_1}$
  giving $(x,y,z)$ in terms of $(a_1,b_1)$.
  \item  At this point, not all the subsets covering $\tilde{Y}_{\mathfrak
      g}$ will be (in general) isomorphic to $\mathbb C^2$: they might be
    (e.g.) defined as hypersurfaces of $\mathbb C^3$. To obtain an atlas with all
    the charts isomorphic to $\mathbb C^2$, we will glue $(a_1,b_1)$ with
    new coordinates $(a_j,b_j) \in U_j \cong \mathbb C^2$ in such a way
    that $\bigcup_{j=0}^{\text{rank}(\mathfrak g)}U_j \cong \tilde{Y}_{\mathfrak g}$. This is equivalent to
    requiring that $\bigcup_{j=0}^{\text{rank}(\mathfrak g)}U_j$ contains exactly
    $\text{rank}(\mathfrak{g})$ $\mathbb P^1$'s, each one with normal bundle
    $\mathcal O(-2)$ and intersecting each other according
    to the Dynkin diagram of $\mathfrak g$.   
\end{enumerate}

This apparently involved procedure sparks an immediate question: why bother with it, when other recipes give an equally valid resolution of ADE singularities? The answer is clear: \textit{at the end of the day, this resolution technique produces a way to read off the low-energy gauge theory quivers presented in section \ref{sec:quivers} in an automatic way, as we have explained in the example of Table \ref{tab:localmodelsE6z}.     }

\subsection{Crepant resolution of \texorpdfstring{$D_4$}{D4} singularity}
\label{sec:resD4res}
 To find the resolution $\varphi: \tilde{Y}_{D_4} \to Y_{D_4}$ we start with the first
 step of the procedure outlined at the end of \cref{sec:resAkres}. We construct a  resolution by first blowing up along the
non-Cartier\footnote{ \ We say that $D$ is a non-Cartier divisor if we can not
  find an atlas for the $Y_{D_4}$ surface such that, on each chart, we can
  describe $D$ as the zero-locus of a certain polynomial.}  divisor $x = z
= 0$ leaving a residual $A_3$ singularity. This will give the map
$\varphi_{\text{part}}: Y^{\text{part}}_{D_4}\to Y_{D_4}$ for the $Y_{D_4}$
singularity. We can then apply
\eqref{eq:resAk} to obtain a resolution of $Y_{D_4}$. 
The total space of the resolution is covered by five charts $(V_0, V_1,V_2,V_3,V_4)$. 
The resolution map restricted to $V_{0}$ is
\begin{equation}\label{eq:D4resI}
    \begin{array}{c}
         \begin{tikzcd}[row sep=tiny]
      V_{0}\cong \Set{a_{0}^2c_{0} + b_{0} + b_0^2c_0^3=0}\arrow{r}{\varphi_0=\varphi|_{V_0}} &Y_{D_4}\\
      \mbox{\scriptsize$\left(a_{0},b_{0},c_{0}\right)$}\arrow[mapsto]{r} &\mbox{\scriptsize{$(x,y,z)=\left(b_0,a_{0},b_0c_{0}\right)$}}
\end{tikzcd}\\
\end{array}
\end{equation}
The restrictions of the resolution map to $V_1, \ldots V_4$ (that cover the $\mathbb P^1$'s associated to the $A_3$ subalgebra represented by the blue, red and black $\mathbb{P}^1$'s in Figure \ref{fig:D4charts}) are 
\begin{equation}
\label{eq:D4resII}
\begin{array}{c}
         \begin{tikzcd}[row sep=tiny]
      V_j\arrow{r}{\varphi|_{V_j}} &Y_{D_4}\\
      \mbox{\scriptsize$\left(a_j,b_j\right)$}\arrow[mapsto]{r} & \scalebox{0.7}{$(x,y,z)=\left(-\frac{i c \left(a-b-2 i c^2\right)}{\sqrt{2}},\frac{a+b}{2},-c^2-\frac{1}{2} i (a-b)\right)$},
\end{tikzcd}
    \end{array} 
\end{equation}

where $(a,b,c)$ depends on $(a_j,b_j)$ according to
\eqref{eq:resAk}. Unpleasantly, the chart $V_0$ is described by three
variables, $(a_0,b_0,c_0)$ constrained by the equation $a_{0}^2c_{0} +
b_{0} + b_0^2c_0^3=0$. To fix this and obtain a more elegant presentation, we have to go through the second step of the procedure
outlined at the end of \cref{sec:resAkres}. To this end, we repackage the information contained in \eqref{eq:D4resI},
\eqref{eq:D4resII} constructing  a local geometry  consisting of four
$\mathbb P^1$, intersecting according to the $D_4$ Dynkin diagram, each
with normal bundle $\mathcal O(-2)$, and identifying it with $\tilde{Y}_{D_4}$. This is achieved considering five charts:
$U_0,...,U_4$, with $U_j \cong V_j$ for $j>0$, and $U_{0} \cong \mathbb
C^2 \ni (a_0,b_0)$. We then glue together  $U_j$, with $j>0$, according to
\eqref{eq:transitionfunctionsAk}, and glue $U_0$ to $U_2$ with
the following transition function:
\begin{equation}
  \label{eq:newpatchD4}
  a_0 = \frac{1}{b_2}, \quad b_0 = b_2^2(a_2-i),
\end{equation}
with $i^2 = -1$.
We have depicted our construction in Figure \ref{fig:D4charts}: in the upper part we depict the usual $D_4$ Dynkin diagram, and below the individual charts are represented, along with the curves that are visible in each chart. Points at infinity in some chart are represented by marked dashed dots in the corresponding $\mathbb{P}^1$.
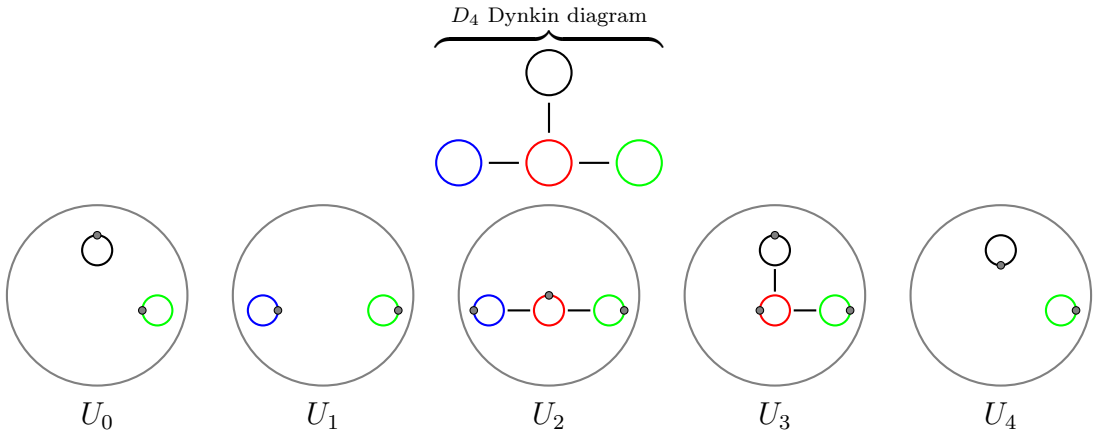
\begin{figure}[H]
    \centering
    $\overbrace{\begin{tikzpicture}
        \draw[thick,red] (0,-0.3) circle (0.3);
        \draw[thick] (0,0.9) circle (0.3);
        \draw[thick,blue] (-1.2,-0.3) circle (0.3);
        \draw[thick,green] (1.2,-0.3) circle (0.3);
        \draw[thick] (0,0.5)--(0,0.1);
        \draw[thick] (-0.4,-0.3)--(-0.8,-0.3);
        \draw[thick] (0.4,-0.3)--(0.8,-0.3);
    \end{tikzpicture}}^{D_4 \text{ Dynkin diagram}}$\\\vspace{0.2cm}
    \begin{tikzpicture}
        \draw[thick,gray] (0,-0.1) circle (1.2);
        \draw[thick] (0,0.5) circle (0.2);
        \draw[fill=gray] (0,0.7) circle (0.05);
        \draw[thick,green] (0.8,-0.3) circle (0.2);
        \draw[fill=gray] (0.6,-0.3) circle (0.05);
        \node at (0,-1.7) {$U_0$};
    \end{tikzpicture}
    \hspace{0.3cm}
    \begin{tikzpicture}
        \draw[thick,gray] (0,-0.1) circle (1.2);
        \draw[thick,green] (0.8,-0.3) circle (0.2);
        \draw[fill=gray] (1,-0.3) circle (0.05);
        \draw[thick,blue] (-0.8,-0.3) circle (0.2);
        \draw[fill=gray] (-0.6,-0.3) circle (0.05);
        \node at (0,-1.7) {$U_1$};
    \end{tikzpicture}
    \hspace{0.3cm}
    \begin{tikzpicture}
        \draw[thick,gray] (0,-0.1) circle (1.2);
        \draw[thick,red] (0,-0.3) circle (0.2);
        \draw[fill=gray] (0,-0.1) circle (0.05);
        \draw[thick,green] (0.8,-0.3) circle (0.2);
        \draw[fill=gray] (1,-0.3) circle (0.05);
        \draw[thick,blue] (-0.8,-0.3) circle (0.2);
        \draw[fill=gray] (-1,-0.3) circle (0.05);
        \draw[thick] (0.25,-0.3)--(0.55,-0.3);
        \draw[thick] (-0.55,-0.3)--(-0.25,-0.3);
        \node at (0,-1.7) {$U_2$};
    \end{tikzpicture}
    \hspace{0.3cm}
    \begin{tikzpicture}
        \draw[thick,gray] (0,-0.1) circle (1.2);
        \draw[thick,red] (0,-0.3) circle (0.2);
        \draw[fill=gray] (-0.2,-0.3) circle (0.05);
        \draw[thick] (0,0.5) circle (0.2);
        \draw[fill=gray] (0,0.7) circle (0.05);
        \draw[thick,green] (0.8,-0.3) circle (0.2);
        \draw[fill=gray] (1,-0.3) circle (0.05);
        \draw[thick] (0,0.25)--(0,-0.05);
        \draw[thick] (0.55,-0.3)--(0.25,-0.3);
        \node at (0,-1.7) {$U_3$};
    \end{tikzpicture}
    \hspace{0.3cm}
    \begin{tikzpicture}
        \draw[thick,gray] (0,-0.1) circle (1.2);
        \draw[thick] (0,0.5) circle (0.2);
        \draw[fill=gray] (0,0.3) circle (0.05);
        \draw[thick,green] (0.8,-0.3) circle (0.2);
        \draw[fill=gray] (1,-0.3) circle (0.05);
        \node at (0,-1.7) {$U_4$};
    \end{tikzpicture}
    \caption{Our preferred choice for an atlas covering the resolution of $D_4$. Each large gray circle represents a chart. Inside the large circles we indicate, using a color-code, which $\mathbb P^1$'s of the resolution of the $D_4$ Dynkin diagram are visible in the  considered chart. Gray dots represent points of the $\mathbb{P}^1$'s that are at infinity in the given chart.}
    \label{fig:D4charts}
\end{figure}

We note that, since in the second equation of \eqref{eq:newpatchD4}
appears the combination $a_2 - i$, the fiber of the normal
bundle of the $\mathbb P^1_{\text{red}}$ (corresponding to the red node of \cref{fig:D4charts}) over the point $a_2 = i \in \mathbb
P^1_{\text{red}}$ is compactified by \eqref{eq:newpatchD4}. To conclude, we
give the blowup equations in the new
coordinates $U_0,...,U_4$. The blowup maps $\varphi\rvert_{U_j}$ with $i > 0$ coincide
with the ones on the $V_j$. To obtain the blowup map over $U_0$ we insert
\eqref{eq:newpatchD4} inside the blowup map $\varphi\rvert_{U_2}$, obtaining
\begin{equation}\label{eq:newpatchD4resmap}
    \begin{array}{c}
         \begin{tikzcd}[row sep=tiny]
      U_{0}\cong \mathbb C^2\arrow{r}{\varphi_0=\varphi|_{V_0}} &Y_{D_4}\\
      \mbox{\scriptsize$\left(a_{0},b_{0}\right)$}\arrow[mapsto]{r} &\mbox{\scriptsize{$(x,y,z)=\left(-\frac{i a_0 b_0^2 \left(a_0^2
   b_0+i\right){}^2}{\sqrt{2}},\frac{1}{2} b_0
   \left(a_0^2 b_0+i\right) \left(a_0^2 b_0+2
   i\right),\frac{1}{2} a_0^2 b_0^2 \left(1-i
   a_0^2 b_0\right)\right)$}}
\end{tikzcd}\\
\end{array}
\end{equation}
We note that all the $\varphi \rvert_{U_j}$ can be obtained also considering
the expression $\varphi \rvert_{U_1}$ and 
substituting, inside $\varphi \rvert_{U_1}$, $(a_1,b_1)$ with $(a_j,b_j)$
using the transition functions.
\subsection{Crepant resolution of \texorpdfstring{$D_5$}{D5} singularity}
\label{sec:resD5res}
We will skip, from now on, the first step of the procedure outlined at the
end of \cref{sec:resAkres}, giving
\begin{enumerate}
\item an atlas $\left\{U_j\right\}$, with $j=0,\ldots \text{rank} (\mathfrak{g})$
  such that $\left(\bigcup_{j=0}^{\text{rank}(\mathfrak g)} U_j\right) \cong
  \tilde{Y}_{\mathfrak g}$ and
\item the restriction of the resolution map  to one of the charts
  of the $U_1 \in \left\{U_j\right\}$.
\end{enumerate}
The remaining $\varphi \rvert_{U_j}$ can be obtained inserting the transition
functions inside $\varphi \rvert_{U_1}$.
\\ \indent For the $Y_{D_5}$ singularity, we take the $U_1$ to cover the $\mathbb P^1$ in blue\footnote{ \ We remark that also other $\mathbb{P}^1$'s can be visible in this chart. This happens also for the resolution of the exceptional singularities.} in Figure \ref{fig:D5charts} (except for the intersection point with the next $\mathbb{P}^1$).
  \begin{figure}[H]
    \centering
   \begin{tikzpicture}
        \draw[thick] (0,-0.3) circle (0.3);
        \draw[thick] (0,0.9) circle (0.3);
        \draw[thick,blue] (-1.2,-0.3) circle (0.3);
        \draw[thick] (1.2,-0.3) circle (0.3);
        \draw[thick] (2.4,-0.3) circle (0.3);
        \draw[thick] (0,0.5)--(0,0.1);
        \draw[thick] (-0.4,-0.3)--(-0.8,-0.3);
        \draw[thick] (0.4,-0.3)--(0.8,-0.3);
        \draw[thick] (1.6,-0.3)--(2.0,-0.3);
        \draw[dashed] (-1.7,-0.8)--(-1.7,0.2)--(-0.5,0.2)--(-0.5,1.4)--(0.5,1.4)--(0.5,-0.8)--(-1.7,-0.8);
    \end{tikzpicture}
    \caption{$D_5$ Dynkin diagram. We have highlighted in blue the $\mathbb{P}^1$ that is partially covered by $U_1$ defined in the text, as well as the $A_3$ subalgebra covered by $U_j$, $j=1,\ldots,4$, encircling it with a dashed line.}
    \label{fig:D5charts}
\end{figure}

  The blowup map $\varphi \rvert_{U_1}$ is
\begin{equation}\label{eq:newpatchD5resmap}
    \begin{array}{c}
         \begin{tikzcd}[row sep=tiny]
      U_{1}\cong \mathbb C^2\arrow{r}{\varphi\rvert_{U_1}} &Y_{D_5}\\
      \mbox{\scriptsize$\left(a_{1},b_{1}\right)$}\arrow[mapsto]{r} &\mbox{\scriptsize{$(x,y,z)=\left(x\to \frac{1}{16} b^2 \left(1+i a^2 b\right) \left(a^2 b+i\right)^3,y\to \frac{1}{4} i a b^2 \left(a^2
   b+i\right)^2,z\to \frac{1}{4} i b \left(a^2 b+i\right)^2\right)$}}
\end{tikzcd}\\
\end{array}
\end{equation}
The $U_j$, with $j=1,\ldots,4$ cover the $\mathbb P^1$'s associated to the $A_3$ subalgebra of
the $D_5$ Dynkin diagram enclosed by dashed lines in Figure \ref{fig:D5charts}. We glue $U_j$, with $j=1,\ldots,4$ together using 
\eqref{eq:transitionfunctionsAk}. The remaining two rightmost nodes of the
$D_5$ Dynkin diagram are covered using coordinates
$(a_0,b_0) \in U_0$ and $(a_{-1},b_{-1}) \in U_1$ glued with
\begin{equation}
  \label{eq:transfuncttailD5}
  a_0= \frac{1}{b_2}, \quad b_0 = (i + a_2)b_2^2,
\end{equation}
with $i$ the imaginary unit and 
\begin{equation}
  \label{eq:transfuncttailD5II}
  a_{-1}= \frac{1}{b_0}, \quad b_{-1} = a_0b_0^2.
\end{equation}

\subsection{Crepant resolution of the \texorpdfstring{$E_6$}{E6}, \texorpdfstring{$E_7$}{E7} and \texorpdfstring{$E_8$}{E8} singularities}
\label{sec:resE6E7E8res}
\indent For the $Y_{E_6}$ singularity, we take $U_1$ to cover the blue $\mathbb P^1$ of figure \ref{fig:E6charts} (except the intersection point with the next $\mathbb{P}^1$).
    \begin{figure}[H]
    \centering
   \begin{tikzpicture}
        \draw[thick] (0,-0.3) circle (0.3);
        \draw[thick] (0,0.9) circle (0.3);
        \draw[thick] (-1.2,-0.3) circle (0.3);
        \draw[thick,blue] (-2.4,-0.3) circle (0.3);
        \draw[thick] (1.2,-0.3) circle (0.3);
        \draw[thick] (2.4,-0.3) circle (0.3);
        \draw[thick] (0,0.5)--(0,0.1);
        \draw[thick] (-0.4,-0.3)--(-0.8,-0.3);
        \draw[thick] (0.4,-0.3)--(0.8,-0.3);
        \draw[thick] (1.6,-0.3)--(2.0,-0.3);
        \draw[thick] (-2.0,-0.3)--(-1.6,-0.3);
        \draw[dashed] (-2.9,-0.8)--(-2.9,0.2)--(2.9,0.2)--(2.9,-0.8)--(-2.9,-0.8);
    \end{tikzpicture}
    \caption{$E_6$ Dynkin diagram. We have highlighted in blue the $\mathbb{P}^1$ that is partially covered by $U_1$ defined in the text, as well as the $A_5$ subalgebra covered by $U_j$, $j=1,\ldots,6$, encircling it with a dashed line.}
    \label{fig:E6charts}
\end{figure}
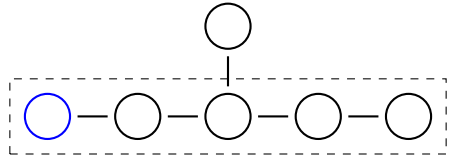

  The blowup map $\varphi
\rvert_{U_1}$ is 
\begin{equation}\label{eq:newpatchE6resmap}
    \begin{array}{c}
         \begin{tikzcd}[row sep=tiny]
      U_{1}\cong \mathbb C^2\arrow{r}{\varphi\rvert_{U_1}} &Y_{E_6}\\
      \mbox{\scriptsize$\left(a_{1},b_{1}\right)$}\arrow[mapsto]{r} &\mbox{\scriptsize{$(x,y,z)=\left(\frac{1}{4} i b_1^2 \left(a_1^3 b_1^2-1\right)^3
   \left(a_1^3 b_1^2+1\right),\frac{1}{2^{2/3}}a_1 b_1^2 \left(a_1^3
   b_1^2-1\right)^2,\frac{1}{2} b_1
   \left(a_1^3 b_1^2-1\right)^2\right)$}}
\end{tikzcd}\\
\end{array}
\end{equation}
The transition functions for the $U_j$, with $j=1,\ldots,6$, are
\eqref{eq:transitionfunctionsAk} (covering the $\mathbb P^1$'s
corresponding to the $A_5$ subalgebra of
the $E_6$ Dynkin diagram). The remaining node is covered using coordinates
$(a_0,b_0) \in U_0$ defined as
\begin{equation}
  \label{eq:transfuncttailE6}
  a_0= \frac{1}{b_3}, \quad b_0 = (a_3-1)b_3^2.
\end{equation}
\\\indent For the $Y_{E_7}$ singularity, we take $U_1$ to cover the blue $\mathbb P^1$ of figure \ref{fig:E7charts} (except the intersection point with the next $\mathbb{P}^1$).
  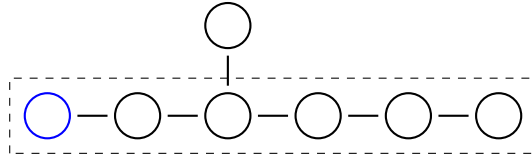
\begin{figure}[H]
    \centering
   \begin{tikzpicture}
        \draw[thick] (0,-0.3) circle (0.3);
        \draw[thick] (0,0.9) circle (0.3);
        \draw[thick] (-1.2,-0.3) circle (0.3);
        \draw[thick,blue] (-2.4,-0.3) circle (0.3);
        \draw[thick] (1.2,-0.3) circle (0.3);
        \draw[thick] (2.4,-0.3) circle (0.3);
        \draw[thick] (3.6,-0.3) circle (0.3);
        \draw[thick] (0,0.5)--(0,0.1);
        \draw[thick] (-0.4,-0.3)--(-0.8,-0.3);
        \draw[thick] (0.4,-0.3)--(0.8,-0.3);
        \draw[thick] (1.6,-0.3)--(2.0,-0.3);
        \draw[thick] (-2.0,-0.3)--(-1.6,-0.3);
        \draw[thick] (2.8,-0.3)--(3.2,-0.3);
        \draw[dashed] (-2.9,-0.8)--(-2.9,0.2)--(4.1,0.2)--(4.1,-0.8)--(-2.9,-0.8);
    \end{tikzpicture}
    \caption{$E_7$ Dynkin diagram. We have highlighted in blue the $\mathbb{P}^1$ that is partially covered by $U_1$ defined in the text, as well as the $A_6$ subalgebra covered by $U_j$, $j=1,\ldots,7$, encircling it with a dashed line.}
    \label{fig:E7charts}
\end{figure}

The blowup map $\varphi
\rvert_{U_1}$ is 
\begin{equation}\label{eq:newpatchE7resmap}
    \begin{array}{c}
         \begin{tikzcd}[row sep=tiny]
      U_{1}\cong \mathbb C^2\arrow{r}{\varphi\rvert_{U_1}} &Y_{E_7}\\
      \mbox{\scriptsize$\left(a_{1},b_{1}\right)$}\arrow[mapsto]{r} &\mbox{\scriptsize{$(x,y,z)=\left(\left(-\frac{1}{4}+\frac{i}{4}\right) b_1^3
   \left(a_1^3 b_1^2-1\right)^5,\frac{1}{2} i b_1^2
   \left(a_1^3 b_1^2-1\right)^3,\frac{a_1 b_1^2 \left(a_1^3
   b_1^2-1\right)^2}{2^{2/3}}\right)$}}
\end{tikzcd}\\
\end{array}
\end{equation}
The transition function for the $U_j$, with $j=1,\ldots,7$, are
\eqref{eq:transitionfunctionsAk} (covering the $\mathbb P^1$'s associated to
the $A_6$ subalgebra of
the $E_7$ algebra). The remaining uppermost node is covered using coordinates
$(a_0,b_0) \in U_0$ defined as
\begin{equation}
  \label{eq:transfuncttailE7}
  a_0= \frac{1}{b_3}, \quad b_0 = (a_3-1)b_3^2.
\end{equation}
\\ \indent For the $Y_{E_8}$ singularity, we take $U_1$ to cover the blue $\mathbb P^1$ of figure \ref{fig:E8charts} (except the intersection point with the next $\mathbb{P}^1$).
  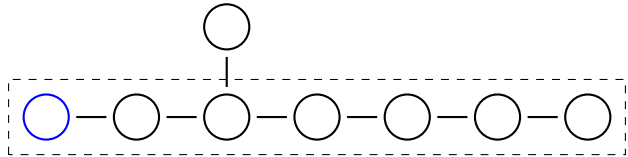
\begin{figure}[H]
    \centering
   \begin{tikzpicture}
        \draw[thick] (0,-0.3) circle (0.3);
        \draw[thick] (0,0.9) circle (0.3);
        \draw[thick] (-1.2,-0.3) circle (0.3);
        \draw[thick,blue] (-2.4,-0.3) circle (0.3);
        \draw[thick] (1.2,-0.3) circle (0.3);
        \draw[thick] (2.4,-0.3) circle (0.3);
        \draw[thick] (3.6,-0.3) circle (0.3);
        \draw[thick] (4.8,-0.3) circle (0.3);
        \draw[thick] (0,0.5)--(0,0.1);
        \draw[thick] (-0.4,-0.3)--(-0.8,-0.3);
        \draw[thick] (0.4,-0.3)--(0.8,-0.3);
        \draw[thick] (1.6,-0.3)--(2.0,-0.3);
        \draw[thick] (-2.0,-0.3)--(-1.6,-0.3);
        \draw[thick] (2.8,-0.3)--(3.2,-0.3);
        \draw[thick] (4.0,-0.3)--(4.4,-0.3);
        \draw[dashed] (-2.9,-0.8)--(-2.9,0.2)--(5.3,0.2)--(5.3,-0.8)--(-2.9,-0.8);
    \end{tikzpicture}
    \caption{$E_8$ Dynkin diagram. We have highlighted in blue the $\mathbb{P}^1$ that is partially covered by $U_1$ defined in the text, as well as the $A_7$ subalgebra covered by $U_j$, $j=1,\ldots,8$, encircling it with a dashed line.}
    \label{fig:E8charts}
\end{figure}

The blowup map $\varphi
\rvert_{U_1}$ is 
\begin{equation}\label{eq:newpatchE8resmap}
    \begin{array}{c}
         \begin{tikzcd}[row sep=tiny]
      U_{1}\cong \mathbb C^2\arrow{r}{\varphi\rvert_{U_1}} &Y_{E_8}\\
      \mbox{\scriptsize$\left(a_{1},b_{1}\right)$}\arrow[mapsto]{r} &\mbox{\scriptsize{$(x,y,z)=\left(\left(-\frac{1}{8}-\frac{i}{8}\right) b_1^5
   \left(a_1^3 b_1^2-1\right)^8,\frac{i a_1 b_1^4
   \left(a_1^3 b_1^2-1\right)^5}{2\
   2^{2/3}},\frac{1}{2} i b_1^2 \left(a_1^3
   b_1^2-1\right)^3\right)$}}
\end{tikzcd}\\
\end{array}
\end{equation}
The transition function for the $U_j$, with $j=1,\ldots,8$, are
\eqref{eq:transitionfunctionsAk} (covering the $\mathbb P^1$'s
corresponding to the $A_7$ subalgebra of
the $E_8$ algebra). The remaining uppermost node of the $E_8$ Dynkin diagram is covered using coordinates
$(a_0,b_0) \in U_0$ defined as
\begin{equation}
  \label{eq:transfuncttailE8}
  a_0= \frac{1}{b_3}, \quad b_0 = (a_3-1)b_3^2.
\end{equation}

\section{Derivation of the toric fan for transversal families of A singularities}
\label{sec:appendixtoriclocalmodels}
In this Appendix we derive the properties of the local toric models that we introduced in \cref{sec:toriclocalmodels}, giving also a formula to extract the normal bundle of the compact curves contained in their resolutions. We are interested in singularities of the form
\begin{equation}
\label{eq:Xhkv}
X_{hk}=\Set{(a,b,u,v)\in\mathbb C^4|uv-a^hb^k=0},
\end{equation}
for $h,k \ge 0$.
The threefold $X_{hk}$ is a toric variety with a family of $A_{k-1}$ singularities on the $a$ axis and a family of $A_{h-1}$ singularities on the $b$ axis. 
We can give the following embedding that endows $X_{hk}$ with the structure of a toric variety:
\[ \begin{tikzcd}[row sep =tiny]
    (\C^*)^3\arrow{r}{\phi} & X_{hk}\\
    (t_1,t_2,t_3)\arrow[mapsto]{r}& \left(t_1^{k-h}t_2^{-1}t_3^h,t_2,t_1^{-1}t_3,t_1 \right).
\end{tikzcd}\]

\noindent Let $v_1,\ldots,v_4\in\Z^3$ be the vectors of the exponents in the entries of $\phi$, i.e.
\renewcommand{\arraystretch}{1}
\[v_1=\begin{pmatrix}
     k-h \\ -1\\ h 
\end{pmatrix},\ v_2=\begin{pmatrix}
     0 \\ 1\\ 0
\end{pmatrix},\ v_3=\begin{pmatrix}
     -1 \\ 0\\ 1 
\end{pmatrix},\ v_4=\begin{pmatrix}
     1 \\ 0\\ 0 
\end{pmatrix}.\]
Let also $M$ denote the lattice $M\cong\Z^3$ generated by the $v_i$'s, i.e. $M=\langle v_1,v_2,v_3,v_4\rangle_{\Z} \subset \Z^3$. Then, we have
\[X_{hk}=\Spec(\C[\sigma^{\vee}]),\]
where $\sigma^{\vee}\subset M$ is the semigroup generated by $v_1,\ldots,v_4$.\\ 
\indent In order to better understand the geometry of $X_{hk}$ and its crepant resolutions, we will focus now on the fan $ {\Sigma}_{hk}$ of $X_{hk}$. Recall that ${ \Sigma}_{hk}$ is a set of cones in $\R^3\cong\R\otimes_\Z N$, where $N=M^\vee$ is the dual lattice of $M$. Since $X_{hk}$ is an affine variety, $ { \Sigma}_{hk}$ contains only one cone $ { \sigma}_{\max} $ of maximal dimension, namely 3. Now, a direct computation shows that the set of primitive generators of the rays, i.e. one-dimensional cones, $\tilde\rho_i\in \Sigma_{hk}^{(1)}\subset\Sigma_{hk}$, is the following
\begin{equation}
 w_1=\begin{pmatrix}
     0 \\ 0\\ 1 
\end{pmatrix},
\   w_2=\begin{pmatrix}
         0 \\ h\\ 1 
\end{pmatrix},
\  w_3=\begin{pmatrix}
  1 \\ k\\ 1
\end{pmatrix},
\  w_4=\begin{pmatrix}
    1 \\ 0\\ 1 
\end{pmatrix}.
\end{equation}
Notice that the $w_i$ are also generators of $\sigma_{\max}$ and that the third coordinate of all the $w_i$ equals to one. This ensures that $X_{hk}$ is a CY variety. Furthermore, from the fan we directly see that there are two smooth toric lines respectively corresponding to the cones generated by $ w_1, w_4$ and $ w_2,  w_3$, and two singular toric lines respectively corresponding to the cones generated by $  w_1,  w_2$ and $  w_3,  w_4$.

 Let $P_{hk}\subset \R^2$ be the toric diagram of $X_{hk}$, which encodes all the information about the geometry of the singularity and of its crepant resolutions (see \Cref{fig:polytope}).

We now compute the normal bundles of the compact toric curves appearing in the crepant resolutions of $X_{hk}$. This can, in principle, be deduced from the general theory of toric CY threefold \cite{Reid1983DecompositionOT}, but we show an ad hoc computation in our setting.
 
\begin{proposition}
    Let $\tilde X_{hk}$ be a crepant resolution of $ X_{hk}$ and let $C\subset \tilde X_{hk}$ be a compact toric curve. In particular $C\cong\Pj^1$. Then, for the normal bundle of $C$ in $\tilde X_{hk} $ there are only two possibilities, namely
    \[N_{C/X}\cong\Calo_{\Pj^1}\oplus\Calo_{\Pj^1}(-2)  \mbox{ or } \mathcal \mathcal N_{C/X}\cong\Calo_{\Pj^1}(-1)\oplus\Calo_{\Pj^1}(-1).\]
\end{proposition}
\begin{proof}
    Locally, near the cones defining $C$ (generated by $(v_2,v_4)$ and $(\tilde{v_1},\tilde{v_4})$), $P_{hk}$ may be of two kinds, namely
    \begin{center}
        $
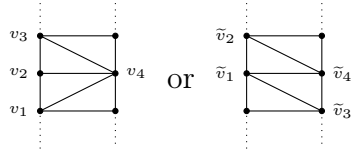
\begin{matrix}
            \begin{tikzpicture}
            \draw[dotted] (0,1)--(0,1.5)(1,1)--(1,1.5) (0,-0.5)--(0,0)(1,0)--(1,-0.5);
                \draw (0,0)--(1,0)--(1,1)--(0,1)--cycle;
        \filldraw (0,0) circle (1pt);
        \filldraw (1,0) circle (1pt);
        \filldraw (0,1) circle (1pt);
        \filldraw (1,1) circle (1pt);
        \filldraw (0,0.5) circle (1pt);
        \filldraw (1,0.5) circle (1pt);
        \draw (0,0.5) --(1,0.5);
        \draw (0,1) --(1,0.5);
        \draw (0,0) --(1,0.5);
        \node[left] at (0,0) {\tiny $v_1$};
        \node[left] at (0,0.5) {\tiny $v_2$};
        \node[left] at (0,1) {\tiny $v_3$};
        \node[right] at (1,0.5) {\tiny $v_4$};
        \end{tikzpicture}
        \end{matrix}$
        or 
        $\begin{matrix}
            \begin{tikzpicture} 
            \draw[dotted] (0,1)--(0,1.5)(1,1)--(1,1.5) (0,-0.5)--(0,0)(1,0)--(1,-0.5);
                \draw (0,0)--(1,0)--(1,1)--(0,1)--cycle;
        \filldraw (0,0) circle (1pt);
        \filldraw (1,0) circle (1pt);
        \filldraw (0,1) circle (1pt);
        \filldraw (1,1) circle (1pt);
        \filldraw (0,0.5) circle (1pt);
        \filldraw (1,0.5) circle (1pt);
        \draw (0,0.5) --(1,0.5);
        \draw (0,0.5) --(1,0);
        \draw (0,1) --(1,0.5);
        \node[right] at (1,0) {\tiny $\widetilde v_3$};
        \node[left] at (0,0.5) {\tiny $\widetilde v_1$};
        \node[left] at (0,1) {\tiny $\widetilde v_2$};
        \node[right] at (1,0.5) {\tiny $\widetilde v_4$};
        \end{tikzpicture}
        \end{matrix}$
    \end{center}
    Now, the statement is a consequence \cite{Reid1983DecompositionOT} of the following two relations:
    \[v_3+v_1+(0)\cdot v_4 +(-2)\cdot v_2=0 \mbox{ and }\widetilde v_2+\widetilde v_3+(-1)\cdot \widetilde v_1 +(-1)\cdot \widetilde v_4=0. \]
\end{proof}

\section{Explicit prepotential computation}\label{sec:prepotential}
In the main body of this work, we have adopted a top-down approach in the construction of 5d SCFTs, starting from a singular threefold and analyzing its resolution. As, by definition, the blow-up operations we have performed can be reversed, the compact divisors that appear on the origin are \textit{shrinkable} in the sense of \cite{Jefferson_2018}, and thus the geometry gives rise to a sensible 5d SCFT. It is nevertheless a useful exercise to check the consistency of the UV description employing a bottom-up point of view. Namely, we can start from the quiver gauge theories listed in section \ref{sec:quivers} and compute their prepotential $\mathcal{F}$ in terms of the scalars appearing in the vector multiplets $\phi_i$: as is well known \cite{Intriligator_1997,Jefferson:2017ahm}, one should check that there exists a continuous path connecting the origin of the scalars moduli space to every point that satisfies:
\begin{equation}\label{prepotential conditions}
    \pdv{\mathcal{F}}{\phi_i}\geq 0 \quad \text{and} \quad \pdv{\mathcal{F}}{\phi_i\partial\phi_j} \text{ positive definite,}
\end{equation}
namely where the string monopole tensions are positive (or at most null) and the metric is positive definite.\\

\indent We show this procedure at work in the case of the quiver theory arising from $X_{E_6}^z$: in general, it is not feasible to solve the problem analitically. We will thus resort to finding numerically at least one point satisfying \eqref{prepotential conditions} and then show that it lies inside a cone that is connected to the origin of the moduli space, where all the volumes of the compact divisors, parametrized by the scalars, vanish. This is of course not a complete proof, yet a strong hint that the 5d quiver theory is sound.\\
\indent The quiver for $X_{E_6}^z$ appearing on page 34 has 5 independent scalars coming from the gauge multiplets, let us write them as:
\begin{equation}
    \Phi = \{\underbrace{a_1,a_2,a_3}_{\mathfrak{su}(3)},\underbrace{b_1,b_2}_{\mathfrak{su}(2) },\underbrace{c_1,c_2}_{\mathfrak{su}(2)},\underbrace{d_1,d_2}_{\mathfrak{su}(2)}\}, \quad \text{with} \quad \sum_i a_i = \sum_j b_j = \sum_k c_k = \sum_l d_l = 0, 
\end{equation}
with the text below the brackets indicating the corresponding node.\\
The generic IMS prepotential reads \cite{Intriligator_1997}:
\begin{equation}
\mathcal{F}(\phi)=\frac{1}{2} m_0 h_{i j} \phi_i \phi_j+\frac{\kappa}{6} d_{i j k} \phi_i \phi_j \phi_k+\frac{1}{12}\left(\sum_{r \in \operatorname{roots}}|r \cdot \phi|^3-\sum_f \sum_{w \in R_f}\left|w \cdot \phi+m_f\right|^3\right),
\end{equation}
where $m_0$ is the inverse squared gauge coupling of the various factors, $m_f$ are the flavor masses, $\kappa$ are the Chern-Simons levels, and the sums run along the roots $r$ of the quiver gauge group and the weights $w$ of the hypermultiplet representations $R_f$.\\
We set:
\begin{itemize}
    \item all the masses to zero, as adding a mass can do nothing but relax the conditions \eqref{prepotential conditions}; namely, if we can find a point in the moduli space satisfying \eqref{prepotential conditions} with the masses set to zero, we are guaranteed to satisfy the bound also with the masses turned on. Furthermore, we are interested in investigating the SCFT limit, where the masses go to zero;
    \item the gauge coupling to infinity, as we want to approach the SCFT point;
    \item the Chern-Simons levels to zero, as we have shown that we can always resolve the geometry in such a way as to ensure this happens.
\end{itemize} 
With these assumptions, the prepotential for the $X_{E_6}^z$ quiver gauge theory reads, in terms of the five independent scalars:
\begin{equation}
    \begin{array}{rl}
\scalemath{0.75}{ \mathcal{F}(a_1,a_2,b_1,c_1,d_1) =\frac{1}{12}\bigg[} &
\scalemath{0.75}{\underbrace{\left| a_1-a_2\right| ^3+\left| 2 a_1+a_2\right| ^3+\left| a_1+2 a_2\right| ^3+8 \left| b_1\right| ^3+8 \left| c_1\right| ^3+8 \left| d_1\right| ^3}_{\text{roots contribution}}+}\\
&\scalemath{0.75}{\underbrace{-\left| -a_1-a_2-b_1\right| ^3-\left| -a_1-a_2+b_1\right| ^3-\left| -a_1-a_2-c_1\right| ^3-\left| a_2+d_1\right| ^3}_{\text{weights contribution}}+}\\
&\scalemath{0.75}{\underbrace{-\left| -a_1-a_2+c_1\right| ^3-\left| -a_1-a_2-d_1\right| ^3-\left| -a_1-a_2+d_1\right| ^3-\left| a_1-b_1\right| ^3}_{\text{weights contribution}}+}\\
&\scalemath{0.75}{\underbrace{-\left| a_1+b_1\right| ^3-\left| a_1-c_1\right| ^3-\left| a_1+c_1\right| ^3-\left| a_1-d_1\right| ^3-\left| a_1+d_1\right| ^3}_{\text{weights contribution}}+ }\\
&\scalemath{0.75}{\underbrace{-\left| a_2-b_1\right| ^3-\left| a_2+b_1\right| ^3-\left| a_2-c_1\right| ^3-\left| a_2+c_1\right| ^3-\left| a_2-d_1\right| ^3}_{\text{weights contribution}}+ }\\[1ex]
&\scalemath{0.75}{ \underbrace{ -2 \left| b_1\right| ^3-2 \left| c_1\right| ^3-2 \left| d_1\right| ^3}_{\text{flavor contribution}}\ \bigg].}
\end{array}
\end{equation}
Employing numeric techniques, it is easy to find a continuous path connecting the origin of the moduli space to a point that satisfies \eqref{prepotential conditions}:
\begin{equation}
    \{a_1,a_2,b_1,c_1,d_1\} = \left\{\frac{151 }{100}\lambda ,\frac{103  }{100}\lambda,-\frac{26}{25} \lambda ,-\frac{21}{20} \lambda ,-\frac{11}{10}  \lambda \right\},
\end{equation}
where $\lambda \in [0,1]$ parameterizes the path.


\begin{thebibliography}{100}

\bibitem{Witten:1995ex}
E.~Witten, ``{String theory dynamics in various dimensions},'' {\em Nucl. Phys.
  B} {\bf 443} (1995) 85--126, \href{http://arXiv.org/abs/hep-th/9503124}{{\tt
  hep-th/9503124}}.

\bibitem{Strominger:1995ac}
A.~Strominger, ``{Open p-branes},'' {\em Phys. Lett. B} {\bf 383} (1996)
  44--47, \href{http://arXiv.org/abs/hep-th/9512059}{{\tt hep-th/9512059}}.

\bibitem{Witten:1995em}
E.~Witten, ``{Five-branes and M theory on an orbifold},'' {\em Nucl. Phys. B}
  {\bf 463} (1996) 383--397, \href{http://arXiv.org/abs/hep-th/9512219}{{\tt
  hep-th/9512219}}.

\bibitem{Ganor:1996mu}
O.~J. Ganor and A.~Hanany, ``{Small E(8) instantons and tensionless noncritical
  strings},'' {\em Nucl. Phys. B} {\bf 474} (1996) 122--140,
  \href{http://arXiv.org/abs/hep-th/9602120}{{\tt hep-th/9602120}}.

\bibitem{Seiberg:1996qx}
N.~Seiberg, ``{Nontrivial fixed points of the renormalization group in
  six-dimensions},'' {\em Phys. Lett. B} {\bf 390} (1997) 169--171,
  \href{http://arXiv.org/abs/hep-th/9609161}{{\tt hep-th/9609161}}.

\bibitem{Seiberg1996}
N.~Seiberg, ``Five dimensional susy field theories, non-trivial fixed points
  and string dynamics,'' {\em Physics Letters B} {\bf 388} (Nov, 1996)
  753–760.

\bibitem{Morrison_1997}
D.~R. Morrison and N.~Seiberg, ``Extremal transitions and five-dimensional
  supersymmetric field theories,'' {\em Nuclear Physics B} {\bf 483} (Jan,
  1997) 229–247.

\bibitem{Douglas:1996xp}
M.~R. Douglas, S.~H. Katz, and C.~Vafa, ``{Small instantons, Del Pezzo surfaces
  and type I-prime theory},'' {\em Nucl. Phys. B} {\bf 497} (1997) 155--172,
  \href{http://arXiv.org/abs/hep-th/9609071}{{\tt hep-th/9609071}}.

\bibitem{Ganor:1996xd}
O.~J. Ganor, ``{Toroidal compactification of heterotic 6-d noncritical strings
  down to four-dimensions},'' {\em Nucl. Phys. B} {\bf 488} (1997) 223--235,
  \href{http://arXiv.org/abs/hep-th/9608109}{{\tt hep-th/9608109}}.

\bibitem{Ganor:1996pc}
O.~J. Ganor, D.~R. Morrison, and N.~Seiberg, ``{Branes, Calabi-Yau spaces, and
  toroidal compactification of the N=1 six-dimensional E(8) theory},'' {\em
  Nucl. Phys. B} {\bf 487} (1997) 93--127,
  \href{http://arXiv.org/abs/hep-th/9610251}{{\tt hep-th/9610251}}.

\bibitem{Gaiotto:2009we}
D.~Gaiotto, ``{N=2 dualities},'' {\em JHEP} {\bf 08} (2012) 034,
  \href{http://arXiv.org/abs/0904.2715}{{\tt 0904.2715}}.

\bibitem{Gaiotto:2009hg}
D.~Gaiotto, G.~W. Moore, and A.~Neitzke, ``{Wall-crossing, Hitchin systems, and
  the WKB approximation},'' {\em Adv. Math.} {\bf 234} (2013) 239--403,
  \href{http://arXiv.org/abs/0907.3987}{{\tt 0907.3987}}.

\bibitem{Benini:2009mz}
F.~Benini, Y.~Tachikawa, and B.~Wecht, ``{Sicilian gauge theories and N=1
  dualities},'' {\em JHEP} {\bf 01} (2010) 088,
  \href{http://arXiv.org/abs/0909.1327}{{\tt 0909.1327}}.

\bibitem{Ohmori:2015pua}
K.~Ohmori, H.~Shimizu, Y.~Tachikawa, and K.~Yonekura, ``{6d $\mathcal{N}=(1,0)$
  theories on $T^2$ and class S theories: Part I},'' {\em JHEP} {\bf 07} (2015)
  014, \href{http://arXiv.org/abs/1503.06217}{{\tt 1503.06217}}.

\bibitem{DelZotto:2015rca}
M.~Del~Zotto, C.~Vafa, and D.~Xie, ``{Geometric engineering, mirror symmetry
  and $ 6{\mathrm{d}}_{\left(1,0\right)}\to
  4{\mathrm{d}}_{\left(\mathcal{N}=2\right)} $},'' {\em JHEP} {\bf 11} (2015)
  123, \href{http://arXiv.org/abs/1504.08348}{{\tt 1504.08348}}.

\bibitem{Ohmori:2015pia}
K.~Ohmori, H.~Shimizu, Y.~Tachikawa, and K.~Yonekura, ``{6d
  $\mathcal{N}=\left(1,\;0\right) $ theories on S$^{1}$ /T$^{2}$ and class S
  theories: part II},'' {\em JHEP} {\bf 12} (2015) 131,
  \href{http://arXiv.org/abs/1508.00915}{{\tt 1508.00915}}.

\bibitem{Bhardwaj:2019fzv}
L.~Bhardwaj, P.~Jefferson, H.-C. Kim, H.-C. Tarazi, and C.~Vafa, ``{Twisted
  Circle Compactifications of 6d SCFTs},'' {\em JHEP} {\bf 12} (2020) 151,
  \href{http://arXiv.org/abs/1909.11666}{{\tt 1909.11666}}.

\bibitem{Sacchi:2021wvg}
M.~Sacchi, O.~Sela, and G.~Zafrir, ``{On the 3d compactifications of 5d SCFTs
  associated with SU(N + 1) gauge theories},'' {\em JHEP} {\bf 05} (2022) 053,
  \href{http://arXiv.org/abs/2111.12745}{{\tt 2111.12745}}.

\bibitem{Gukov:2020btk}
S.~Gukov, P.-S. Hsin, and D.~Pei, ``{Generalized global symmetries of $T[M]$
  theories. Part I},'' {\em JHEP} {\bf 04} (2021) 232,
  \href{http://arXiv.org/abs/2010.15890}{{\tt 2010.15890}}.

\bibitem{Bashmakov:2022uek}
V.~Bashmakov, M.~Del~Zotto, A.~Hasan, and J.~Kaidi, ``{Non-invertible
  symmetries of class S theories},'' {\em JHEP} {\bf 05} (2023) 225,
  \href{http://arXiv.org/abs/2211.05138}{{\tt 2211.05138}}.

\bibitem{Bashmakov:2022jtl}
V.~Bashmakov, M.~Del~Zotto, and A.~Hasan, ``{On the 6d Origin of Non-invertible
  Symmetries in 4d},'' \href{http://arXiv.org/abs/2206.07073}{{\tt
  2206.07073}}.

\bibitem{Carta:2023bqn}
F.~Carta, S.~Giacomelli, N.~Mekareeya, and A.~Mininno, ``{Comments on
  Non-invertible Symmetries in Argyres-Douglas Theories},''
  \href{http://arXiv.org/abs/2303.16216}{{\tt 2303.16216}}.

\bibitem{Bashmakov:2023kwo}
V.~Bashmakov, M.~Del~Zotto, and A.~Hasan, ``{Four-manifolds and Symmetry
  Categories of 2d CFTs},'' \href{http://arXiv.org/abs/2305.10422}{{\tt
  2305.10422}}.

\bibitem{Chen:2023qnv}
J.~Chen, W.~Cui, B.~Haghighat, and Y.-N. Wang, ``{SymTFTs and Duality Defects
  from 6d SCFTs on 4-manifolds},'' \href{http://arXiv.org/abs/2305.09734}{{\tt
  2305.09734}}.

\bibitem{Alday:2009aq}
L.~F. Alday, D.~Gaiotto, and Y.~Tachikawa, ``{Liouville Correlation Functions
  from Four-dimensional Gauge Theories},'' {\em Lett. Math. Phys.} {\bf 91}
  (2010) 167--197, \href{http://arXiv.org/abs/0906.3219}{{\tt 0906.3219}}.

\bibitem{Wyllard:2009hg}
N.~Wyllard, ``{A(N-1) conformal Toda field theory correlation functions from
  conformal N = 2 SU(N) quiver gauge theories},'' {\em JHEP} {\bf 11} (2009)
  002, \href{http://arXiv.org/abs/0907.2189}{{\tt 0907.2189}}.

\bibitem{Nekrasov:2009rc}
N.~A. Nekrasov and S.~L. Shatashvili, ``{Quantization of Integrable Systems and
  Four Dimensional Gauge Theories},'' in {\em {16th International Congress on
  Mathematical Physics}}, pp.~265--289.
\newblock 8, 2009.
\newblock \href{http://arXiv.org/abs/0908.4052}{{\tt 0908.4052}}.

\bibitem{Dimofte:2011ju}
T.~Dimofte, D.~Gaiotto, and S.~Gukov, ``{Gauge Theories Labelled by
  Three-Manifolds},'' {\em Commun. Math. Phys.} {\bf 325} (2014) 367--419,
  \href{http://arXiv.org/abs/1108.4389}{{\tt 1108.4389}}.

\bibitem{Gadde:2013sca}
A.~Gadde, S.~Gukov, and P.~Putrov, ``{Fivebranes and 4-manifolds},'' {\em Prog.
  Math.} {\bf 319} (2016) 155--245, \href{http://arXiv.org/abs/1306.4320}{{\tt
  1306.4320}}.

\bibitem{Cordova:2013cea}
C.~Cordova and D.~L. Jafferis, ``{Complex Chern-Simons from M5-branes on the
  Squashed Three-Sphere},'' {\em JHEP} {\bf 11} (2017) 119,
  \href{http://arXiv.org/abs/1305.2891}{{\tt 1305.2891}}.

\bibitem{Cordova:2016cmu}
C.~Cordova and D.~L. Jafferis, ``{Toda Theory From Six Dimensions},'' {\em
  JHEP} {\bf 12} (2017) 106, \href{http://arXiv.org/abs/1605.03997}{{\tt
  1605.03997}}.

\bibitem{Dedushenko:2017tdw}
M.~Dedushenko, S.~Gukov, and P.~Putrov, ``{Vertex algebras and 4-manifold
  invariants},'' in {\em {Nigel Hitchin's 70th Birthday Conference}}, vol.~1,
  pp.~249--318.
\newblock 5, 2017.
\newblock \href{http://arXiv.org/abs/1705.01645}{{\tt 1705.01645}}.

\bibitem{Dedushenko:2019mnd}
M.~Dedushenko and Y.~Wang, ``{4d/2d $\rightarrow $ 3d/1d: A song of protected
  operator algebras},'' \href{http://arXiv.org/abs/1912.01006}{{\tt
  1912.01006}}.

\bibitem{Cecotti:2011iy}
S.~Cecotti, C.~Cordova, and C.~Vafa, ``{Braids, Walls, and Mirrors},''
  \href{http://arXiv.org/abs/1110.2115}{{\tt 1110.2115}}.

\bibitem{Nahm:1977tg}
W.~Nahm, ``{Supersymmetries and their Representations},'' {\em Nucl. Phys. B}
  {\bf 135} (1978) 149.

\bibitem{Minwalla:1997ka}
S.~Minwalla, ``{Restrictions imposed by superconformal invariance on quantum
  field theories},'' {\em Adv. Theor. Math. Phys.} {\bf 2} (1998) 783--851,
  \href{http://arXiv.org/abs/hep-th/9712074}{{\tt hep-th/9712074}}.

\bibitem{Cordova:2016emh}
C.~Cordova, T.~T. Dumitrescu, and K.~Intriligator, ``{Multiplets of
  Superconformal Symmetry in Diverse Dimensions},'' {\em JHEP} {\bf 03} (2019)
  163, \href{http://arXiv.org/abs/1612.00809}{{\tt 1612.00809}}.

\bibitem{Jefferson:2017ahm}
P.~Jefferson, H.-C. Kim, C.~Vafa, and G.~Zafrir, ``{Towards Classification of
  5d SCFTs: Single Gauge Node},'' {\em SciPost Phys.} {\bf 14} (2023) 122,
  \href{http://arXiv.org/abs/1705.05836}{{\tt 1705.05836}}.

\bibitem{Jefferson:2018irk}
P.~Jefferson, S.~Katz, H.-C. Kim, and C.~Vafa, ``{On Geometric Classification
  of 5d SCFTs},'' {\em JHEP} {\bf 04} (2018) 103,
  \href{http://arXiv.org/abs/1801.04036}{{\tt 1801.04036}}.

\bibitem{Bhardwaj:2018vuu}
L.~Bhardwaj and P.~Jefferson, ``{Classifying 5d SCFTs via 6d SCFTs: Arbitrary
  rank},'' {\em JHEP} {\bf 10} (2019) 282,
  \href{http://arXiv.org/abs/1811.10616}{{\tt 1811.10616}}.

\bibitem{Apruzzi:2019opn}
F.~Apruzzi, C.~Lawrie, L.~Lin, S.~Sch\"afer-Nameki, and Y.-N. Wang, ``{Fibers
  add Flavor, Part I: Classification of 5d SCFTs, Flavor Symmetries and BPS
  States},'' {\em JHEP} {\bf 11} (2019) 068,
  \href{http://arXiv.org/abs/1907.05404}{{\tt 1907.05404}}.

\bibitem{Bhardwaj:2019jtr}
L.~Bhardwaj, ``{On the classification of 5d SCFTs},'' {\em JHEP} {\bf 09}
  (2020) 007, \href{http://arXiv.org/abs/1909.09635}{{\tt 1909.09635}}.

\bibitem{Bhardwaj:2020gyu}
L.~Bhardwaj and G.~Zafrir, ``{Classification of 5d $ \mathcal{N} $ = 1 gauge
  theories},'' {\em JHEP} {\bf 12} (2020) 099,
  \href{http://arXiv.org/abs/2003.04333}{{\tt 2003.04333}}.

\bibitem{Xie:2022lcm}
D.~Xie, ``{Classification of rank one 5d $\mathcal{N}=1$ and 6d $(1,0)$
  SCFTs},'' \href{http://arXiv.org/abs/2210.17324}{{\tt 2210.17324}}.

\bibitem{DMDZGSinprep}
M.~De~Marco, M.~Del~Zotto, M.~Graffeo, and A.~Sangiovanni, ``\textit{In
  preparation},''.

\bibitem{Heckman:2013pva}
J.~J. Heckman, D.~R. Morrison, and C.~Vafa, ``{On the Classification of 6D
  SCFTs and Generalized ADE Orbifolds},'' {\em JHEP} {\bf 05} (2014) 028,
  \href{http://arXiv.org/abs/1312.5746}{{\tt 1312.5746}}. [Erratum: JHEP 06,
  017 (2015)].

\bibitem{Heckman:2015bfa}
J.~J. Heckman, D.~R. Morrison, T.~Rudelius, and C.~Vafa, ``{Atomic
  Classification of 6D SCFTs},'' {\em Fortsch. Phys.} {\bf 63} (2015) 468--530,
  \href{http://arXiv.org/abs/1502.05405}{{\tt 1502.05405}}.

\bibitem{Bhardwaj:2019hhd}
L.~Bhardwaj, ``{Revisiting the classifications of 6d SCFTs and LSTs},'' {\em
  JHEP} {\bf 03} (2020) 171, \href{http://arXiv.org/abs/1903.10503}{{\tt
  1903.10503}}.

\bibitem{Cordova:2016xhm}
C.~Cordova, T.~T. Dumitrescu, and K.~Intriligator, ``{Deformations of
  Superconformal Theories},'' {\em JHEP} {\bf 11} (2016) 135,
  \href{http://arXiv.org/abs/1602.01217}{{\tt 1602.01217}}.

\bibitem{Hull:2000zn}
C.~M. Hull, ``{Strongly coupled gravity and duality},'' {\em Nucl. Phys. B}
  {\bf 583} (2000) 237--259, \href{http://arXiv.org/abs/hep-th/0004195}{{\tt
  hep-th/0004195}}.

\bibitem{Hull:2000rr}
C.~M. Hull, ``{Symmetries and compactifications of (4,0) conformal gravity},''
  {\em JHEP} {\bf 12} (2000) 007,
  \href{http://arXiv.org/abs/hep-th/0011215}{{\tt hep-th/0011215}}.

\bibitem{Henningson:2004dh}
M.~Henningson, ``{Self-dual strings in six dimensions: Anomalies, the
  ADE-classification, and the world-sheet WZW-model},'' {\em Commun. Math.
  Phys.} {\bf 257} (2005) 291--302,
  \href{http://arXiv.org/abs/hep-th/0405056}{{\tt hep-th/0405056}}.

\bibitem{Cordova:2015vwa}
C.~Cordova, T.~T. Dumitrescu, and X.~Yin, ``{Higher derivative terms, toroidal
  compactification, and Weyl anomalies in six-dimensional (2, 0) theories},''
  {\em JHEP} {\bf 10} (2019) 128, \href{http://arXiv.org/abs/1505.03850}{{\tt
  1505.03850}}.

\bibitem{Beem:2014kka}
C.~Beem, L.~Rastelli, and B.~C. van Rees, ``{$ \mathcal{W} $ symmetry in six
  dimensions},'' {\em JHEP} {\bf 05} (2015) 017,
  \href{http://arXiv.org/abs/1404.1079}{{\tt 1404.1079}}.

\bibitem{Beem:2015aoa}
C.~Beem, M.~Lemos, L.~Rastelli, and B.~C. van Rees, ``{The (2, 0)
  superconformal bootstrap},'' {\em Phys. Rev. D} {\bf 93} (2016), no.~2,
  025016, \href{http://arXiv.org/abs/1507.05637}{{\tt 1507.05637}}.

\bibitem{DelZotto:2014fia}
M.~Del~Zotto, J.~J. Heckman, D.~R. Morrison, and D.~S. Park, ``{6D SCFTs and
  Gravity},'' {\em JHEP} {\bf 06} (2015) 158,
  \href{http://arXiv.org/abs/1412.6526}{{\tt 1412.6526}}.

\bibitem{Morrison:2016nrt}
D.~R. Morrison and C.~Vafa, ``{F-theory and $ \mathcal{N} $ = 1 SCFTs in Four
  Dimensions},'' {\em JHEP} {\bf 08} (2016) 070,
  \href{http://arXiv.org/abs/1604.03560}{{\tt 1604.03560}}.

\bibitem{DelZotto:2018tcj}
M.~Del~Zotto and G.~Lockhart, ``{Universal Features of BPS Strings in
  Six-dimensional SCFTs},'' {\em JHEP} {\bf 08} (2018) 173,
  \href{http://arXiv.org/abs/1804.09694}{{\tt 1804.09694}}.

\bibitem{Bhardwaj:2018jgp}
L.~Bhardwaj, D.~R. Morrison, Y.~Tachikawa, and A.~Tomasiello, ``{The frozen
  phase of F-theory},'' {\em JHEP} {\bf 08} (2018) 138,
  \href{http://arXiv.org/abs/1805.09070}{{\tt 1805.09070}}.

\bibitem{DelZotto:2014hpa}
M.~Del~Zotto, J.~J. Heckman, A.~Tomasiello, and C.~Vafa, ``{6d Conformal
  Matter},'' {\em JHEP} {\bf 02} (2015) 054,
  \href{http://arXiv.org/abs/1407.6359}{{\tt 1407.6359}}.

\bibitem{Heckman:2016ssk}
J.~J. Heckman, T.~Rudelius, and A.~Tomasiello, ``{6D RG Flows and Nilpotent
  Hierarchies},'' {\em JHEP} {\bf 07} (2016) 082,
  \href{http://arXiv.org/abs/1601.04078}{{\tt 1601.04078}}.

\bibitem{Heckman:2018pqx}
J.~J. Heckman, T.~Rudelius, and A.~Tomasiello, ``{Fission, Fusion, and 6D RG
  Flows},'' {\em JHEP} {\bf 02} (2019) 167,
  \href{http://arXiv.org/abs/1807.10274}{{\tt 1807.10274}}.

\bibitem{Tachikawa:2015wka}
Y.~Tachikawa, ``{Frozen singularities in M and F theory},'' {\em JHEP} {\bf 06}
  (2016) 128, \href{http://arXiv.org/abs/1508.06679}{{\tt 1508.06679}}.

\bibitem{Bhardwaj:2015xxa}
L.~Bhardwaj, ``{Classification of 6d $ \mathcal{N}=\left(1,0\right) $ gauge
  theories},'' {\em JHEP} {\bf 11} (2015) 002,
  \href{http://arXiv.org/abs/1502.06594}{{\tt 1502.06594}}.

\bibitem{Argyres:2015ffa}
P.~Argyres, M.~Lotito, Y.~L\"u, and M.~Martone, ``{Geometric constraints on the
  space of $ \mathcal{N} $ = 2 SCFTs. Part I: physical constraints on relevant
  deformations},'' {\em JHEP} {\bf 02} (2018) 001,
  \href{http://arXiv.org/abs/1505.04814}{{\tt 1505.04814}}.

\bibitem{Xie:2015rpa}
D.~Xie and S.-T. Yau, ``{4d N=2 SCFT and singularity theory Part I:
  Classification},'' \href{http://arXiv.org/abs/1510.01324}{{\tt 1510.01324}}.

\bibitem{Argyres:2015gha}
P.~C. Argyres, M.~Lotito, Y.~L\"u, and M.~Martone, ``{Geometric constraints on
  the space of $ \mathcal{N} $ = 2 SCFTs. Part II: construction of special
  K\"ahler geometries and RG flows},'' {\em JHEP} {\bf 02} (2018) 002,
  \href{http://arXiv.org/abs/1601.00011}{{\tt 1601.00011}}.

\bibitem{Chen:2016bzh}
B.~Chen, D.~Xie, S.-T. Yau, S.~S.~T. Yau, and H.~Zuo, ``{4D $\mathcal{N} = 2$
  SCFT and singularity theory. Part II: complete intersection},'' {\em Adv.
  Theor. Math. Phys.} {\bf 21} (2017) 121--145,
  \href{http://arXiv.org/abs/1604.07843}{{\tt 1604.07843}}.

\bibitem{Argyres:2016xmc}
P.~Argyres, M.~Lotito, Y.~L\"u, and M.~Martone, ``{Geometric constraints on the
  space of $ \mathcal{N}$ = 2 SCFTs. Part III: enhanced Coulomb branches and
  central charges},'' {\em JHEP} {\bf 02} (2018) 003,
  \href{http://arXiv.org/abs/1609.04404}{{\tt 1609.04404}}.

\bibitem{Caorsi:2018zsq}
M.~Caorsi and S.~Cecotti, ``{Geometric classification of 4d $\mathcal{N}=2$
  SCFTs},'' {\em JHEP} {\bf 07} (2018) 138,
  \href{http://arXiv.org/abs/1801.04542}{{\tt 1801.04542}}.

\bibitem{Argyres:2020nrr}
P.~Argyres and M.~Martone, ``{Construction and classification of Coulomb branch
  geometries},'' \href{http://arXiv.org/abs/2003.04954}{{\tt 2003.04954}}.

\bibitem{Argyres:2020wmq}
P.~C. Argyres and M.~Martone, ``{Towards a classification of rank r$
  \mathcal{N} $ = 2 SCFTs. Part II. Special Kahler stratification of the
  Coulomb branch},'' {\em JHEP} {\bf 12} (2020) 022,
  \href{http://arXiv.org/abs/2007.00012}{{\tt 2007.00012}}.

\bibitem{Cecotti:2021ouq}
S.~Cecotti, M.~Del~Zotto, M.~Martone, and R.~Moscrop, ``{The Characteristic
  Dimension of Four-Dimensional ${\mathcal {N}}$~=~2 SCFTs},'' {\em Commun.
  Math. Phys.} {\bf 400} (2023), no.~1, 519--540,
  \href{http://arXiv.org/abs/2108.10884}{{\tt 2108.10884}}.

\bibitem{Closset:2021lhd}
C.~Closset and H.~Magureanu, ``{The $U$-plane of rank-one 4d $\mathcal{N}=2$ KK
  theories},'' {\em SciPost Phys.} {\bf 12} (2022) 065,
  \href{http://arXiv.org/abs/2107.03509}{{\tt 2107.03509}}.

\bibitem{Argyres:2022lah}
P.~C. Argyres and M.~Martone, ``{The rank 2 classification problem I: scale
  invariant geometries},'' \href{http://arXiv.org/abs/2209.09248}{{\tt
  2209.09248}}.

\bibitem{Argyres:2022puv}
P.~C. Argyres and M.~Martone, ``{The rank 2 classification problem II: mapping
  scale-invariant solutions to SCFTs},''
  \href{http://arXiv.org/abs/2209.09911}{{\tt 2209.09911}}.

\bibitem{Argyres:2022fwy}
P.~C. Argyres and M.~Martone, ``{The rank-2 classification problem III: curves
  with additional automorphisms},'' \href{http://arXiv.org/abs/2209.10555}{{\tt
  2209.10555}}.

\bibitem{Cecotti:2022uep}
S.~Cecotti, ``{Fuchsian ODEs as Seiberg dualities},''
  \href{http://arXiv.org/abs/2212.09370}{{\tt 2212.09370}}.

\bibitem{Beem:2013qxa}
C.~Beem, L.~Rastelli, and B.~C. van Rees, ``{The $\mathcal N=4$ Superconformal
  Bootstrap},'' {\em Phys. Rev. Lett.} {\bf 111} (2013) 071601,
  \href{http://arXiv.org/abs/1304.1803}{{\tt 1304.1803}}.

\bibitem{Beem:2014zpa}
C.~Beem, M.~Lemos, P.~Liendo, L.~Rastelli, and B.~C. van Rees, ``{The $
  \mathcal{N}=2 $ superconformal bootstrap},'' {\em JHEP} {\bf 03} (2016) 183,
  \href{http://arXiv.org/abs/1412.7541}{{\tt 1412.7541}}.

\bibitem{Bonetti:2018fqz}
F.~Bonetti, C.~Meneghelli, and L.~Rastelli, ``{VOAs labelled by complex
  reflection groups and 4d SCFTs},'' {\em JHEP} {\bf 05} (2019) 155,
  \href{http://arXiv.org/abs/1810.03612}{{\tt 1810.03612}}.

\bibitem{Kaidi:2022sng}
J.~Kaidi, M.~Martone, L.~Rastelli, and M.~Weaver, ``{Needles in a haystack. An
  algorithmic approach to the classification of 4d $ \mathcal{N} $ = 2
  SCFTs},'' {\em JHEP} {\bf 03} (2022) 210,
  \href{http://arXiv.org/abs/2202.06959}{{\tt 2202.06959}}.

\bibitem{Chacaltana:2010ks}
O.~Chacaltana and J.~Distler, ``{Tinkertoys for Gaiotto Duality},'' {\em JHEP}
  {\bf 11} (2010) 099, \href{http://arXiv.org/abs/1008.5203}{{\tt 1008.5203}}.

\bibitem{Chacaltana:2011ze}
O.~Chacaltana and J.~Distler, ``{Tinkertoys for the $D_N$ series},'' {\em JHEP}
  {\bf 02} (2013) 110, \href{http://arXiv.org/abs/1106.5410}{{\tt 1106.5410}}.

\bibitem{Chacaltana:2012zy}
O.~Chacaltana, J.~Distler, and Y.~Tachikawa, ``{Nilpotent orbits and
  codimension-two defects of 6d N=(2,0) theories},'' {\em Int. J. Mod. Phys. A}
  {\bf 28} (2013) 1340006, \href{http://arXiv.org/abs/1203.2930}{{\tt
  1203.2930}}.

\bibitem{Chacaltana:2012ch}
O.~Chacaltana, J.~Distler, and Y.~Tachikawa, ``{Gaiotto duality for the twisted
  A$_{2N-1}$ series},'' {\em JHEP} {\bf 05} (2015) 075,
  \href{http://arXiv.org/abs/1212.3952}{{\tt 1212.3952}}.

\bibitem{Chacaltana:2013oka}
O.~Chacaltana, J.~Distler, and A.~Trimm, ``{Tinkertoys for the Twisted
  D-Series},'' {\em JHEP} {\bf 04} (2015) 173,
  \href{http://arXiv.org/abs/1309.2299}{{\tt 1309.2299}}.

\bibitem{Chacaltana:2015bna}
O.~Chacaltana, J.~Distler, and A.~Trimm, ``{Tinkertoys for the Twisted $E_6$
  Theory},'' \href{http://arXiv.org/abs/1501.00357}{{\tt 1501.00357}}.

\bibitem{Chacaltana:2017boe}
O.~Chacaltana, J.~Distler, A.~Trimm, and Y.~Zhu, ``{Tinkertoys for the E$_{7}$
  theory},'' {\em JHEP} {\bf 05} (2018) 031,
  \href{http://arXiv.org/abs/1704.07890}{{\tt 1704.07890}}.

\bibitem{Chacaltana:2018vhp}
O.~Chacaltana, J.~Distler, A.~Trimm, and Y.~Zhu, ``{Tinkertoys for the $E_8$
  Theory},'' \href{http://arXiv.org/abs/1802.09626}{{\tt 1802.09626}}.

\bibitem{Tachikawa:2015bga}
Y.~Tachikawa, ``{A review of the $T_N$ theory and its cousins},'' {\em PTEP}
  {\bf 2015} (2015), no.~11, 11B102,
  \href{http://arXiv.org/abs/1504.01481}{{\tt 1504.01481}}.

\bibitem{Argyres:1995jj}
P.~C. Argyres and M.~R. Douglas, ``{New phenomena in SU(3) supersymmetric gauge
  theory},'' {\em Nucl. Phys. B} {\bf 448} (1995) 93--126,
  \href{http://arXiv.org/abs/hep-th/9505062}{{\tt hep-th/9505062}}.

\bibitem{Eguchi:1996ds}
T.~Eguchi and K.~Hori, ``{N=2 superconformal field theories in four-dimensions
  and A-D-E classification},'' in {\em {Conference on the Mathematical Beauty
  of Physics (In Memory of C. Itzykson)}}, pp.~67--82.
\newblock 7, 1996.
\newblock \href{http://arXiv.org/abs/hep-th/9607125}{{\tt hep-th/9607125}}.

\bibitem{Cecotti:2010fi}
S.~Cecotti, A.~Neitzke, and C.~Vafa, ``{R-Twisting and 4d/2d
  Correspondences},'' \href{http://arXiv.org/abs/1006.3435}{{\tt 1006.3435}}.

\bibitem{Cecotti:2011gu}
S.~Cecotti and M.~Del~Zotto, ``{On Arnold's 14 `exceptional' N=2 superconformal
  gauge theories},'' {\em JHEP} {\bf 10} (2011) 099,
  \href{http://arXiv.org/abs/1107.5747}{{\tt 1107.5747}}.

\bibitem{Xie:2012hs}
D.~Xie, ``{General Argyres-Douglas Theory},'' {\em JHEP} {\bf 01} (2013) 100,
  \href{http://arXiv.org/abs/1204.2270}{{\tt 1204.2270}}.

\bibitem{Closset:2020afy}
C.~Closset, S.~Giacomelli, S.~Schafer-Nameki, and Y.-N. Wang, ``{5d and 4d
  SCFTs: Canonical Singularities, Trinions and S-Dualities},''
  \href{http://arXiv.org/abs/2012.12827}{{\tt 2012.12827}}.

\bibitem{Cecotti:2012jx}
S.~Cecotti and M.~Del~Zotto, ``{Infinitely many N=2 SCFT with ADE flavor
  symmetry},'' {\em JHEP} {\bf 01} (2013) 191,
  \href{http://arXiv.org/abs/1210.2886}{{\tt 1210.2886}}.

\bibitem{Cecotti:2013lda}
S.~Cecotti, M.~Del~Zotto, and S.~Giacomelli, ``{More on the N=2 superconformal
  systems of type $D_p(G)$},'' {\em JHEP} {\bf 04} (2013) 153,
  \href{http://arXiv.org/abs/1303.3149}{{\tt 1303.3149}}.

\bibitem{giaco1}
S.~Giacomelli, ``{RG flows with supersymmetry enhancement and geometric
  engineering},'' {\em JHEP} {\bf 06} (2018) 156,
\href{http://arXiv.org/abs/1710.06469}{{\tt 1710.06469}}.

\bibitem{Wang:2015mra}
Y.~Wang and D.~Xie, ``{Classification of Argyres-Douglas theories from M5
  branes},'' {\em Phys. Rev. D} {\bf 94} (2016), no.~6, 065012,
  \href{http://arXiv.org/abs/1509.00847}{{\tt 1509.00847}}.

\bibitem{Kang:2021lic}
M.~J. Kang, C.~Lawrie, and J.~Song, ``{Infinitely many 4D N=2 SCFTs with a=c
  and beyond},'' {\em Phys. Rev. D} {\bf 104} (2021), no.~10, 105005,
  \href{http://arXiv.org/abs/2106.12579}{{\tt 2106.12579}}.

\bibitem{Kang:2021ccs}
M.~J. Kang, C.~Lawrie, K.-H. Lee, and J.~Song, ``{Infinitely many 4D N=1 SCFTs
  with a=c},'' {\em Phys. Rev. D} {\bf 105} (2022), no.~12, 126006,
  \href{http://arXiv.org/abs/2111.12092}{{\tt 2111.12092}}.

\bibitem{Acharya:2023bth}
B.~S. Acharya, M.~Del~Zotto, J.~J. Heckman, M.~Hubner, and E.~Torres,
  ``{Junctions, Edge Modes, and $G_2$-Holonomy Orbifolds},''
  \href{http://arXiv.org/abs/2304.03300}{{\tt 2304.03300}}.

\bibitem{DelZotto:2017pti}
M.~Del~Zotto, J.~J. Heckman, and D.~R. Morrison, ``{6D SCFTs and Phases of 5D
  Theories},'' {\em JHEP} {\bf 09} (2017) 147,
  \href{http://arXiv.org/abs/1703.02981}{{\tt 1703.02981}}.

\bibitem{Jefferson_2018}
P.~Jefferson, S.~Katz, H.-C. Kim, and C.~Vafa, ``On geometric classification of
  5d scfts,'' {\em Journal of High Energy Physics} {\bf 2018} (Apr, 2018).

\bibitem{Bhardwaj:2018yhy}
L.~Bhardwaj and P.~Jefferson, ``{Classifying $5d$ SCFTs via $6d$ SCFTs: Rank
  one},'' {\em JHEP} {\bf 07} (2019) 178,
  \href{http://arXiv.org/abs/1809.01650}{{\tt 1809.01650}}. [Addendum: JHEP 01,
  153 (2020)].

\bibitem{Bhardwaj:2019xeg}
L.~Bhardwaj, ``{Do all 5d SCFTs descend from 6d SCFTs?},'' {\em JHEP} {\bf 04}
  (2021) 085, \href{http://arXiv.org/abs/1912.00025}{{\tt 1912.00025}}.

\bibitem{Xie_2017}
D.~Xie and S.-T. Yau, ``Three dimensional canonical singularity and five
  dimensional n $ \mathcal{N}$ = 1 scft,'' {\em Journal of High Energy Physics}
  {\bf 2017} (Jun, 2017).

\bibitem{Apruzzi_2020}
F.~Apruzzi, C.~Lawrie, L.~Lin, S.~Schäfer-Nameki, and Y.-N. Wang, ``Fibers add
  flavor. part {II}. 5d {SCFTs}, gauge theories, and dualities,'' {\em Journal
  of High Energy Physics} {\bf 2020} (mar, 2020).

\bibitem{Bourget:2023wlb}
A.~Bourget, A.~Collinucci, and S.~Schafer-Nameki, ``{Generalized Toric
  Polygons, T-branes, and 5d SCFTs},''
  \href{http://arXiv.org/abs/2301.05239}{{\tt 2301.05239}}.

\bibitem{Collinucci:2021ofd}
A.~Collinucci, M.~De~Marco, A.~Sangiovanni, and R.~Valandro, ``{Higgs branches
  of 5d rank-zero theories from geometry},'' {\em JHEP} {\bf 10} (2021),
  no.~18, 018, \href{http://arXiv.org/abs/2105.12177}{{\tt 2105.12177}}.

\bibitem{Collinucci:2021wty}
A.~Collinucci, A.~Sangiovanni, and R.~Valandro, ``{Genus zero Gopakumar-Vafa
  invariants from open strings},'' {\em JHEP} {\bf 09} (2021) 059,
  \href{http://arXiv.org/abs/2104.14493}{{\tt 2104.14493}}.

\bibitem{DeMarco:2021try}
M.~De~Marco and A.~Sangiovanni, ``{Higgs Branches of rank-0 5d theories from
  M-theory on (A$_{j}$, A$_{l}$) and (A$_{k}$, D$_{n}$) singularities},'' {\em
  JHEP} {\bf 03} (2022) 099, \href{http://arXiv.org/abs/2111.05875}{{\tt
  2111.05875}}.

\bibitem{DeMarco:2022dgh}
M.~De~Marco, A.~Sangiovanni, and R.~Valandro, ``{5d Higgs branches from
  M-theory on quasi-homogeneous cDV threefold singularities},'' {\em JHEP} {\bf
  10} (2022) 124, \href{http://arXiv.org/abs/2205.01125}{{\tt 2205.01125}}.

\bibitem{Collinucci:2022rii}
A.~Collinucci, M.~De~Marco, A.~Sangiovanni, and R.~Valandro, ``{Flops of any
  length, Gopakumar-Vafa invariants and 5d Higgs branches},'' {\em JHEP} {\bf
  08} (2022) 292, \href{http://arXiv.org/abs/2204.10366}{{\tt 2204.10366}}.

\bibitem{Acharya:2021jsp}
B.~Acharya, N.~Lambert, M.~Najjar, E.~E. Svanes, and J.~Tian, ``{Gauging
  discrete symmetries of T$_{N}$-theories in five dimensions},'' {\em JHEP}
  {\bf 04} (2022) 114, \href{http://arXiv.org/abs/2110.14441}{{\tt
  2110.14441}}.

\bibitem{Tian:2021cif}
J.~Tian and Y.-N. Wang, ``{5D and 6D SCFTs from $\mathbb{C}^3$ orbifolds},''
  {\em SciPost Phys.} {\bf 12} (2022), no.~4, 127,
  \href{http://arXiv.org/abs/2110.15129}{{\tt 2110.15129}}.

\bibitem{DelZotto:2022fnw}
M.~Del~Zotto, J.~J. Heckman, S.~N. Meynet, R.~Moscrop, and H.~Y. Zhang,
  ``{Higher symmetries of 5D orbifold SCFTs},'' {\em Phys. Rev. D} {\bf 106}
  (2022), no.~4, 046010, \href{http://arXiv.org/abs/2201.08372}{{\tt
  2201.08372}}.

\bibitem{KatzVafa}
S.~H. Katz and C.~Vafa, ``{Matter from geometry},'' {\em Nucl. Phys. B} {\bf
  497} (1997) 146--154, \href{http://arXiv.org/abs/hep-th/9606086}{{\tt
  hep-th/9606086}}.

\bibitem{Closset:2020scj}
C.~Closset, S.~Schafer-Nameki, and Y.-N. Wang, ``{Coulomb and Higgs Branches
  from Canonical Singularities: Part 0},''
  \href{http://arXiv.org/abs/2007.15600}{{\tt 2007.15600}}.

\bibitem{Closset:2021lwy}
C.~Closset, S.~Sch\"afer-Nameki, and Y.-N. Wang, ``{Coulomb and Higgs branches
  from canonical singularities. Part I. Hypersurfaces with smooth Calabi-Yau
  resolutions},'' {\em JHEP} {\bf 04} (2022) 061,
  \href{http://arXiv.org/abs/2111.13564}{{\tt 2111.13564}}.

\bibitem{Benini:2009gi}
F.~Benini, S.~Benvenuti, and Y.~Tachikawa, ``{Webs of five-branes and N=2
  superconformal field theories},'' {\em JHEP} {\bf 09} (2009) 052,
  \href{http://arXiv.org/abs/0906.0359}{{\tt 0906.0359}}.

\bibitem{Hayashi:2019fsa}
H.~Hayashi, P.~Jefferson, H.-C. Kim, K.~Ohmori, and C.~Vafa, ``{SCFTs,
  Holography, and Topological Strings},''
  \href{http://arXiv.org/abs/1905.00116}{{\tt 1905.00116}}.

\bibitem{Hollowood:2003cv}
T.~J. Hollowood, A.~Iqbal, and C.~Vafa, ``{Matrix models, geometric engineering
  and elliptic genera},'' {\em JHEP} {\bf 03} (2008) 069,
  \href{http://arXiv.org/abs/hep-th/0310272}{{\tt hep-th/0310272}}.

\bibitem{Intriligator_1997}
K.~Intriligator, D.~R. Morrison, and N.~Seiberg, ``Five-dimensional
  supersymmetric gauge theories and degenerations of calabi-yau spaces,'' {\em
  Nuclear Physics B} {\bf 497} (Jul, 1997) 56–100.

\bibitem{Closset_2019}
C.~Closset, M.~Del~Zotto, and V.~Saxena, ``Five-dimensional scfts and gauge
  theory phases: an m-theory/type iia perspective,'' {\em SciPost Physics} {\bf
  6} (May, 2019).

\bibitem{Tachikawa}
Y.~Tachikawa, ``{Instanton operators and symmetry enhancement in 5d
  supersymmetric gauge theories},'' {\em PTEP} {\bf 2015} (2015), no.~4,
  043B06, \href{http://arXiv.org/abs/1501.01031}{{\tt 1501.01031}}.

\bibitem{Yonekura}
K.~Yonekura, ``{Instanton operators and symmetry enhancement in 5d
  supersymmetric quiver gauge theories},'' {\em JHEP} {\bf 07} (2015) 167,
  \href{http://arXiv.org/abs/1505.04743}{{\tt 1505.04743}}.

\bibitem{Zafrir:2015uaa}
G.~Zafrir, ``{Instanton operators and symmetry enhancement in 5d supersymmetric
  USp, SO and exceptional gauge theories},'' {\em JHEP} {\bf 07} (2015) 087,
  \href{http://arXiv.org/abs/1503.08136}{{\tt 1503.08136}}.

\bibitem{Closset:2018bjz}
C.~Closset, M.~Del~Zotto, and V.~Saxena, ``{Five-dimensional SCFTs and gauge
  theory phases: an M-theory/type IIA perspective},'' {\em SciPost Phys.} {\bf
  6} (2019), no.~5, 052, \href{http://arXiv.org/abs/1812.10451}{{\tt
  1812.10451}}.

\bibitem{Witten:1996qb}
E.~Witten, ``{Phase transitions in M theory and F theory},'' {\em Nucl. Phys.
  B} {\bf 471} (1996) 195--216, \href{http://arXiv.org/abs/hep-th/9603150}{{\tt
  hep-th/9603150}}.

\bibitem{Aharony:1997bh}
O.~Aharony, A.~Hanany, and B.~Kol, ``{Webs of (p,q) five-branes,
  five-dimensional field theories and grid diagrams},'' {\em JHEP} {\bf 01}
  (1998) 002, \href{http://arXiv.org/abs/hep-th/9710116}{{\tt hep-th/9710116}}.

\bibitem{Gaiotto:2015una}
D.~Gaiotto and H.-C. Kim, ``{Duality walls and defects in 5d $ \mathcal{N}=1 $
  theories},'' {\em JHEP} {\bf 01} (2017) 019,
  \href{http://arXiv.org/abs/1506.03871}{{\tt 1506.03871}}.

\bibitem{Bhardwaj:2019ngx}
L.~Bhardwaj, ``{Dualities of 5d gauge theories from S-duality},'' {\em JHEP}
  {\bf 07} (2020) 012, \href{http://arXiv.org/abs/1909.05250}{{\tt
  1909.05250}}.

\bibitem{Bhardwaj:2020ruf}
L.~Bhardwaj, ``{Flavor symmetry of 5d SCFTs. Part I. General setup},'' {\em
  JHEP} {\bf 09} (2021) 186, \href{http://arXiv.org/abs/2010.13230}{{\tt
  2010.13230}}.

\bibitem{Bhardwaj:2020avz}
L.~Bhardwaj, ``{Flavor symmetry of 5$d$ SCFTs. Part II. Applications},'' {\em
  JHEP} {\bf 04} (2021) 221, \href{http://arXiv.org/abs/2010.13235}{{\tt
  2010.13235}}.

\bibitem{DuVal}
P.~Du~Val, ``On isolated singularities of surfaces which do not affect the
  conditions of adjunction. {I},'' {\em Proc. Camb. Philos. Soc.} {\bf 30}
  (1934) 453--459.

\bibitem{reidminimal1983}
M.~Reid, ``Minimal models of canonical 3-folds,'' 1983.

\bibitem{reid1983minimal}
M.~Reid {\em et al.}, ``Minimal models of canonical 3-folds,'' {\em Adv. Stud.
  Pure Math} {\bf 1} (1983) 131--180.

\bibitem{eisenbud2000geometry}
D.~Eisenbud, J.~Harris, J.~Harris, and Springer-Verlag, {\em The Geometry of
  Schemes}.
\newblock Graduate Texts in Mathematics. Springer, 2000.

\bibitem{beauville_1996}
A.~Beauville, {\em Complex Algebraic Surfaces}.
\newblock London Mathematical Society Student Texts. Cambridge University
  Press, 2~ed., 1996.

\bibitem{Tachikawa:2015mha}
Y.~Tachikawa, ``{Instanton operators and symmetry enhancement in 5d
  supersymmetric gauge theories},'' {\em PTEP} {\bf 2015} (2015), no.~4,
  043B06, \href{http://arXiv.org/abs/1501.01031}{{\tt 1501.01031}}.

\bibitem{conformalmatter}
M.~Del~Zotto, J.~J. Heckman, A.~Tomasiello, and C.~Vafa, ``{6d Conformal
  Matter},'' {\em JHEP} {\bf 02} (2015) 054,
  \href{http://arXiv.org/abs/1407.6359}{{\tt 1407.6359}}.

\bibitem{Bergman:2014kza}
O.~Bergman and G.~Zafrir, ``{Lifting 4d dualities to 5d},'' {\em JHEP} {\bf 04}
  (2015) 141, \href{http://arXiv.org/abs/1410.2806}{{\tt 1410.2806}}.

\bibitem{Hayashi:2014hfa}
H.~Hayashi, Y.~Tachikawa, and K.~Yonekura, ``{Mass-deformed T$_{N}$ as a linear
  quiver},'' {\em JHEP} {\bf 02} (2015) 089,
  \href{http://arXiv.org/abs/1410.6868}{{\tt 1410.6868}}.

\bibitem{Gukov_2000}
S.~Gukov, C.~Vafa, and E.~Witten, ``Cft’s from calabi–yau four-folds,''
  {\em Nuclear Physics B} {\bf 584} (Sep, 2000) 69–108.

\bibitem{Candelas:1989js}
P.~Candelas and X.~C. de~la Ossa, ``{Comments on Conifolds},'' {\em Nucl. Phys.
  B} {\bf 342} (1990) 246--268.

\bibitem{Collinucci:2020jqd}
A.~Collinucci and R.~Valandro, ``{The role of U(1)'s in 5d theories, Higgs
  branches, and geometry},'' \href{http://arXiv.org/abs/2006.15464}{{\tt
  2006.15464}}.

\bibitem{reidyoungpersonsing}
M.~Reid, ``Young person guide to canonical singularities,'' 1987.

\bibitem{COSTELLAZIONI}
M.~Graffeo, ``Moduli spaces of $\mathbb{Z}/k\mathbb{Z}$-constellations over
  $\mathbb{A}^2$,'' 2022.

\bibitem{Reid1983DecompositionOT}
M.~Reid, ``Decomposition of toric morphisms,'' in {\em Arithmetic and
  Geometry}.
\newblock 1983.

\end{thebibliography}
\providecommand{\href}[2]{#2}

\end{document}